\documentclass[10pt]{article}      
 \usepackage[top=0.8in, bottom=0.7in, left=0.7in, right=0.8in]{geometry}
 \usepackage{appendix}
\usepackage{graphicx}
\usepackage{authblk}
\usepackage[numbers,sort&compress]{natbib}
\usepackage{natbib}
\usepackage{pdflscape}
\usepackage{url}
\usepackage{bm}
\usepackage{color}
\usepackage{amssymb}\usepackage{amsmath,amsthm}
\usepackage{amsmath} \usepackage{multirow}
\usepackage{mathrsfs}

\renewcommand{\theequation}{\thesection.\arabic{equation}}

\usepackage{caption}
\usepackage{subfig}
\graphicspath{ {./images/} }

\numberwithin{equation}{section} 
\usepackage{amscd,epstopdf}
\usepackage{graphicx}\usepackage{multirow}
\usepackage{float}
\usepackage{subcaption}

\usepackage{afterpage} \usepackage{rotating}
\usepackage{enumerate}
 
 \newtheorem{thm}{Theorem}[section]
 
 \newtheorem{lem}{Lemma}[section]
 
 \newtheorem{corollary}{Corollary}[section]

 \newtheorem{definition}{Definition}[section]
 \usepackage{setspace}
\usepackage{indentfirst} 
\usepackage{graphicx}
\usepackage{bibentry}

\begin{document}

\title{SVD-based Causal Emergence for Gaussian Iterative Systems}

\author[1]{Kaiwei Liu}
\author[1]{Linli Pan}
\author[1]{Zhipeng Wang}
\author[1]{Mingzhe Yang}
\author[2]{Bing Yuan} 
\author[1,2]{Jiang Zhang  \thanks{Email: zhangjiang@bnu.edu.cn}   } 
 
\affil[1]{ School of Systems Science, Beijing Normal University, 100875, Beijing, China}
\affil[2]{Swarma Research, 102300, Beijing, China
}
       
 \renewcommand\Authands{ and }
 \date{}
 \maketitle
 
\begin{abstract}
Causal emergence (CE) based on effective information (EI) demonstrates that macro-states can exhibit stronger causal effects than micro-states in dynamics. However, the identification of CE and the maximization of EI both rely on coarse-graining strategies, which is a key challenge. A recently proposed CE framework based on approximate dynamical reversibility, utilizing singular value decomposition (SVD), is independent of coarse-graining. Still, it is limited to transition probability matrices (TPM) in discrete states. To address this, this article proposes a novel CE quantification framework for Gaussian iterative systems (GIS), based on approximate dynamical reversibility derived from the SVD of inverse covariance matrices in forward and backward dynamics. The positive correlation between SVD-based and EI-based CE, along with the equivalence condition, is given analytically. After that, we provide precise coarse-graining strategies directly from singular value spectra and orthogonal matrices. This new framework can be applied to any dynamical system with continuous states and Gaussian noise, such as auto-regressive growth models, Markov-Gaussian systems, and even SIR modeling using neural networks (NN). Numerical simulations on typical cases validate our theory and offer a new approach to studying the CE phenomenon, emphasizing noise and covariance over dynamical functions in both known models and machine learning.
\end{abstract}


\section{Introduction}

Complex systems with dynamics are ubiquitous in the world around us, such as ecosystems \cite{Joergensen2000}, organisms \cite{Wicks2007, Hartman2006}, brains \cite{Jingnan2021, Sporns2021, Varley2023,wang2025causal}, and cells \cite{Zhao2023, Dong2016}. The interrelations between different dimensions and the accumulation of randomness result in entropy production and disorder, thereby complicating the analysis of dynamics of microscopic composition, such as individuals in society, cells in the human body, and atoms in matter \cite{Ma2024,Bennett2000}. Many scholars contend that complex systems conceal profound patterns and regularities within their apparent disorder. For example, it is difficult to accurately describe the motion trajectory of microscopic molecules in gases or liquids, but the changes of both can be clearly described macroscopically using temperature and pressure. So they try to derive dynamical models of the systems at the macro-level by coarse-graining the dynamical systems, discovering that the strength of the causal effect \cite{albrecht1996causality,schreiber2000measuring,rosenblum2001detecting,harnack2017topological} for macro-states can surpass those of micro-states, a phenomenon called causal emergence (CE) \cite{tononi2003measuring, Hoel2013,Hoel2017,Rosas2020,Barnett2023,Yuan2024}. 

The typical quantitative framework for CE proposed by Hoel et al. \cite{Hoel2013, Hoel2016, Hoel2017} is based on effective information (EI) \cite{tononi2003measuring}, demonstrating that coarse-grained macro-states may have larger EI than micro-states, especially in Boolean networks \cite{bornholdt2008boolean,Hoel2013}, cellular automatons \cite{weisstein2002cellular,Yang2023}, even in Gaussian iterative systems (GIS) \cite{Chvykov2020,Liu2024}. With theoretical support, the study in Neural Information Squeezer plus (NIS+) \cite{Zhang2022,Yang2024} based on Hoel's EI framework utilizes neural networks (NN) to maximize EI and identify CE in data of the Boids model \cite{Reynolds1987} or fRMI \cite{smith2011network}. However, the dependence on coarse-graining strategies during maximizing EI within Hoel's framework presents challenges.
Because when facing practical problems, we need to find the optimal coarse-graining strategy that maximizes EI, which is difficult in high-dimensional complex systems.


Some research is independent of predefined coarse-graining strategies. A notable example is Rosas' framework \cite{Rosas2020}, which is based on integrated information decomposition theory \cite{Williams2010, Luppi2021}, and CE is quantified by calculating the synergic information ($Syn$) across micro-variables. However, this approach is challenging to implement due to the combinatorial explosion that occurs during the precise calculation of $Syn$. Although Rosas proposed approximate methods \cite{Song2020,Kaplanis2023,mcsharry2024learning} with results on real-world datasets, such as ECoG \cite{chao2010long} or MEG \cite{muthukumaraswamy20181}, these methods require predefined macro-state variables and representation learning, which also depends on the coarse-graining strategy. In addition, Barnett and Seth proposed a framework \cite{Barnett2023} for quantifying CE via dynamical independence. Yet, this framework has only been applied to linear systems and also requires predefined coarse-graining strategies \cite{byrne1965attraction}. Both methods utilize mutual information, making their results dependent on the distribution of variable states and unable to reflect causality.


To find a more essential quantification of CE, Zhang et al. recently proposed a CE quantification framework based on singular value decomposition (SVD) \cite{Zhang2025} for probability transition matrices (TPM), which is independent of the optimization of coarse-graining strategies and can simplify max EI calculation using SVD. Besides, the framework introduces the new concept of approximate dynamical reversibility derived from SVD to quantify the strength of the causal effects, which describes the proximity of TPM to a permutation matrix. SVD-based CE can be quantified as the potential maximal efficiency increase of approximate dynamical reversibility. This work also finds that the essence of CE lies in redundancy, represented by irreversible and correlated information pathways. What's more, this article demonstrates a strong correlation between the SVD-based and maximized EI-based CE. However, this method only applies to discrete Markov chains \cite{zhang2019spectral} or complex networks \cite{lee2019review, Liu2023} and lacks an analytical and exact relation between the two quantifications of CE. More research is needed to extend this framework to continuous systems.

Based on the above theory, the motivation for this paper is to derive a new theory of CE for GIS based on dynamical reversibility and SVD, and prove that this theory is approximately equivalent to Hoel's CE theory based on maximizing EI. To illustrate our theoretical framework, after briefly introducing linear GIS (Section \ref{Fundamentaltheories}), we first introduce the concepts of causality and approximate dynamical reversibility in Sections \ref{Causalitycan Reversibility} and \ref{sce:causalemergence}, respectively. We then present the SVD-based CE framework and compare its equivalence to the EI-based CE (Section \ref{sce:causalemergence}), as well as a coarse-graining strategy derived from SVD. Finally, we use 3 examples to illustrate how to use SVD to analyze the causal emergence properties of a dynamical system (Section \ref{results}).

\section{Theories}
\subsection{Linear Gaussian iterative systems}\label{Fundamentaltheories}
Many real-world complex systems, whether physical, biological, economic, or social, exhibit dynamical behavior constructed by the interplay of deterministic laws and random disturbances as
\begin{equation}
\label{micro-dynamics normal}
x_{t+1}=f(x_t)+\eta_t, \eta_t\sim\mathcal{N}(0,\Sigma), 
\end{equation}
$x_t\in\mathcal{R}^n, \Sigma\in\mathcal{R}^{n\times n}$, which we called Gaussian iterative systems (GIS) \cite{Dunsmuir1976,arnold1995random}. $f(x_t): \mathcal{R}^n\to \mathcal{R}^n$ captures deterministic dynamics, $\eta_t$ denotes random noise. Like Logistic map \cite{jordan1995logistic}, Hénon map \cite{benedicks1991dynamics}, including some continuous dynamic systems (such as Kuramoto \cite{acebron2005kuramoto}, SIR \cite{Satsuma2004}, Lotka-Volterra \cite{wangersky1978lotka}), can be approximated in this form. 

For nonlinear functions, we usually linearize them to analyze their local properties. For any function $f(x_t):\mathcal{R}^n\to\mathcal{R}^n$, it can be locally linearized by Taylor expansion as
\begin{equation}
\label{micro-dynamics}
x_{t+1}=a_0+Ax_t+\eta_t, \eta_t\sim\mathcal{N}(0,\Sigma), 
\end{equation}
which we called linear GIS, where $a_0,x_t\in\mathcal{R}^n$, and $A,\Sigma\in\mathcal{R}^{n\times n}$.  $A=\nabla f(x_t)$ corresponds to the Jacobian matrix around $x_t$. This equation can also be written in the form of a transition probability as
\begin{equation}
\label{micro-dynamics-TPM}
p(x_{t+1}|x_t)=\mathcal{N}(Ax_{t}+a_0,\Sigma).
\end{equation}
This is the physical model we mainly analyze in this article, several one-dimensional cases of which are shown in Fig.\ref{fig:distribution}.

\subsection{Causality and reversibility for linear systems}\label{Causalitycan Reversibility}
\subsubsection{Causality}
\label{Causality}
\begin{figure}[h]
    \centering
\includegraphics[width=0.8\textwidth,trim=8 430 330 10, clip]
{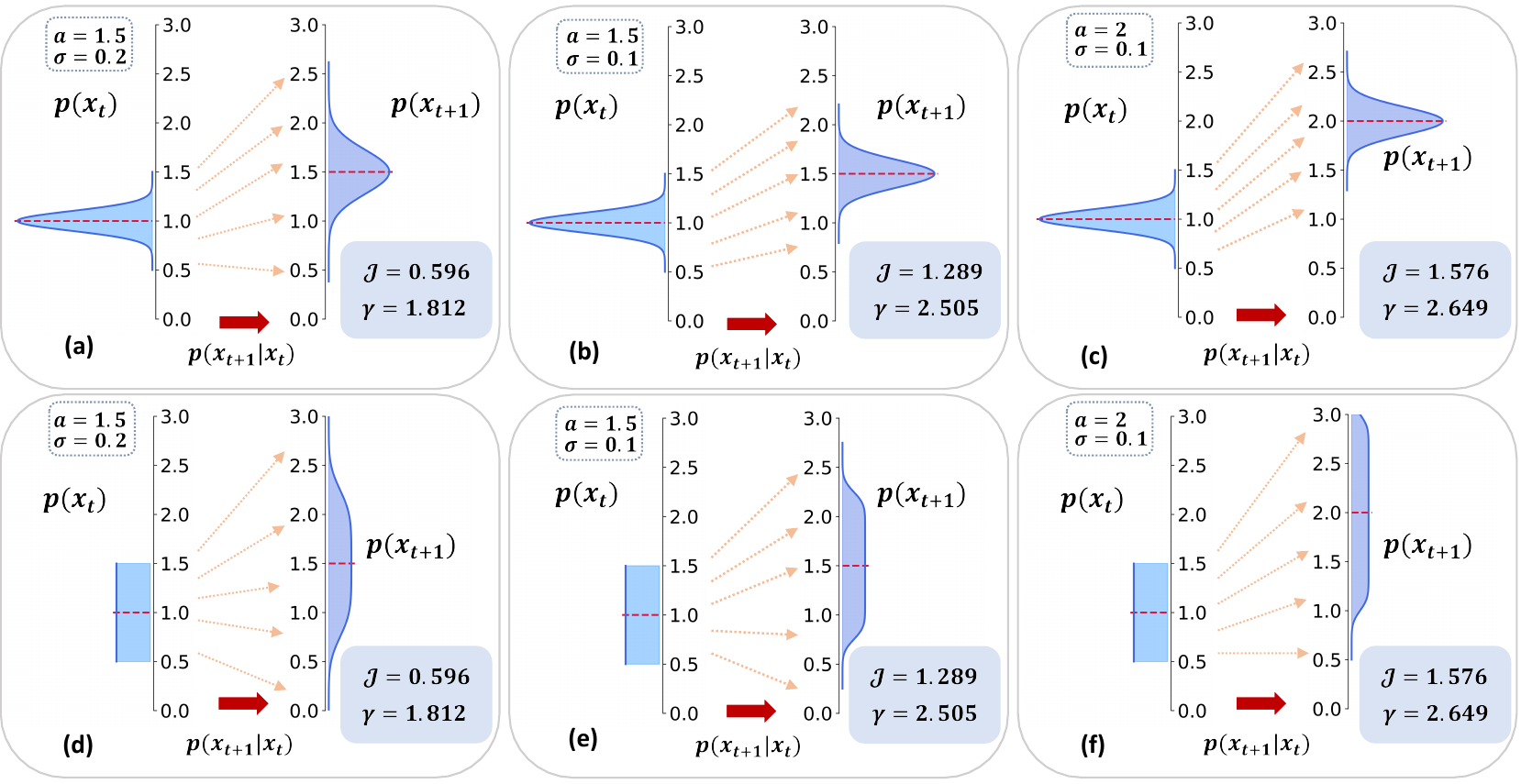}
    \caption{For 1-dimensional Gaussian information system $x_{t+1}=ax_t+\eta_t, \eta_t\sim\mathcal{N}(0,\sigma^2)$, we examine the dynamics under three different parameter settings and two distinct intervention distributions of $p(x_t)$, showing the resulting distributions $p(x_{t+1})$ and the corresponding measures of dimension averaged causal effect ($\mathcal{J}$) and approximate reversible information ($\gamma\equiv\gamma_1$).
    For any given probability density $p(x_t)$ at time $t$, the conditional probability $p(x_{t+1}|x_t)=\mathcal{N}(ax_t, \Sigma)$ corresponding to the forward dynamics gives the probability density $p(x_{t+1})$ at the next time step $t+1$. In (a) and (d), when we stipulate $a=1.5,\sigma=0.2$, the dimensionally averaged EI $\mathcal{J}=0.596$ and the dimensional average reversible information $\gamma=1.812$ are satisfied. In (b) and (e), we decrease $\sigma$ to $\sigma=0.1$. Due to the increase in determinism, the dimensionally averaged EI increases to $\mathcal{J}=1.289$, and the dimensional average reversible information increases to $\gamma=2.505$, indicating that the approximate reversibility and causal effect strength are increased simultaneously. Then increase $a$ in (c) (f) to $a=2$. Due to the increase in non-degeneracy, the dimensionally averaged EI increases to $\mathcal{J}=1.576$, and the dimensional average reversible information increases to $\gamma=2.649$. From these simple cases, we can initially obtain three pieces of information: (1) $\mathcal{J}$ and $\gamma$ are only related to $p(x_{t+1}|x_t)$, and independent on the probability density $p(x_t)$ of $x_t$ can be defined arbitrarily; (2) $\mathcal{J}$ and $\gamma$ are both positively correlated with determinism and non-degeneracy, and increase synchronously with the increase of $a$ and decrease synchronously with the increase of $\sigma$; (3) $\mathcal{J}$ and $\gamma$ are positively correlated.
}
\label{fig:distribution}
\end{figure}

Dynamics expressed in equations Eq.(\ref{micro-dynamics}) or Eq.(\ref{micro-dynamics-TPM}) can be considered a causal model, the state $x_t$ at time $t$ causes changes in state $x_{t+1}$ at the subsequent time $t+1$. Thus, $x_t$ is the cause variable and $x_{t+1}$ is the effect variable. The causality discussed in this paper refers specifically to the strength of the causal effect of cause variables on effect variables when given their causal relationship, which is the degree of certainty that when the cause changes, the effect will also change accordingly.

Fig.\ref{fig:distribution} shows several examples of 1-dimensional linear GIS in Eq.(\ref{micro-dynamics}), demonstrating the process of causality change from left to right. Specifically, when the dynamic parameter $a$ in Eq.(\ref{micro-dynamics}) is larger or the noise standard deviation $\sigma$ is smaller, it is generally believed that the strength of the causal effect will be greater. Moreover, the strength of the causal effect should be independent of the input state distribution $p(x_t)$ and only related to the dynamics (or causal mechanism) $p(x_{t+1}|x_t)$ (compare the upper and lower figures in each column in Fig.\ref{fig:distribution}).

Tononi et al. and Hoel et al.\cite{tononi2003measuring, Hoel2013} invented effective information ($EI=I(x_t,x_{t+1}|do(x_t))$) to quantify the strength of causal effect for transitional matrices. Then Hoel et al. and Liu et al. \cite{Chvykov2020,Liu2024} extended EI to continuous systems. For linear GIS defined in Eq.(\ref{micro-dynamics}) and Eq.(\ref{micro-dynamics-TPM}), the dimensionally averaged EI  (derived in Appendix B.1-B.2) is defined as
\begin{eqnarray}\label{JGaussain}
		\mathcal{J}\equiv \frac{1}{n}I(x_t,x_{t+1}|do(x_t\sim \mathcal{U}([-L/2,L/2]^n)))=\ln\frac{{\rm pdet}(A^T\Sigma^{-1}A)^\frac{1}{2n}L}{(2\pi e)^{\frac{1}{2}}}=\frac{1}{2n}\sum_{i=1}^{r_s}\ln s_i+\ln\frac{L}{(2\pi e)^{\frac{1}{2}}},
	\end{eqnarray}
where, ${\rm pdet}(\cdot)$ represents the pseudo-determinant \cite{knill2014cauchy} of matrix $\cdot$. The operator $do(\cdot)$ signifies manual intervention, where $x_t\sim\mathcal{U}([- L/2, L/2]^n)$ after applying $do(\cdot)$ on $x_t$, with $L$ being the interval size. Eq.(\ref{JGaussain}) expands $\mathcal{J}$ into the sum of $\ln s_i$, where $s_1 \geq\cdots\geq s_{r_s}$ are the singular values of $A^T \Sigma^{-1}A$, and $r_s$ is its rank.  By comparing a-c and d-f in Fig.\ref{fig:distribution}, we can see that $\mathcal{J}$ decreases as the standard deviation $\sigma$ increases and increases as the parameter $a$ increases.  By comparing the upper and lower rows (Fig.\ref{fig:distribution}a and d, b and e, c and f), we can see that $\mathcal{J}$ is independent of the distribution $p(x_t)$ of the state $x_t$, so $\mathcal{J}$ can well measure the strength of the causal effect.

$\mathcal{J}$ in Eq.(\ref{JGaussain}) can be divided into two terms \cite {Hoel2013, Liu2024}: \textbf{determinism}
$Det=\ln\left[(2\pi e)^{-\frac{1}{2}}{\rm pdet}(\Sigma^{-1})^\frac{1}{2n}\right]$
and \textbf{non-degeneracy}
$Nondeg=\ln\left[{\rm pdet}(A^T\Sigma^{-1} A)^\frac{1}{2n}{\rm pdet}(\Sigma)^\frac{1}{2n}L\right]$.
These two indicators describe the degree of uncertainty of the effect variable when given the cause variable and the degree of uncertainty of the cause variable when given the effect variable.

\subsubsection{Approximate dynamical reversibility}\label{sec:reversibility}
\begin{figure}
    \centering
\includegraphics[width=0.9\textwidth,trim=300 700 330 20, clip]
{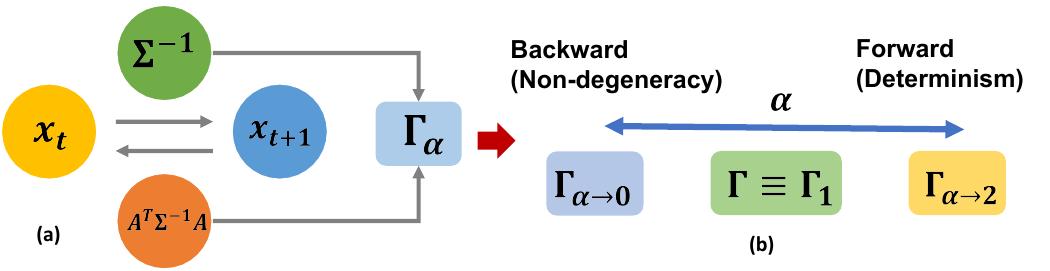}
    \caption{(a) For $p(x_{t+1}|x_t)=\mathcal{N}(Ax_t+a_0,\Sigma)$, the inverse covariance matrices of forward dynamics $x_{t+1}=a_0+Ax_t+\eta_t$ and backward dynamics $x_t=A^\dagger x_{t+1}-A^\dagger a_0-A^\dagger\eta_t$ are $\Sigma^{-1}$ and $A^T\Sigma^{-1}A$, respectively. $\Sigma^{-1}$ and $A^T\Sigma^{-1}A$ determine the approximate dynamical reversibility as $\Gamma_\alpha$ together. (b) $\alpha=1$ is often chosen to balance $\Gamma_\alpha$'s emphasis on both forward (determinism) and backward (non-degeneracy) dynamics. For $\alpha<1$, $\Gamma_\alpha$ tends to capture more information about backward dynamics $p(x_{t}|x_{t+1})$ (non-degeneracy). For $\alpha> 1$, $\Gamma_\alpha$ emphasizes more information about forward dynamics $p(x_{t+1}|x_{t})$ (determinism).}
\label{fig:concept}
\end{figure}

Another important property of dynamics that people are concerned about is dynamical reversibility. It is generally believed that the dynamics of the microscopic physical world (such as Newton's second law or the Schrödinger equation \cite{tsutsumi1987schrodinger} in quantum mechanics) are dynamically reversible, while the second law of thermodynamics in the macroscopic world stems from the irreversibility of macroscopic dynamic processes \cite{ao2008emerging, yuan2017sde}. In Eq.(\ref{micro-dynamics}), when $\eta_t=0$ and $f(x_t)=a_0+Ax_t$ is a reversible bijection mapping on $\mathcal{R}^{n}$, we assume that the dynamics of GIS described by Eq.(\ref{micro-dynamics}) is strictly reversible as $x_{t+1}=a_0+Ax_t$, and the corresponding backward dynamics of Eq.(\ref{micro-dynamics}) can be directly established as $x_{t}=A^{-1}(x_{t+1}-a_0)$. However, real-world dynamical systems around us are not so simple and perfect. 

The presence of random noise and the irreversibility of the matrix $A$ in mapping $f:\mathcal{R}^{n}\to\mathcal{R}^{n}$ can make the dynamics irreversible. However, systems may still exhibit approximate reversibility when noise is small. If $f(x_t)$ is bijective, then the dynamics with a smaller noise variance are closer to reversible. Therefore, we can use approximate dynamical reversibility to measure their closeness to reversible dynamics. In \cite{Zhang2025}, approximate dynamical reversibility can be quantified by the $\alpha$-powered sum of the singular values of the TPM (Appendix C.1). Here, we extend this concept to GIS like Eq.(\ref{micro-dynamics}) using singular value spectra from functional analysis \cite{Lax2014} and Fourier methods \cite{Bracewell1989}. We then derive the corresponding approximate dynamical reversibility indicator for GIS (derived in Appendix C.2-C.3), which is
 \begin{eqnarray}\label{Gamma_gaussian}
		\begin{aligned}
		\Gamma_\alpha=\left(\frac{2\pi}{\alpha}\right)^\frac{n}{2}{\rm pdet}(A^T\Sigma^{-1}A)^{\frac{1}{2}-\frac{\alpha}{4}}{\rm pdet}(\Sigma^{-1})^\frac{\alpha}{4}.
   \end{aligned}
\end{eqnarray} 
Apart from the constant term, the other two terms correspond to the pseudo-determinants of the inverse covariance matrices, i.e., $\Sigma^{-1}$ and $A^T\Sigma^{-1}A$, of forward and backward dynamics, respectively (see Fig.\ref{fig:concept}a). 

$\alpha\in(0,2)$ is a parameter to control the relative weight of the forward and backward inverse covariance matrices in $\Gamma_\alpha$ (see Fig.\ref{fig:concept}b). When $\alpha\to 0$, $\Gamma_{\alpha\to 0}=\left(\frac{2\pi}{\alpha}\right)^\frac{n}{2}_{\alpha\to 0}{\rm pdet}(A^T\Sigma^{-1}A)^\frac{1}{2}$ is dominated by the backward dynamics and reflects non-degeneracy after removing forward covariance interference. When $\alpha\to 2$, $\Gamma_{\alpha\to2}=\left(\pi\right)^\frac{n}{2}{\rm pdet}(\Sigma^{-1})^\frac{1}{2}$ is dominated by the forward dynamics and reflects determinism. In general, we set $\alpha=1$ gives $\Gamma\equiv\Gamma_1$, where both directions contribute equally. 

$\Gamma_\alpha$ quantifies approximate dynamical reversibility because when the uncertainties of the forward and backward dynamics are large, the pseudo-determinants of $\Sigma^{-1}$ and $A^T\Sigma^{-1}A$ are small, indicating strong irreversibility. Conversely, when these uncertainties are small (approaching zero), the pseudo-determinants grow large, implying approximate reversible dynamics.

To avoid the excessive influence of variable dimension on the reversibility of dynamics, we can use the dimensional average reversible information $\hat\gamma_\alpha\equiv\frac{1}{n}\ln\Gamma_\alpha$ to more effectively quantify the approximate dynamical reversibility. After removing the constant term, $\gamma_\alpha=\hat\gamma_\alpha-\frac{1}{2}\ln\left(\frac{2\pi}{\alpha}\right)$ is a variable only determined by the singular values after the SVD of $\Sigma^{-1}$ and $A^T\Sigma^{-1}A$ (Lemma Appendix C.2 in Appendix) as
 \begin{eqnarray}\label{dimensionaveraged}
		\begin{aligned}
		\gamma_\alpha=\frac{1}{n}\left[(\frac{1}{2}-\frac{\alpha}{4})\sum_{i=1}^{r_s}\ln s_i+\frac{\alpha}{4}\sum_{i=1}^{r_\kappa}\ln\kappa_i\right].
   \end{aligned}
\end{eqnarray}
Here $r_s,r_\kappa$ denote the ranks of $A^T\Sigma^{-1}A$ and $\Sigma^{-1}$, respectively. $s_{i}$ and $\kappa_i$ denote the $i$-th singular value of $A^T\Sigma^{-1}A$ and $\Sigma^{-1}$ respectively. Generally, we also set $\alpha=1$ and $\gamma\equiv\gamma_1$.

Below each $\mathcal{J}$ indicator in Fig.\ref{fig:distribution}, the corresponding $\gamma$ for each mapping is also shown (where $\alpha=1$). It can be seen that, same as $\mathcal{J}$, $\gamma$ decreases with the increase of the standard deviation $\sigma$ and increases with the increase of the dynamic parameter $a$. It is only related to the dynamics (the causal mechanism $p(x_{t+1}|x_t)$) and independent of the distribution $p(x_t)$. 

\subsubsection{Nearly linear relationship}
This observation is not accidental.
In fact, we have a general conclusion: $\gamma_\alpha$ is approximately linearly related to $\mathcal{J}$ (Theorem Appendix C.5 in Appendix). When the backward dynamics $p(x_t|x_{t+1})$ is close to a normalized normal distribution as $A^T\Sigma^{-1}A\approx I_n$, $\mathcal{J}$ and $\gamma_\alpha$ satisfies
\begin{eqnarray}\label{gammaabeic}
\begin{aligned}
		\gamma_\alpha\simeq (1-\frac{\alpha}{4})\mathcal{J}+C,
    \end{aligned}
\end{eqnarray}
where $C=(1-\frac{\alpha}{4})\ln(\frac{L}{\sqrt{2\pi e}})$ is a constant, and $\simeq$ denotes approximate equality, this equality holds strictly when $A$ is of full rank and $A^T\Sigma^{-1}A=I_n$, meaning the backward dynamics $p(x_t|x_{t+1})$ is a normalized normal distribution. Numerical experiments of randomly sampled $A$ and $\Sigma$ verify this linear relation, which is shown in Appendix C.4 and Figure 2 in the Appendix.

\subsection{Causal emergence (CE)}\label{sce:causalemergence}
CE refers to a special property of Markov dynamics: if the system's dynamics can yield a macroscopic dynamics with larger EI under a predefined coarse-graining strategy, then the system exhibits CE. However, this definition relies on the selection of coarse-graining strategies. Hoel et al. proposed the principle of maximizing EI \cite{Hoel2013, Hoel2017, Yang2024, Liu2024} to eliminate the arbitrariness of coarse-graining. However, this approach requires solving an often intractable optimization problem to obtain the optimal coarse-graining strategy and derive CE. To overcome this difficulty, Zhang et al. \cite{Zhang2025} introduced a new measure of CE based on approximate dynamical reversibility with SVD. Independent of coarse-graining, CE can be quantified directly from the singular value spectrum of the system’s TPM. In the following section, we will extend this framework to GIS.

\subsubsection{Causal emergence based on approximate dynamical reversibility with SVD}\label{sce:causalemergencesvd}

From the definition of dimensional average reversible information $\gamma_\alpha$ (Eq.(\ref{dimensionaveraged})) in the Section \ref{sec:reversibility}, we can obtain that its final value depends only on two spectra, one is the spectrum $s_1\geq\cdots\geq s_n$ of $A^T\Sigma^{-1}A$, and the other is the spectrum $\kappa_1\geq\cdots\geq \kappa_n$ of $\Sigma^{-1}$. Thus, CE can be defined using these singular value spectra: it occurs when most reversible information is concentrated in a few dimensions with large singular values, allowing the dynamics to be coarse-grained via SVD-based dimensional reduction.

Thus, we can quantify CE based on approximate dynamical reversibility as follows:
   \begin{eqnarray}
\label{eq:degree_vague_emergence}
        \begin{aligned}
\Delta\Gamma_\alpha(\epsilon)\equiv{\gamma}_\alpha(\epsilon)-\gamma_\alpha=\frac{1}{r_\epsilon}\left[(\frac{1}{2}-\frac{\alpha}{4})\sum_{i=1}^{r_\epsilon}\ln s_i+\frac{\alpha}{4}\sum_{i=1}^{r_\epsilon}\ln\kappa_i\right]-\frac{1}{n}\left[(\frac{1}{2}-\frac{\alpha}{4})\sum_{i=1}^{r_s}\ln s_i+\frac{\alpha}{4}\sum_{i=1}^{r_\kappa}\ln\kappa_i\right].
        \end{aligned}
        \end{eqnarray}
Here $\epsilon\geq0$ is a threshold for cutting the small singular values. $\Delta\Gamma_\alpha(\epsilon)$ is the strength of CE, and CE occurs with a threshold $\epsilon$ only if $\Delta \Gamma_{\alpha}(\epsilon)>0$. Specifically, we call \textbf{clear CE} occurs when $\epsilon=0$ (Definition Appendix D.1 in Appendix), and \textbf{vague CE} occurs in other cases (Definition Appendix D.2 in Appendix). If $A^T\Sigma^{-1}A$ or $\Sigma^{-1}$ is not full-rank, clear CE occurs; otherwise, vague CE may occur.

The basic idea behind this definition is that the differences of $\gamma$ before and after the small singular values (that is, satisfying $\kappa_i\leq\epsilon$ or $s_i\leq\epsilon$) are removed. And the number of remaining dimensions that contain two singular values is the effective rank $r_\epsilon$ as the dimension of the new system.

\begin{figure}[h]
    \centering
\includegraphics[width=0.8\textwidth,trim=21 400 280 30, clip]
{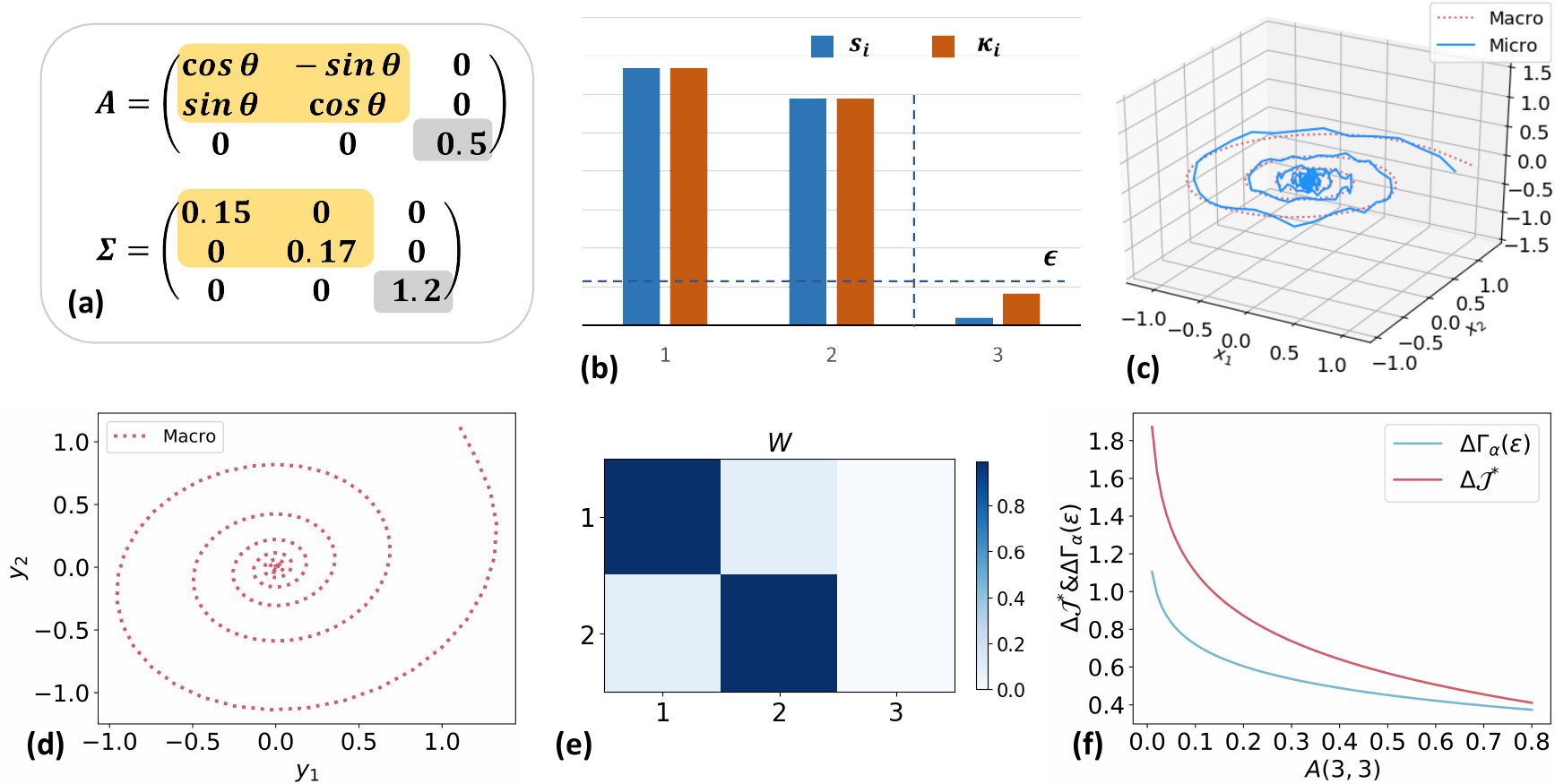}
    \caption{A case where vague CE occurs. $A$ and $\Sigma$ in (a) construct the GIS dynamics of a rotation model. In (b), the two singular values of the third dimension are less than the threshold $\epsilon=1.0$, so $\Delta\Gamma_\alpha(\epsilon)>0$ and vague CE occurs. In (c), the dynamic trajectory of macroscopic states in 3-dimensional Euclidean space $\mathcal{R}^3$ is shown. (d) and (e) correspond to the macroscopic dynamics trajectory in $\mathcal{R}^2$ obtained by the SVD method and the coarse-graining strategy based on which. (f) shows that as $A_{33}$, which represents the element of the 3rd row and 3rd column of $A$, increases. The strength of CE ($\Delta\Gamma_{\alpha}$) will decrease due to the decrease in the difference $s_2-s_3(>0)$ between singular values $s_2$ and $s_3$ of $A^T\Sigma^{-1}A$. As a comparison, the strength CE based on the maximizing EI framework ($\Delta\mathcal{J}$) also decreases similarly.}
\label{fig:singularvalues}
\end{figure}

A simple example is given to show the occurrence of vague CE in Fig.\ref{fig:singularvalues} with $\epsilon= 1.0$. This is a linear dynamical system with the given $A$ and $\Sigma$ as shown in Fig.\ref{fig:singularvalues}a, and its singular value spectra of forward and backward dynamics are shown in Fig.\ref{fig:singularvalues}b. The trajectory in three-dimensional space is shown in Fig.\ref{fig:singularvalues}c. From Fig.\ref{fig:singularvalues}b, we can observe that the spectra exhibit a sharp decay between dimensions 2 and 3, therefore CE occurs. This phenomenon can be quantified by $\gamma(\epsilon)$ by setting the threshold $\epsilon$ (horizontal dashed line) to truncate near-zero singular values (right of the vertical dashed line). Actually, zero or near-zero singular values indicate redundant dimensions with near-zero dimensional
average reversible information $\gamma$ (e.g., the $x_3$ dimension in Fig.\ref{fig:singularvalues}c). Removing these dimensions through a coarse-graining strategy (which will be mentioned in the following section) enhances $\gamma$, and a more informative dynamics (macro-dynamics) can be obtained as shown in Fig.\ref{fig:singularvalues}d. Furthermore, Fig.\ref{fig:singularvalues}f shows the significance of CE decreases with the increase of the value of the right-bottom element in $A$.

\subsubsection{Equivalence between the two frameworks}
\label{sce:causalemergenceratioality}
This framework is reasonable, and the most important reason is that our new definition of CE is consistent with Hoel et al.'s definition \cite{Hoel2013, Hoel2016, Hoel2017, Klein2020, Chvykov2020, Hoel2024,comolatti2025consilience} (Appendix D.1) based on maximizing EI. There is an approximate linear relationship, analytically, between CE based on SVD and CE based on maximizing EI (Theorem Appendix D.1 in Appendix). When the backward dynamics $p(x_t|x_{t+1})$ is close to a normalized normal distribution as $A^T\Sigma^{-1}A\approx I_n$ and $r_\epsilon=k$, $k$ is the dimension of macro dynamics we set in the framework of maximizing EI (The operational recipe for choosing the macro dimension $k$ is given in Appendix D.1.2), $\Delta\Gamma_\alpha(\epsilon)$ and $\Delta\mathcal{J}^{*}$ are approximately linear and positively correlated as
\begin{eqnarray}
\label{eq:dgammaanddJ}
        \begin{aligned}
\Delta\Gamma_\alpha(\epsilon)\simeq(1-\frac{\alpha}{4})\Delta\mathcal{J}^*.
        \end{aligned}
        \end{eqnarray}
The equal sign holds when the backward dynamics $p(x_t|x_{t+1})$ is a normalized normal distribution as $A^T\Sigma^{-1}A=I_n$. Numerical experiments of randomly sampled $A$ and $\Sigma$ verify this linear relation, which is shown in Appendix D.3 and Figure 2 in the Appendix. In Eq.(\ref{eq:dgammaanddJ}), $\Delta\mathcal{J}^{*}$ is the optimal result of CE based on the EI maximization framework proposed by Hoel et al. and Liu et al. \cite{Hoel2013,  Hoel2017, Yang2024, Liu2024}
defined as
\begin{eqnarray}\label{DeltaJstar}
\Delta\mathcal{J}^{*}=\frac{1}{2k}\sum_{i=1}^{k}\ln s_i-\frac{1}{2n}\sum_{i=1}^{r}\ln s_i,
\end{eqnarray}
in which $s_1\geq\cdots\geq s_n$ represents the singular values of the matrix $A^T\Sigma^{-1} A$, $r$ is its rank. In the simple example of CE shown in Fig. \ref{fig:singularvalues}, both $\Delta\Gamma_{\alpha}$ and $\Delta \mathcal{J}^*$ decrease with increasing $A_{33}$, which is the value of the 3-row and 3-column element of $A$.

From another perspective, our definition of CE is actually an approximate necessary condition for maximizing EI-based CE as defined by Hoel et al. \cite{Hoel2013,Hoel2017}. In other words, for any Markov dynamics, if we want to find a coarse-graining strategy that maximizes EI, we can actually perform SVD on the dynamics first. Although there is a high likelihood that the two frameworks overlap, they are not entirely equivalent because the coarse-graining strategy obtained by maximizing EI ignores the covariance of forward dynamics. So, we need to propose a method that constructs a coarse-graining strategy based on SVD, in which both forward and backward dynamics are taken into account.

\subsubsection{Coarse-graining strategy based on SVD}\label{Coarse-graining}
For any GIS, we can first perform SVD on both $A^T\Sigma^{-1}A$ and $\Sigma^{-1}$ to obtain the two spectra $s_1>\cdots>s_n$ and $\kappa_1>\cdots>\kappa_n$. We then need to consider both sets of spectra for truncation. However, since their corresponding singular vectors may not all be orthogonal, their subspaces may overlap or conflict. Therefore, we combine two singular vector matrices and perform a second round of SVD. Based on this, we can obtain the coarse-graining matrix and the corresponding macro dynamics accordingly by dimension reduction. After these two rounds of SVD, we obtain the final result (detailed in Appendix E.2).

\subsection{Extension for nonlinear system}\label{sce:localcausalemergence}
The content discussed above is for linear GIS. For the more universal form of Eq.(\ref{micro-dynamics normal}), we can perform a local Taylor expansion of $f(x)$, thereby transforming it into the linear form of Eq.(\ref{micro-dynamics}). In this way, we can calculate the quantification of CE ($\Delta\Gamma_{\alpha}^{x_t}(\epsilon)$) locally (which depends on $x_t$).

Since most of the dynamics we consider are stationary Markov processes, we can discuss their global CE properties. For this purpose, we can approximate the global properties of CE by averaging the local properties of dynamics.
For general problems, if $x_{t+1}$ is approximately distributed within the interval $\mathcal{X}$, and $A^T_{x_t}\Sigma^{-1}_{x_t}A_{x_t}$ and $\Sigma^{-1}_{x_t}$ do not fluctuate or diverge significantly, we can first randomly sample and then average the parameter matrix to obtain
$A^T\Sigma^{-1}A\approx\left<A^T_{x_t}\Sigma^{-1}_{x_t}A_{x_t}\right>_{x_t\in\mathcal{X}}$
and
$\Sigma^{-1}\approx\left<\Sigma^{-1}_{x_t}\right>_{x_t\in\mathcal{X}}$ to calculate the two global inverse covariance matrices. Finally, we can calculate global CE based on the averaged results of these two matrices (detailed in Appendix A.1).

The above is our theoretical part. In the next section, we show the results of applying our framework to three practical cases.

\begin{figure}[htbp]
    \centering
\includegraphics[width=0.9\textwidth,trim=0 65 470 0, clip]
{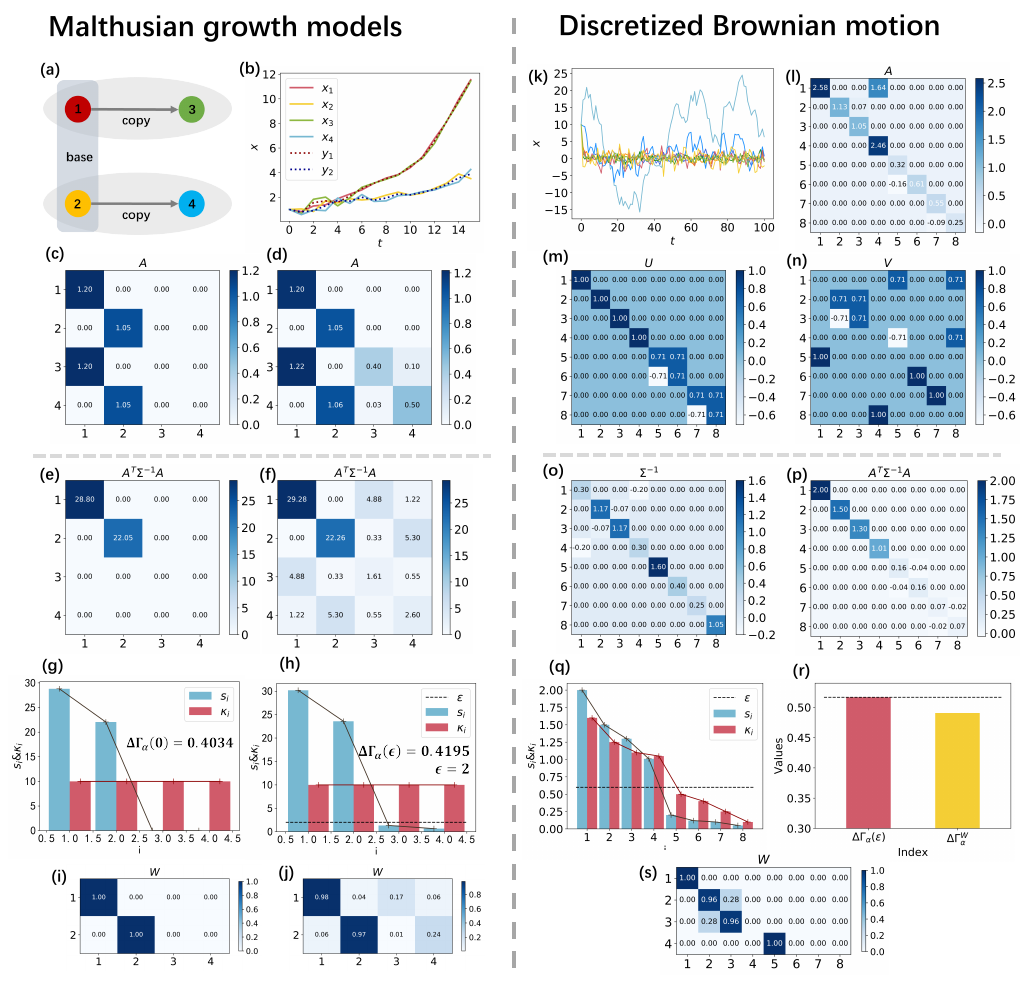}
    \caption{
    Two examples of dynamics where CE occurs, one's CE is primarily caused by the coefficient $A$, the other is primarily caused by the covariance $\Sigma$. (a) A Malthusian growth model where $x_3,x_4$ are copies of $x_1,x_2$. (b) A sample with growth rates of 0.2 and 0.05. $x_i$, $i=1,2,3,4$, are micro-states, while $y_1$ and $y_2$ are macro-states we derived by the method mentioned in Section \ref{sce:causalemergencesvd}. (c) The original parameter matrix $A$. (d) Parameter matrix $A$ with perturbations. (e) The original backward dynamics inverse covariance matrix $A^T\Sigma^{-1}A$ with rank $r_s=2$. (f) $A^T\Sigma^{-1}A$ after random perturbations to $A$. (g) The singular value spectrum of $A^T\Sigma^{-1}A$ and $\Sigma^{-1}$ when $A^T\Sigma^{-1}A$ only has two nonzero singular values. Clear CE occurs, and the strength can be calculated as $\Delta\Gamma_\alpha(0)=0.4034$. (h) The singular value spectrum of $A^T\Sigma^{-1}A$ and $\Sigma^{-1}$ when $A$ is perturbed. Vague CE occurs, and the strength of this CE is $\Delta\Gamma_\alpha(\epsilon)=0.4195$ when $\epsilon=2$. (i) The derived coarse-graining parameter $W$ according to the method mentioned in Section \ref{Coarse-graining} and Appendix E.2, where the number of columns represents the macroscopic dimensions, and the number of rows represents the microscopic dimensions. (j) The coarse-graining parameter $W$ for the vague CE case, the first macro-state dimension is determined by $x_1,x_3$, and the second macro-state dimension is determined by $x_2,x_4$. (k) The samples of trajectories of $x_t$ in Brownian motion. (l) Parameter $A$ of the drift term $f(x_t)$. (m) and (n) show singular matrices $U$ and $V$. (o) and (p) show the inverse of the covariance matrices of forward and backward dynamics, which are $\Sigma^{-1}$ and $A^T\Sigma^{-1}A$. (q) Singular value spectra of $\Sigma^{-1}$ and $A^T\Sigma^{-1}A$ with $\epsilon=0.6$. (r) Since $\Delta\Gamma_\alpha(\epsilon)$ is the theoretical value of CE, we can use this value to compare with the difference $\gamma_\alpha^W-\gamma_\alpha$ between the $\gamma^W_\alpha$ of our coarse-grained macro dynamics and the $\gamma_\alpha$ of the original micro dynamics. Approximate CE based on SVD $\Delta\Gamma_\alpha^W=\gamma_\alpha^W-\gamma_\alpha=0.4907$ as $W$ is derived from Section \ref{Coarse-graining}, which is close to the theoretical value $\Delta\Gamma_\alpha(\epsilon)=0.5167$. (s) The coarse-graining strategy parameter matrix $W$ derive from methods in Section \ref{Coarse-graining}, which preserves the non-conflicting 2nd and 3th dimensions, along with the 1st and 5th dimensions with larger singular values $\kappa_1$ and $s_1$.}
\label{fig:Known_model}
\end{figure}
\section{Results}
\label{results}
In this section, we present the results of three numerical experimental examples to illustrate the concepts and conclusions derived in the previous sections. The first two are for GIS with known models, and the third is the application on a neural network trained on time series data generated by a Susceptible-Infected-Recovered (SIR) model.
\subsection{Malthusian growth models}\label{growth models}

The first example contains 4 variables, in which the first two variables $x_1,x_2$ follow the Malthusian growth model \cite{Galor2000} with different growth rates of 0.2 and 0.05. To study CE, we define the other two variables $x_3,x_4$ as the copies of the first two variables shown as Fig.\ref{fig:Known_model}a, thus, they are redundant dimensions. 

If $x=(x_1,x_2,x_3,x_4)$, the evolution of $x$ is a GIS $x_{t+1}=a_0+Ax_t+\eta_t, \eta_t\sim\mathcal{N}(0,\sigma^2 I_4)$ as $x_t,x_{t+1}\in\mathcal{R}^{4}$, $\sigma^2=0.1$, $a_0=0$, and $A$ is shown in Fig.\ref{fig:Known_model}c. Dynamical trajectories of this model are shown in Fig.\ref{fig:Known_model}b, $y_1$ and $y_2$ are theoretical macro-states. In this model, $\Sigma$ is a full rank matrix, so we only need to study the inverse backward dynamics covariance matrix $A^T\Sigma^{-1}A$ (Fig.\ref{fig:Known_model}e). This matrix only has two singular values as the singular value spectrum is shown in Fig.\ref{fig:Known_model}g with $r=2<4$, the horizontal axis represents the sequence number of singular values arranged in descending order as $s_1\geq\dots\geq s_4$, while the vertical axis represents the magnitude of singular values $s_i$ and $\kappa_i$ as $\kappa_1=\cdots=\kappa_4=10$. So when $r_s=2<r_\kappa=4$, clear CE occurs and the strength can be calculated as $\Delta\Gamma_\alpha(0)=0.4034$. Fig.\ref{fig:Known_model}i shows the coarse-graining parameter $W\in\mathcal{R}^{2\times 4}$ obtained by methods in Section \ref{Coarse-graining}, which is equal to remove $x_3$ and $x_4$ and only retain $x_1$ and $x_2$ to describe the evolution of the growth model since variables $x_1$ and $x_2$ hold all the relevant information contained within $x_3$ and $x_4$.

To show the concept of vague CE, we can add some perturbations to $A$ in Fig.\ref{fig:Known_model}c and obtain $A$ in Fig.\ref{fig:Known_model}d. Then $A^T\Sigma^{-1}A$ (Fig.\ref{fig:Known_model}f) is a full rank matrix as $r_s=r_\kappa=n=4$, and clear CE occurs with the strength $\Delta\Gamma_\alpha(0)=0$. However, by observing the singular value spectrum in Fig.\ref{fig:Known_model}h, we can see that only two dimensions have significant impacts, so we need to select a threshold ($\epsilon=2$) to truncate the small singular values such that the occurrence of vague CE can be shown. Therefore, the CE strength can be calculated as $\Delta\Gamma_\alpha(\epsilon)=0.4195$ and $r_\epsilon=2$ can also be derived. According to Section \ref{Coarse-graining}, the coarse-graining parameter $W$ is shown in Fig.\ref{fig:Known_model}j, the columns represent the macroscopic dimensions, and the rows represent the microscopic dimensions. Due to disturbances, $x_3$ and $x_4$ also contain some independent information. Coarse-graining strategy $\phi(x)=Wx$ is to merge $x_1$ and $x_3$, $x_2$ and $x_4$ by weighted summation, and the two row-vectors that make up $W$ are exactly equal to the singular vectors corresponding to the two largest singular values $s_1, s_2$.

This case illustrates that sometimes high dimensions may reduce the efficiency of reversibility of a system due to dimensional redundancy. Retaining only the maximum $r_\epsilon$ singular values of the inverse covariance matrix can improve the efficiency of reversibility and generate CE.

\subsection{Discretized Brownian motion}\label{Discretized Brownian motion}
The design of this example is to demonstrate the need to simultaneously filter the two singular value spectra of $\Sigma^{-1}$ and $A^T\Sigma^{-1}A$.
Discretized Brownian motion is an approximation of continuous Brownian motion in discrete time, which is always used for numerical simulation and stochastic process modeling. Eq.(\ref{micro-dynamics}) can be regarded as a discrete version of the Ornstein-Uhlenbeck (OU) \cite{maller2009ornstein} process. In this model, $f(x_t)=a_0+Ax_t$ is the drift vector, which influences the evolution of the state, and $A$ is shown in Fig.\ref{fig:Known_model}l 
and $a_0=0$. Covariance matrix $\Sigma$ representing the diffusion coefficient, which determines the magnitude and correlation of random fluctuations across dimensions of $\eta_t$ like Fig.\ref{fig:Known_model}k.

As we set $A$ as shown in Fig.\ref{fig:Known_model}l, after perform SVD on $\Sigma^{-1}$ and $A^T\Sigma^{-1}A$ in Fig.\ref{fig:Known_model}o and Fig.\ref{fig:Known_model}p, as $A^T\Sigma^{-1}A=USU^T,\Sigma^{-1}=VKV^T$, we can obtain singular vector matrices $U=(u_1,\cdots,u_n)$ and $V=(v_1,\cdots,v_n)$ as shown in Fig.\ref{fig:Known_model}m and n and the diagonal matrices of singular values $S={\rm diag}(2,1.5,1.3,1.01,0.2,0.12,0.1,0.05)$, $K={\rm diag}(1.6,1.25,1.1,1.05,0.5,0.4,0.25,0.1)$. After obtaining the singular value spectra in Fig.\ref{fig:Known_model}q, we specify $\epsilon=0.6$ to get $\Delta\Gamma_\alpha(\epsilon)=0.5167$ and the number of macro-states is $r_\epsilon=4$. Ideally, $s_1,\cdots,s_4$ and $\kappa_1,\cdots,\kappa_4$ should be retained.

As we get the singular value spectrum of $\Sigma^{-1}$ and $A^T\Sigma^{-1} A$ in Fig.\ref{fig:Known_model}q, we can also calculate the coarse-graining strategy based on SVD. According to Section \ref{Coarse-graining}, we can obtain the optimal coarse-graining strategy parameter matrix $W$ in Fig.\ref{fig:Known_model}s, which preserves the non-conflicting 2nd and 3rd dimensions, along with the 1st and 5th dimensions with larger singular values $\kappa_1$ and $s_1$. So in reality, we retained $s_1,s_2,s_3,s_5$ and $\kappa_1,\kappa_2,\kappa_3,\kappa_5$.

 Since $\Delta\Gamma_\alpha(\epsilon)$ is the theoretical value of CE, we can use this value to compare with the difference $\gamma_\alpha^W(\epsilon)-\gamma_\alpha$ between the $\gamma^W_\alpha(\epsilon)$ of coarse-grained macro dynamics and the $\gamma_\alpha$ of the original micro dynamics. The macro dynamics obtained through $\phi=Wx_t$ have their own approximate dynamical reversibility $\gamma_\alpha^W$ as $W$ is derived from Section 2.3.3. We can obtain an approximate CE based on the coarse-graining strategy $\Delta\Gamma_\alpha^W=\gamma_\alpha^W-\gamma_\alpha=0.4907$, which is close to the theoretical $\Delta\Gamma_\alpha(\epsilon)=0.5167$. By comparing in Fig.\ref{fig:Known_model}r, we can see that $\Delta\Gamma_\alpha^W$ obtained by $\phi(x_t)=Wx_t$ is close to the theoretical value $\Delta\Gamma_\alpha(\epsilon)$.

\begin{figure}[htbp]
    \centering
\includegraphics[width=1\textwidth,trim=0 180 490 0, clip]
{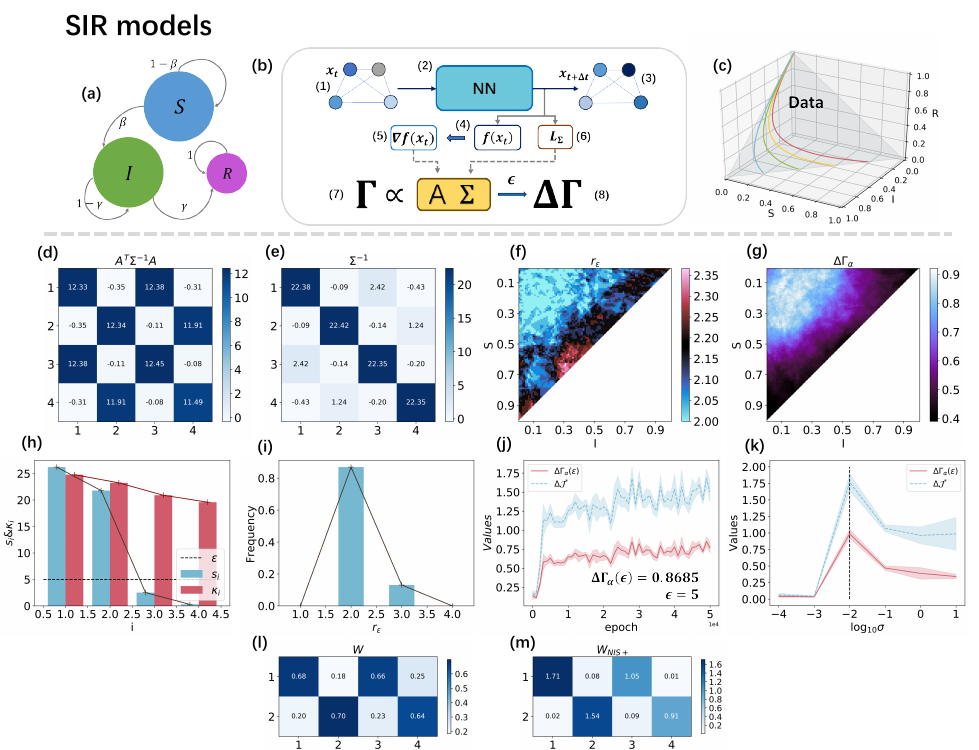}
    \caption{An example of applying machine learning to extract dynamics from data and measure CE of the dynamical system. (a) Schematic diagram of the SIR model. (b) Schematic diagram of learning dynamics and calculating $\Gamma$ through NN based on data. (1) and (3) are the variables $x_t$ and predicted $x_{t+\Delta t}$ at adjacent time points. (2) NN model for machine learning, whose structure is in Appendix F.3. (4) Dynamics $f(x_t)$ learned by NN model based on data of SIR. (5) $A_{x_t}=\nabla f(x_t)$ is the Jacobian matrix of the trained model at $x_t$. (6) The output $L_\Sigma$ of the Covariance Learner Network, the covariance matrix $\Sigma_{x_t}$ is derived from which. (7) The numerical solution of $\Gamma$, which is related to $A$ and $\Sigma$. (8) Applying $A^T\Sigma^{-1}A=\left<A^T_{x_t}\Sigma^{-1}_{x_t}A_{x_t}\right>_{x_t\in\mathcal{X}}$ and $\Sigma^{-1}=\left<\Sigma^{-1}_{x_t}\right>_{x_t\in\mathcal{X}}$, SVD-based CE $\Delta\Gamma(\epsilon)$ can be calculated under suitable $\epsilon$, $\mathcal{X}$ is the domain of SIR. (c) Training data which is the same as NIS+ \cite{Yang2024}. The full dataset (entire triangular region) used for training is displayed, along with four example trajectories with the same infection and recovery or death rates. The method of generating training data can be found in Appendix F.3. (d) $A^T\Sigma^{-1}A=\left<A^T_{x_t}\Sigma^{-1}_{x_t}A_{x_t}\right>_{x_t\in\mathcal{X}}$ derived by NN trained SIR model. (e) $\Sigma^{-1}=\left<\Sigma^{-1}_{x_t}\right>_{x_t\in\mathcal{X}}$ derived by NN trained SIR model. (f) Distribution of $r_\epsilon$ obtained by traverse the dynamical space of SIR with different $(S,I)$ and $x_t$ generated from which. (g) Distribution of $\Delta\Gamma_\alpha$ under different $(S,I)$ obtained by traverse the dynamical space. (h) The singular value spectra of $A^T\Sigma^{-1}A$ and $\Sigma^{-1}$ in NN trained SIR model. (i) The frequency of $r_\epsilon$ under different samples $x_t$ in test dataset. (j) Maximum $\Delta\Gamma_\alpha(\epsilon)=0.8685$ when the training period is around 50,000 as $\epsilon=5$. (k) The changing trend of CE under different $\sigma$, the threshold for trend change is around $\sigma=0.01$. (l) The coarse-graining matrix $W$ derived from Section \ref{Coarse-graining} and Appendix E.2 use $A^T\Sigma^{-1}A=\left<A^T_{x_t}\Sigma^{-1}_{x_t}A_{x_t}\right>_{x_t\in\mathcal{X}}$ and $\Sigma^{-1}=\left<\Sigma^{-1}_{x_t}\right>_{x_t\in\mathcal{X}}$. (m) The coarse-graining matrix $W_{NIS+}$ derived from NIS+. For the convenience of observation, we have taken absolute values for each dimension in $W$ and $W_{NIS+}$.}
\label{fig:SIR}
\end{figure}

\subsection{SIR based on NN}\label{NNSIR}
Except for analytical solutions for CE on known models, this framework can also be applied to observed time series data. Our third case is to show the phenomenon of CE obtained by a well-trained neural network (NN) on time series data generated by the Susceptible-Infected-Recovered (SIR) dynamical model \cite{Satsuma2004} in Fig.\ref{fig:SIR}a with the following dynamics:
\begin{eqnarray}\label{SIRmodel}
\begin{aligned}
\begin{cases}  
\frac{\mathrm{d}S}{\mathrm{d}t}=-\beta SI,  \\
\frac{\mathrm{d}I}{\mathrm{d}t}=\beta SI - \gamma I, \\
\frac{\mathrm{d}R}{\mathrm{d}t}= \gamma I,
\end{cases}
\end{aligned}
\end{eqnarray}
where $S,I,R\in[0,1]$ represent the rate of susceptible, infected, and recovered individuals in a population, $\beta=1$ and $\gamma=0.5$ are parameters for infection and recovery rates.

To generate the time series data of the micro-state, we adopt the same method as in our previous work of NIS+ \cite{Yang2024}. In which we duplicate the macro-state $(S_t,I_t)$ as shown in Fig.\ref{fig:SIR}c and add Gaussian random noise to form the micro-state $x_t$. By feeding the micro-state data into a forward neural network (NN) called \textbf{Covariance Learner Network} in Fig.\ref{fig:SIR}b, we can use this model to approximate the micro-dynamics of the SIR model  (Appendix F.3 shows the details of generating data and the structure of NN). 

We can obtain the CE identification result directly by calculating $A_{x_t}=\nabla f(x_t)$ as the Jacobian matrix, and $\Sigma_{x_t}$ is the covariance matrix that neural networks can directly output at sample $x_t$. We can randomly generate $x_t$ with a uniform distribution on the domain $\mathcal{X}$ of the SIR dynamics and take the average value as Section \ref{sce:localcausalemergence} to calculate global $A^T\Sigma^{-1} A$ and $\Sigma^{-1}$ in Fig.\ref{fig:SIR}d and Fig.\ref{fig:SIR}e. In Fig.\ref{fig:SIR}h, we can see the singular value spectrum of matrices $A^T\Sigma^{-1} A$ and $\Sigma^{-1}$ in which the horizontal axis represents the sequence number of singular values arranged in descending order, while the vertical axis represents the magnitude of singular values $s_i,\kappa_i$. Due to the existence of the copy operation, even the simplest NN can recognize only two dimensions with larger singular values as $r_\epsilon=2$ when $\epsilon=5$. In model training, we get the value of Vague CE as $\Delta\Gamma_\alpha(\epsilon)=0.8685$ as shown in Fig.\ref{fig:SIR}j.

Locally, based on directly calculating $r_\epsilon^{x_t}$ by $A^T_{x_t}\Sigma^{-1}_{x_t}A_{x_t}$ and $\Sigma^{-1}_{x_t}$ at samples $x_t$ obtained by traverse the entire SIR dynamic space $\mathcal{X}$, the resulting matrix shows $r_\epsilon = 2$ at most positions in Fig.\ref{fig:SIR}f. Frequency analysis further confirms a similar observation, revealing that $r_\epsilon = 2$ occurs with the highest probability across different samples in the test data of NN training, as shown in Fig.\ref{fig:SIR}i. Computing $\Delta\Gamma_\alpha(\epsilon)$ traverses the entire space and visualizing it in Fig.\ref{fig:SIR}g, we observe that the regions of highest $\Delta\Gamma_\alpha(\epsilon)$ align closely with the most stable $r_\epsilon$ areas in Fig.\ref{fig:SIR}f. Both are concentrated within a circle of radius 0.5 centered at $(0,0)$. It can be observed that due to the stability of the dynamical space, the local and global measures of CE are similar.

We can also compare the results of our framework with those of NIS+ mentioned in \cite{Yang2024}. From Fig.\ref{fig:SIR}l and Fig.\ref{fig:SIR}m, we can see that $W$ obtained by methods in Section \ref{Coarse-graining} and $W_{NIS+}$ of the NIS+ are similar as both the values of $W$ and $W_{NIS+}$ are consistent with the copy method (Fig.\ref{fig:Known_model}a). Additionally, we can examine the changing trend of CE under different $\sigma$ values. From Fig.\ref{fig:SIR}k, we can see that the turning point for CE is around $\sigma=0.01$, which is consistent with the value obtained by NIS+ \cite{Yang2024}. 

This simple case integrates our new framework with machine learning, providing a fundamental reference for addressing real-world problems with real data.

\section{Conclusion and discussion}
This work introduces the notion of approximate dynamical reversibility $\Gamma_\alpha$ and an alternative exact quantification of CE based on it for GIS. $\Gamma_\alpha$ depends only on the singular values of $\Sigma^{-1}$ and $A^T\Sigma^{-1}A$ and independent of predefined coarse-graining method. By retaining the top $r_\epsilon$ singular values of $\Sigma^{-1}$ and $A^T\Sigma^{-1}A$, CE can be quantified directly as $\Delta\Gamma_\alpha(\epsilon)$. Compared to EI-based CE, this SVD-based approach offers higher accuracy and computational efficiency, laying the foundation for detecting CE in complex systems.

Meanwhile, the coarse-graining strategy can also be derived directly from the singular vector matrices of $\Sigma^{-1}$ and $A^T\Sigma^{-1}A$, balancing the preservation of key singular values of two matrices and improving $\Delta\Gamma_\alpha(\epsilon)$. The resulting strategy reveals natural aggregation strategies: correlated dimensions merge (like the copied dimension in the growth model and SIR model), and weak ones are discarded (like the discarded 4th, 6th, 7th, and 8th dimensions with larger noise in Brownian motion). Unlike conventional CE frameworks that maximize EI through manually designed coarse-graining strategies, this method simplifies computation and produces more interpretable, physically grounded macroscopic predictions.

Compared with the SVD-based CE of TPM, $\Delta\Gamma_\alpha$ in GIS more clearly reveals the link between SVD-based and EI-based CE. Analytical results show that $\Delta\Gamma_\alpha(\epsilon)\simeq(1-\frac{\alpha}{4})\Delta\mathcal{J}^*$, with equality when the backward dynamics follow a normalized normal distribution. Randomly sampled $A$ and $\Sigma$ confirm this linear relation. The parameter $\alpha$ further modulates determinism and non-degeneracy. In continuous space, our method can provide theoretical support for some discrete TPM results that can only be obtained through numerical experiments, and can be applied to more real-world systems.


Beyond introducing the CE based on dynamical reversibility, this work also innovates in the use of SVD. SVD has long served as a key tool for analyzing system dynamics—from detecting critical transitions in Earth and market systems \cite{Sun2021,Hu2019,fan2021statistical,sornette2014physics} used ensemble singular values, to model reduction in control theory, neural networks, even the Koopman operator \cite{Gallivan1994,Gugercin2008,Antoulas2005,schmid2010dynamic,williams2015data}. Inspired by these foundations, we extend SVD to stochastic dynamics by applying it to continuous Markovian dynamics, which jointly captures the deterministic mapping $A$ and stochastic covariance $\Sigma$. Furthermore, our framework treats the inverse covariance matrices $A^T\Sigma^{-1}A$ and $\Sigma^{-1}$ as intrinsic information carriers, simultaneously accounting for both forward and backward dynamics. The singular values of the inverse covariance matrices reveal the reversible information embedded in the system, offering a unified view that encompasses certainty, stochasticity, and reversibility to expose the fundamental degrees of freedom driving the evolution of complex systems.

While our approach has made progress, several challenges remain unresolved. The first limitation is that our model is currently restricted to linear GIS, with nonlinear GIS approximated in a locally linear form. However, this approximation introduces stability issues, as the gradient of nonlinear functions may be ill-conditioned. In particular, cases where $\nabla f(x)=0$ or $\nabla f(x)\to\infty$ lead to the breakdown of our framework. To address this, incorporating higher-order derivatives as a refined CE metric warrants further investigation.

The second issue arises when time is continuous rather than discrete; the existing CE quantification lacks objective formulations. The approach in this study discretizes time, converting differential equations into difference equations, where the choice of hyperparameter $\Delta t$ strongly influences CE. As $\Delta t\to 0$, state transitions $x_t\to x_{t+\Delta t}$ exhibit minimal variation, causing the rate of change to vanish. For continuous-time stochastic differential equations, a more principled CE measure is required to account for infinitesimal evolution dynamics.

The third issue is that both SVD-based and EI-based CE quantification methods require training an NN to infer dynamics when the governing equations are unknown. However, NN-based approaches are data-dependent and prone to parameter estimation errors, particularly in capturing interdimensional correlations. In our case, a multivariate Gaussian model approximates both the dynamical function and covariance. Yet, under high noise or limited data, the learned dynamics may deviate from the true system, leading to unreliable CE estimates. For systems with unknown models, it is crucial to develop representations that jointly approximate both the underlying dynamics and noise structure.

Future work will focus on optimizing the existing framework and extending its application to more complex systems. In numerical simulations, machine learning can be leveraged to learn intricate dynamical models, such as Vicsek and Kuramoto, which are analytically intractable. This allows for data-driven CE estimation and the exploration of its relationship with critical states. Additionally, our approach can be applied to real-world datasets, such as meteorological and EEG data, to identify practical problems where CE provides meaningful insights.

\section{Supplementary Material}
The supplementary material provides detailed calculation processes for effective information $EI$ and approximate dynamical reversibility $\Gamma_\alpha$, from the simplest 1-dimensional system to general high-dimensional situations. Especially for $EI$ and $\Gamma_\alpha$ in cases where $A$ is a full-rank or non-square matrix, we provide computational references. After providing a detailed proof process, we use randomly generated scatter plots to show the linear correlation between two indicators $\Gamma_\alpha$ and $\mathcal{J}$, as well as $\Delta\Gamma_\alpha(\epsilon)$ and $\Delta\mathcal{J}^{*}$ based on which. The original coarse-graining strategy and the coarse-graining strategy of the standardized model, which were not mentioned in the article, have also been discussed in the supplementary material. In the end, we add the case of the Markov Gaussian system, providing ideas for combining our framework with integrated information theory, Markov chain with TPM, and even neural data.\\

\textbf{Acknowledgments:} This work is inspired by the insightful discussions held in the "Causal Emergence" series reading groups organized by Swarma Club and the "Swarma-Kaifeng" Workshop. Sincere thanks to Mr. Jiayun Wu, Ms. Ruiyang Feng, Mr. Zhang Zhang, Ms. Ruyi Tao, Mr. Keng-Hou Leong, Ms. Wanting Su, Mr. Aobo Lyu, and Mr. Muyun Mou for their suggestions and support on theory, methods, and writing. We would also like to extend our sincere gratitude to Prof. Yizhuang You, Dr. Peng Cui, Dr. Jing He, Dr. Chaochao Lu, and Dr. Yanbo Zhang for their invaluable contributions and insightful discussions. We express our appreciation for the support from the Save 2050 Programme, a joint initiative of Swarma Club and X-OrderSwarma Research. Additionally, we are grateful for the support received from Swarma Research. We also thank the large language models ChatGPT and Deepseek for their assistance in refining our paper writing. \\

\textbf{Data Availability Statement:} All the codes and data are available at: \url{https://github.com/kilovoltage/SVD-based_CE_Gaussian}.


\bibliographystyle{ieeetr}
\bibliography{jiaoda11}

\newpage

\section*{\centering{Supplementary Information for SVD-based Causal Emergence for Gaussian Iterative Systems}}

\appendix
\renewcommand{\thesection}{Appendix \Alph{section}}
\renewcommand\theequation{\Alph{section}.\arabic{equation}} 
The supplementary material provides detailed calculation processes for effective information $EI$ and approximate dynamical reversibility $\Gamma_\alpha$, from the simplest one-dimensional system to general high-dimensional situations. Especially for $EI$ and $\Gamma_\alpha$ in cases where $A$ is full rank or non-square matrices, we provide computational references. After providing a detailed proof process, we use randomly generated scatter plots to show the linear correlation between two indicators $\Gamma_\alpha$ and $\mathcal{J}$, as well as $\Delta\Gamma_\alpha(\epsilon)$ and $\Delta\mathcal{J}^{*}$ based on which. The original coarse-graining strategy and the coarse-graining strategy of the standardized model, which were not mentioned in the article, have also been discussed in the supplementary material. In the end, we add the case of the Markov Gaussian system, providing ideas for combining our framework with integrated information theory, Markov chains with TPM, and even neural data.
\section{Linear Gaussian iterative systems}\label{Fundamentaltheories-appendix}

In this Section, we provide some supplements to the content about dynamics in the main text. Many real-world complex systems exhibit dynamical behavior constructed by the interplay of deterministic laws and random disturbances as
\begin{equation}
\label{micro-dynamics normal-appendix}
x_{t+1}=f(x_t)+\eta_t, \eta_t\sim\mathcal{N}(0,\Sigma), 
\end{equation}
$x_t\in\mathcal{R}^n, \Sigma\in\mathcal{R}^{n\times n}$, which we called Gaussian iterative systems (GIS) \cite{Dunsmuir1976}. $f(x_t): \mathcal{R}^n\to \mathcal{R}^n$ captures deterministic dynamics, $\eta_t$ denotes random noise. Some continuous dynamic systems can be approximated in this form.

This paper mainly focuses on linear high-dimensional stochastic iteration systems \cite{Hannan1984} to reduce our calculation.  For variable $x_t\in\mathcal{R}^n$ at time $t$ in the stochastic iteration system, the evolution of $x_t$ follows the stochastic iterative equation The GIS we discuss is presented as 
\begin{equation}
\label{micro-dynamics-appendix}
x_{t+1}=a_0+Ax_t+\eta_t, \eta_t\sim\mathcal{N}(0,\Sigma), 
\end{equation}
where, $a_0,x_t\in\mathcal{R}^n, A,\Sigma\in\mathcal{R}^{n\times n}$. 
We specify the ranks of matrices ${\rm rk}(A)={\rm rk}(\Sigma)=n$. ${\rm rk}(\cdot)$ means the rank of a matrix. We can easily observe that the probability of $x_{t+1}$ falling in a region centered with $Ax_{t}$.
\subsection{Linearization of nonlinear dynamics}\label{Nonlineardynamics}

Most of the analysis in the article is based on linear models, but there are also many nonlinear iterative models in reality. We can apply our framework to nonlinear models under certain conditions. Nonlinear stochastic iterative systems like $x_{t+1}=f(x_t)+\eta_t$, $\eta_t\sim\mathcal{N}(0,\Sigma)$, $x_t\in\mathcal{R}^n$, $\Sigma\in\mathcal{R}^{n\times n}$, do not have the same known parameter matrix $A$ as linear stochastic iterative systems. However, when $f:\mathcal{R}^n\to \mathcal{R}^n$ and spatially continuous, the Taylor expansion $f(x)=f(x^{'})+\nabla f(x^{'})(x-x^{'})+o(x-x^{'})$ is very similar to linear function $Ax$ around $x^{'}$. By using the definition that $A$ is the slope, we can obtain an approximate expression of $A_{x^{'}}$ when $x\approx x^{'}$ as
\begin{eqnarray}\label{afx}
    A_{x^{'}}=\frac{f(x)-f(x^{'})}{x-x^{'}}\approx\nabla f(x^{'}).
\end{eqnarray}
So when iterative models are nonlinear, we can replace $A$ with the gradient matrix as $A_{x_t}=\nabla f(x_t)\in \mathcal{R}^{n\times n}$ at $x_t$ to calculate $\Delta\mathcal{J}$ or $\Delta\Gamma_\alpha$. Since causal emergence is related to $x_t$, to determine the CE of the entire system, we can take the average of $\Sigma^{-1}$ and $A^T\Sigma^{-1}A$ in $x_t$'s space $\mathcal{X}$ as
\begin{eqnarray}\label{afxbar}
    \Sigma^{-1}=\frac{1}{|\mathcal{X}|}\displaystyle\int_{\mathcal{X}}\Sigma^{-1}_{x_t}dx_t
\end{eqnarray}
and
\begin{eqnarray}\label{afxbar}
    A^T\Sigma^{-1}A=\frac{1}{|\mathcal{X}|}\displaystyle\int_{\mathcal{X}}\nabla f(x_t)^T\Sigma^{-1}_{x_t}\nabla f(x_t) dx_t
\end{eqnarray}
to calculate the CE of the entire system, $|\mathcal{X}|$ represents the size of $\mathcal{X}$ and $\Sigma_{x_t}$ is the covariance at $x_t$.

For the case of sampling, we can change integration to summation. Since most of the dynamics we consider are stationary Markov processes, we can discuss their overall approximate dynamical reversibility and CE properties. For this purpose, we can approximate the overall properties by averaging the properties of each local point in a certain interval.
For general problems, if $x_{t+1}$ is approximately distributed within the interval $\mathcal{X}$, and $A^T_{x_t}\Sigma^{-1}_{x_t}A_{x_t}$ and $\Sigma^{-1}_{x_t}$ do not fluctuate or diverge significantly, we can first randomly sample and then average the parameter matrix to obtain
\begin{eqnarray}
A^T\Sigma^{-1}A\approx\left<A^T_{x_t}\Sigma^{-1}_{x_t}A_{x_t}\right>_{x_t\in\mathcal{X}}
\end{eqnarray}
and
\begin{eqnarray}
\Sigma^{-1}\approx\left<\Sigma^{-1}_{x_t}\right>_{x_t\in\mathcal{X}}
\end{eqnarray}
to calculate the two global inverse covariance matrices. Both matrices are dependent only on the spatial $\mathcal{X}$ and are independent of time, deriving the same results regardless of the system's evolution. Finally, we can calculate global CE based on the definition of $\Gamma_\alpha$ and empirically verify that it better captures the essence of the system's dynamics. Under the condition that $\mathcal{X}$ is stable and limited, this result is approximately equal to the result of calculating the average value of local $\Gamma_\alpha^{x_t}$ first.
\subsection{Discretization of continuous time}\label{SED}
Another case requiring special handling is time-continuous systems, where Fokker-plank equations describe dynamics \cite{Maoutsa2020} as
\begin{eqnarray}\label{dxdeltax}
    \frac{ \partial p(x_t) }{ \partial t }=-\nabla\cdot \left[f(x_t)p(x_t)-\frac{1}{2}\Sigma\nabla p(x_t)\right],
\end{eqnarray}
or equivalent stochastic differential equations \cite{Evans2012} as 
\begin{eqnarray}\label{dxdeltax}
    \frac{dx_t}{dt}=f(x_t)+\Sigma^\frac{1}{2}\eta_t
\end{eqnarray}
which is also expressed as
$dx_t=f(x_t)dt+\Sigma^\frac{1}{2}dW_t$, where
$x_t\in\mathcal{R}^n$, $\eta_t=dW_t/dt\in\mathcal{R}^n$
and $dW_t\sim\mathcal{N}(0,dtI_n)$. We can approximate the differential of Brownian motion $W_t\in\mathcal{R}^n$ as $dW_t\approx\sqrt{dt}1_n$, $1_n=(1,\cdots,1)_n^T$.

The common method for dealing with this situation is to use finite difference $\Delta x_t$ approximation for differentiation \cite{Ikeda2012} $dx_t$ as 
$\Delta x_t=f(x_t)\Delta t+\Sigma^\frac{1}{2}\Delta W_t\sim\mathcal{N}(f(x_t)\Delta t,\Delta t\Sigma)$. In this way, we can also approximate the stochastic differential equation in the form of a GIS as
\begin{eqnarray}\label{dxdeltax}
    x_{t+\Delta t}\approx x_t+\Delta x_t\sim\mathcal{N}(x_t+f(x_t)\Delta t,\Delta t\Sigma).
\end{eqnarray}
We can calculate CE by setting $A=I_n+\nabla f(x_t)\Delta t$ and covariance $\Delta t\Sigma\in\mathcal{R}^{n\times n}$. Due to the need for discretization in continuous models in machine learning, this method can also provide a reference for the effectiveness of machine learning models. However this method has a big problem in that the $\Delta t$ has a significant impact on the calculation of CE. 

Ao et al. \cite{ao2008emerging,yuan2017sde} presented that through a transformation to a force equation, decomposition of the SDE into three components: potential function $f(x_t)$, dissipative matrix $R(x_t)$, and transverse matrix $T(x_t)$ as
\begin{eqnarray}\label{dxdeltax}
     [R(x_t)+T(x_t)]\frac{dx_t}{dt}=-\nabla f(x_t)+R^\frac{1}{2}(x_t)\eta_t.
\end{eqnarray}
$\nabla f(x_t)$ and $R(x_t)$ happen to be highly correlated in their impact on our matrix $A$ and $\Sigma$. So in the future, we can try to find more effective methods just based on three components to find a CE quantization scheme that does not rely on $\Delta t$.

\subsection{Backward dynamics of GIS}
The third consistent concept is backward dynamics. We can analogy to the approximate dynamical reversibility \cite{Zhang2025} of TPM to derive the approximate dynamical reversibility of GIS determined by Gaussian mapping \cite{Lifshits2013} like 
\begin{equation}
\label{pGaussian}
    p(x_{t+1}|x_t)=\mathcal{N}(Ax_{t}+a_0,\Sigma)\equiv \frac{1}{(2\pi)^\frac{n}{2}\det(\Sigma)^\frac{1}{2}}\exp\left\{-\frac{1}{2}(x_{t+1}-Ax_t-a_0)^T\Sigma^{-1}(x_{t+1}-Ax_t-a_0)\right\},
\end{equation} 
where, $A\in \mathcal{R}^{n\times n}$ and $\Sigma\in\mathcal{R}^{n\times n}$. In a Markov chain, its TPM is a permutation matrix when the dynamics are reversible. This reversibility can be analogized to GIS. If we treat the Gaussian map defined in Equation (\ref{pGaussian}) as a TPM, the state mapping between $x_t$ and $x_{t+1}$ is bijective if the TPM is reversible. Here, we need to define the backward dynamics for GIS in Definition \ref{S_ASA}. 

\begin{definition}  
\label{S_ASA}
(Backward dynamics): For a GIS $x_{t+1}=a_0+Ax_t+\eta_t,\eta_t\sim\mathcal{N}(0,\Sigma)$, also presented as $p(x_{t+1}|x_t)=\mathcal{N}(Ax_t+a_0,\Sigma)$, where $x_t\in\mathcal{R}^n$ to $x_{t+1}\in\mathcal{R}^n$, $A\in \mathcal{R}^{n\times n}$ and $\Sigma\in\mathcal{R}^{n\times n}$, it has a unique backward dynamics 
\begin{eqnarray}\label{BD}
x_t=A^\dagger x_{t+1}-A^\dagger a_0-A^\dagger\eta_t,A^\dagger\eta_t\sim\mathcal{N}(0,A^\dagger\Sigma(A^\dagger)^T).
\end{eqnarray}
$A^\dagger\Sigma(A^\dagger)^T$ is the covariance matrix of backward dynamics as $A^\dagger$ is the Moore-Penrose generalized inverse matrix of $A$.
\end{definition}

In continuous space, if the state mapping between two consecutive time points $x_t$ and $x_{t+1}$ is bijective, $r(A)=n$, $\Sigma=0$ and $A^\dagger\Sigma(A^\dagger)^T=0$ needs to be satisfied, then the probability distribution would be a Dirac distribution \cite{Baldiotti2015}, i.e., all the probability mass is concentrated at a single point as
\begin{eqnarray}\label{ED}
		\begin{aligned}
		p(x_{t+1}|x_t)=\delta(x_{t+1}-Ax_t-a_0)=\begin{cases}
		     \infty, &x_{t+1}=Ax_{t}+a_0  \\
		     0, &x_{t+1}\neq Ax_{t}+a_0
		\end{cases}
   \end{aligned}
\end{eqnarray}
in which $\int_{-\infty}^{+\infty}p(x_{t+1}|x_t)dx_{t+1}^n=1$, $a_0,x_{t},x_{t+1}\in\mathcal{R}^n$, $A,\Sigma\in\mathcal{R}^{n\times n}$. Backward dynamics are the same as $p(x_t|x_{t+1})=\delta(x_t-A^\dagger x_{t+1}+A^\dagger a_0)$. Then, both mappings can be directly written as linear functions $x_{t+1}=Ax_t+a_0$ and $x_t=A^{\dagger}x_{t+1}-A^{\dagger}a_0$, where $A^\dagger=A^{-1}$ in this case.

When $A$ is irreversible or $A^\dagger\Sigma(A^\dagger)^T,\Sigma>0$, the closer $p(x_{t+1}|x_t)$ and $p(x_t|x_{t+1})$ to Dirac distributions and $A$ to a full-rank matrix, the stronger $p(x_{t+1}|x_t)$'s reversibility. Therefore, when the covariance of both forward and backward dynamics satisfy $\Sigma\to 0$, $A^\dagger\Sigma(A^\dagger)^T\to 0$ and $A$ is close to a full-rank matrix, $x_{t+1},x_t$ are close to bijective and $x_{t+1}=a_0+Ax_t+\eta_t\to Ax_t+a_0$, $x_{t}=A^\dagger x_{t+1}-A^\dagger a_0-A^\dagger\eta_t\to A^\dagger x_t-A^\dagger a_0$ are approximate reversible dynamics. From this, we know that we need to find an indicator that includes both $\Sigma$ and $A^\dagger\Sigma(A^\dagger)^T$ to quantify the approximate reversibility of $p(x_{t+1}|x_t)$. In the next subsection, we provide this indicator $\Gamma_\alpha$.

It is not difficult to find that for Gaussian distributions, $\mathcal{N}(0,\sigma^2)$ can be defined for cases where the variance is equal to $\sigma=0$ or $\sigma=\infty$, rather than simply a singularity. $\sigma=0$ means that the Gaussian density function is a Dirac distribution, and $\sigma=\infty$ means the probability density function can be seen as a uniform distribution $\mathcal{U}(-\infty,+\infty)$ with a probability density of 0 everywhere. For this, we can specifically define an inverse matrix for the covariance matrix.

\begin{definition}  
\label{sigma-1}
If $\Sigma$ is a covariance matrix with a normal distribution and $\Sigma=\tilde{V} \Lambda \tilde{V}^T$, is its eigenvalue (singular value) decomposition $\Lambda={\rm diag}(\lambda_1,\cdots, \lambda_n)$, $\tilde{V}=(\tilde{v}_1,...,\tilde{v}_n)$, note that singular values here can be set to 0 or infinity as $\infty\geq\lambda_1\geq\cdots\geq\lambda_n\geq0$. And we define $\kappa_i=1/\lambda_{n-i+1}$, if $\lambda_{n-i+1}=0$ then $\kappa_i=\infty$, and if $\lambda_{n-i+1}=\infty$ then $\kappa_i=0$. We can obtain the singular value matrix of $\Sigma^{-1}$ as $K={\rm diag}(\kappa_1,\cdots, \kappa_n)$, so
\begin{eqnarray}\label{sigma-1eq}
		\begin{aligned}
		 \Sigma^{-1} = VKV^T
   \end{aligned}
\end{eqnarray}
and $V=(v_1,...,v_n)$ as $v_i=\tilde{v}_{n-i+1}$. So, even if the covariance matrix or inverse covariance matrix is not full-rank, we can still perform inverse operations.

\end{definition}

The results obtained in this section are also applicable to SDE and NN, but require special processing.

\section{Causality}
Here, we briefly introduce the theory of causality provided by Hoel et al., which is based on the information-theoretic measure known as Effective Information (EI). Concepts about EI were initially introduced by Tononi et al. in \cite{tononi2003measuring} and later used by Hoel et al. to quantify CE in \cite{Hoel2013, Hoel2016, Hoel2017, Klein2020, Chvykov2020, Comolatti2022,Hoel2024}. For a given transitional probability $P(x_{t+1}|x_t)$,
\begin{equation}
\label{original_EI}
    EI=I(\tilde{x}_{t+1},x_t|do(x_t\sim \mathcal{U}(X))),
\end{equation}
where $x_t, \tilde{x}_{t+1}, \forall t\geq 0$ represent the state variables defined on $\mathcal{S}$ at time step $t$ and $t+1$, respectively. $\mathcal{U}(X)$ means the uniform distribution or maximum entropy distribution on $X$. The intervention is formalized using Judea Pearl's theory \cite{Pearl2009} of causality, particularly through $do(\cdot)$ operations, which means artificially defining the probability distribution space of a variable. Consequently, the distribution of $x_{t+1}$ will be indirectly altered by the intervention on $x_t$ through the causal mechanism $P(x_{t+1}|x_t)$. Therefore, $\tilde{x}_{t+1}$ denotes the random variable representing $x_{t+1}$ following the intervention on $x_t$. To be noticed, Equation (\ref{original_EI}) is only an operational definition and the intervention of $do(\cdot)$ operator is just an imaginary operation to calculate $EI$ but has no real physical meanings. In other words, $do(x_t\sim \mathcal{U}(X))$ is an operator to intervene in the input variable $x_t$ to follow a uniform distribution in its domain $X$ and keep the causal mechanism $P(x_{t+1}|x_t)$ unchanged.

We will introduce the calculation of effective information in different situations in the following subsections.

\subsection{One-dimensional variables}
The simplest way to express it is when the variable is in a one-dimensional space as $A=a,\Sigma=\sigma^2$ and $a,\sigma,x_t\in\mathcal{R}^1$. It is particularly emphasized that the effective information can be expressed as
 \begin{eqnarray}\label{EIDeterminismnon-degeneracy}
    \begin{aligned}
    	EI(a,\sigma^2)
     &=\ln\frac{|a|L}{\sqrt{2\pi e}\sigma},
     \end{aligned}
    \end{eqnarray}
as $do(x_t\sim \mathcal{U}(-L/2, L/2)$, which is the simplest form of effective information for Gaussian iterative systems. 

$EI$ can be decomposed into two terms, determinism $-\left<H(p(x_{t+1}|x_t))\right>$ and non-degeneracy $H(E_D(x_{t+1}))$ and $EI=-\left<H(p(x_{t+1}|x_t))\right>+H(E_D(x_{t+1}))$. Determinism 
 \begin{eqnarray}\label{EIDeterminismnon-degeneracy}
    \begin{aligned}
    	EI_{1}(a,\sigma^2)
     &=\ln\frac{1}{\sqrt{2\pi e}\sigma},
     \end{aligned}
    \end{eqnarray}
and non-degeneracy
 \begin{eqnarray}\label{EIDeterminismnon-degeneracy}
    \begin{aligned}
    	EI_{2}(a,\sigma^2)
     &=\ln |a|L.
     \end{aligned}
    \end{eqnarray}
From the formula, we can directly see that the effective information is positively correlated with the parameter $a$ and negatively correlated with the covariance $\sigma$.
\subsection{Multidimensional variables}

By referring to the one-dimensional calculation method, effective information can be extended to multiple dimensions. Considering the possibility that $A$ may not be a full rank matrix as $r\equiv r(A^{T}\Sigma^{-1}A)\leq r(A)\leq n$.
\begin{eqnarray}\label{ASAEI}
		EI=\ln\left(\frac{|{\rm pdet}(A^T\Sigma^{-1}A)|^\frac{1}{2}L^n}{(2\pi e)^\frac{n}{2}}\right).
	\end{eqnarray}
  \begin{definition}
	\label{thm.Effective-information}
	(Dimension averaged effective information for GIS):
	For GIS like $x_{t+1}=a_0+Ax_t+\eta_t, \eta_t\sim\mathcal{N}(0,\Sigma)$, $a_0,x_t\in\mathcal{R}^n, A,\Sigma\in\mathcal{R}^{n\times n}$ and $r\equiv r(A^{T}\Sigma^{-1}A)\leq r(A)\leq n$, the \textbf{dimension averaged effective information} of the dynamical system is calculated as
\begin{eqnarray}\label{JGaussain-appendix}
		\mathcal{J}=\frac{EI}{n}=\ln\left(\frac{{\rm pdet}(A^T\Sigma^{-1}A)^\frac{1}{2n}L}{\sqrt{2\pi e}}\right).
	\end{eqnarray}
    ${\rm pdet}(\cdot)$ represents the generalized determinant value corresponding to matrix $\cdot$ with rank $r(\cdot)$ and singular values $s_i(\cdot),i=1,\dots,r_{(\cdot)}$, ${\rm pdet}(\cdot)=s_1(\cdot)s_2(\cdot)\dots s_{r(\cdot)}(\cdot)$. $L$ represents the size of the probability space with a uniform distribution determined by the do-operator $do(\cdot)$, which is an intervention that enforces $x_t\sim \mathcal{U}([-L/2,L/2]^n)$, where $\mathcal{U}$ represents uniform distribution.
\end{definition}
\begin{proof}
For the linear stochastic iteration system like $y=Ax+\eta,\eta\sim\mathcal{N}(0,\Sigma)$, $y\in \mathcal{R}^m$ follows a normal distribution about $x\in \mathcal{R}^n$, so $y\sim \mathcal{N}(Ax,\Sigma)$, $A\in\mathcal{R}^{m\times n}$, $\Sigma\in\mathcal{R}^{m\times m}$,
\begin{eqnarray}\label{p(xtp1|doxt)}
		p(y|x)=\displaystyle\frac{1}{\det(\Sigma)^\frac{1}{2}(2\pi)^\frac{m}{2}}\exp\left\{-\frac{1}{2}(y-Ax)^{T}\Sigma^{-1}(y-Ax)\right\}
	\end{eqnarray}
	Following the calculation method of causal geometry, assuming $\eta_t$ is small, $do(x\sim U([-L/2,L/2]^n))\equiv do(x\sim U)$, then $p(x)=1/L^n, L>0$, we can calculate the effect distribution
	\begin{eqnarray}\label{ED}
		\begin{aligned}	E_D(y)&=\displaystyle\int_{\mathcal{X}} p(y|x)p(x)d^n x\\
			&=\displaystyle\int_{\mathcal{X}}  \displaystyle\frac{1}{\det(\Sigma)^\frac{1}{2}(2\pi)^\frac{m}{2}}\exp\left\{-\frac{1}{2}(y-Ax)^{T}\Sigma^{-1}(y-Ax)\right\}\frac{1}{L^n}d^nx\\      &=\displaystyle\int_{\mathcal{X}}\displaystyle\frac{1}{\det(\Sigma)^\frac{1}{2}(2\pi)^\frac{m}{2}}\exp\left\{-\frac{1}{2}(x-A^\dagger y)^{T}A^T\Sigma^{-1}A(x-A^\dagger y)\right\}\frac{1}{L^n}d^nx\\
			&=\frac{(2\pi)^\frac{n-m}{2}}{|{\rm pdet}(A^T\Sigma^{-1}A)|^\frac{1}{2}\det(\Sigma)^\frac{1}{2}}\displaystyle\int_{\mathcal{X}}\displaystyle\frac{1}{|{\rm pdet}(A^{'}\Sigma^{-1} A)|^{-\frac{1}{2}}(2\pi)^\frac{n}{2}}\exp\left\{-\frac{1}{2}(x-A^\dagger y)^{T}A^T\Sigma^{-1}A(x-A^\dagger y)\right\}\frac{1}{L^n}d^nx\\
   &\approx \frac{(2\pi)^\frac{n-m}{2}}{|{\rm pdet}(A^T\Sigma^{-1}A)|^\frac{1}{2}\det(\Sigma)^\frac{1}{2}L^n}
		\end{aligned}
	\end{eqnarray}
	According to the properties of information entropy, determinism information	\begin{eqnarray}\label{Determinism}
		\begin{aligned}
			-\left<H(p(y|x))\right>&=-\left<-\int_{\mathcal{X}}p(y|x)\ln(p(y|x))d^my\right>\\
			&=\displaystyle\int_{\mathcal{X}}\frac{d^n x}{L^n}\displaystyle\int_{\mathcal{X}}  \displaystyle\frac{1}{\det(\Sigma)^\frac{1}{2}(2\pi)^\frac{m}{2}}\exp\left\{-\frac{1}{2}(x_{t+1}-Ax_t)^{T}\Sigma^{-1}(x_{t+1}-Ax_t)\right\}\\&\left[-\ln\left(\det(\Sigma)^\frac{1}{2}(2\pi)^\frac{m}{2}\right)-\frac{1}{2}(y-Ax)^{T}\Sigma^{-1}(y-Ax)\right]d^my\\
			&=\left[-\ln\left(\det(\Sigma)^\frac{1}{2}(2\pi)^\frac{m}{2}\right)-\frac{m}{2}\right]\\
			&=\ln\frac{1}{\det(\Sigma)^\frac{1}{2}(2\pi e)^\frac{m}{2}},
		\end{aligned}
	\end{eqnarray}
	degeneracy information	\begin{eqnarray}\label{Degeneracy}
		\begin{aligned}	-H(E_D(y))&=\int_{\mathcal{X}}E_D(y)\ln(E_D(y))d^my\\
			&=\displaystyle\int_{\mathcal{X}}E_D(y)\ln\left(\frac{(2\pi)^\frac{n-m}{2}}{|{\rm pdet}(A^T\Sigma^{-1}A)|^\frac{1}{2}\det(\Sigma)^\frac{1}{2}L^n}\right)\\&\approx\ln\left(\frac{(2\pi)^\frac{n-m}{2}}{|{\rm pdet}(A^T\Sigma^{-1}A)|^\frac{1}{2}\det(\Sigma)^\frac{1}{2}L^n}\right).
		\end{aligned}
	\end{eqnarray}
	It's obvious that	\begin{eqnarray}\label{DegeneracyDeterminismEI}
		EI=-\left<H(p(x_{t+1}|x_t))\right>-(-H(E_D(x_{t+1})))=\ln\left(\frac{|{\rm pdet}(A^T\Sigma^{-1}A)|^\frac{1}{2}L^n}{(2\pi)^\frac{n}{2}e^\frac{m}{2}}\right).
	\end{eqnarray}
When $m=n$, we can directly calculate dimension averaged effective information for GIS in Eq.(\ref{JGaussain-appendix}).
\end{proof}
 According to the properties of information entropy, $\mathcal{J}$ can be decomposed into two terms as $\mathcal{J}=\mathcal{J}_1+\mathcal{J}_2$, determinism 
\begin{eqnarray}\label{DegeneracyDeterminismEI1}
		\mathcal{J}_1=-\ln\left(\sqrt{2\pi e}\det(\Sigma)^\frac{1}{2n}\right)
	\end{eqnarray}
measures how the current state $x_t$ can deterministically (sufficiently) influence the state $x_{t+1}$ in the future, non-degeneracy
\begin{eqnarray}\label{DegeneracyDeterminismEI2}
		\mathcal{J}_2=\ln\left({\rm pdet}(A^T\Sigma^{-1} A)^\frac{1}{2n}\det(\Sigma)^\frac{1}{2n}L\right),
	\end{eqnarray}
measures how exactly we can infer (necessarily) the state $x_{t-1}$ in the past from the current state $x_t$.

Meanwhile, if matrix $A$ is not of rank, we can directly use the dimension of matrix rank for averaging on $EI$. $EI$ on a special macro-state dynamics can be obtained when the micro-dynamics matrix $A$ is not full-rank in the following corollary.

\begin{corollary}
For the linear stochastic iteration systems like Equation (\ref{micro-dynamics-appendix}), maximum effective information of the system after coarse-graining  $y_t=\phi(x_t)=Wx_t$, $W\in \mathcal{R}^{r\times n}$, is calculated as 
\begin{eqnarray}\label{DegeneracyDeterminismEI}
		\mathcal{J}_R=\frac{EI_R}{r}=\ln\left(\frac{|{\rm pdet}(A^T\Sigma^{-1}A)|^\frac{1}{2r}L}{(2\pi e)}\right).
	\end{eqnarray}
\end{corollary}

\section{Dynamical reversibility}
Another important property of dynamics that people are concerned about is dynamical reversibility. It is generally believed that the dynamics of the microscopic physical world (such as Newton's second law or the Schrödinger equation in quantum mechanics) are dynamically reversible, while the second law of thermodynamics in the macroscopic world stems from the irreversibility of macroscopic dynamic processes. 

\subsection{Approximate dynamical reversibility of TPM}
Zhang's previous work \cite{Zhang2025} has provided an indicator of CE for Markov chains called approximate dynamical reversibility, which describes the proximity of a transitional probability matrix(TPM) to a permutation matrix. This new framework for CE is distinctive for its independence from coarse-grained strategies. The calculation of the approximate dynamical reversibility of TPM $P$ requires solving its singular values as $P\cdot P^T=VZV^T$ and $Z ={\rm diag} (\zeta_1,\zeta_2,\cdots,\zeta_N)$ as $\zeta_1\geq \zeta_2\geq\cdots\geq\zeta_N\geq 0$, then the $\alpha$-ordered approximate dynamical reversibility of $P$ is defined as:
    \begin{equation}
        \label{eqn.formal_def_gamma}
        \Gamma_{\alpha} \equiv \sum_i^N \zeta_i^{\alpha},
    \end{equation}
    where $\alpha\in (0,2)$ is a parameter. The maximum $\Gamma_{\alpha}$ can be achieved if $P$ is a permutation matrix. Therefore, $\Gamma_{\alpha}$ can be used as an index of the reversibility of TPM and an approximate relationship exists as\cite{Zhang2025} 
\begin{equation}
\label{eq:EI_Gamma_Approx}
    EI\sim \log\Gamma_{\alpha}.
\end{equation}
However, this approximate equivalence can only be confirmed through numerical experiments, and there is currently no clear mathematical analytical proof.

\subsection{Approximate dynamical reversibility of GIS in $\mathcal{R}^1$}
In Eq.(\ref{micro-dynamics-appendix}), when $\eta_t=0$ and $f(x_t)=a_0+Ax_t$ is a reversible bijection mapping on $\mathcal{R}^{n}$, we assume that the dynamics of GIS described by Eq.(\ref{micro-dynamics-appendix}) is strictly reversible as $x_{t+1}=a_0+Ax_t$, and the corresponding backward dynamics of Eq.(\ref{micro-dynamics-appendix}) can be directly established as $x_{t}=A^{-1}(x_{t+1}-a_0)$. However, real-world dynamical systems around us are not so simple and perfect. 

Similar to the framework in Zhang's article \cite{Zhang2025}, to calculate approximate reversibility, it is necessary to obtain the singular value spectrum of TPM, and the same applies to GIS. For GIS, we need to treat $p(x_{t+1}|x_t)$ as a TPM with continuous states and apply operations from functional analysis \cite{Lax2014} to calculate its singular value spectrum. Firstly, we calculate the Gaussian kernel which corresponds to $P\cdot P^T$ of TPM in \cite{Zhang2025} when the probability space is continuous as
\begin{eqnarray}\label{ED}
		\begin{aligned}
  K(\textbf{x},\textbf{y})\equiv\mathcal{K}(\textbf{x}-\textbf{y})=(2\pi)^{-\frac{n}{2}}\det(2\Sigma)^{-\frac{1}{2}}\exp\left\{-\frac{1}{4}(\textbf{x}-\textbf{y})^T(A^T\Sigma^{-1}A)(\textbf{x}-\textbf{y})\right\}.
   \end{aligned}
\end{eqnarray}
Suppose $\zeta$ is a singular value of $p(x_{t+1}|x_t)$, then $\zeta^2$ is the eigenvalue of $K(\textbf{x},\textbf{y})=\mathcal{K}(\textbf{x}-\textbf{y})$ and $\psi(\textbf{x})$ is the eigenfunction corresponding to the eigenvalue $\zeta^2$. By analogy with the theorem of the spectrum of stochastic integral operators in functional analysis \cite{Lax2014}, we can derive $\zeta^2$ by $(\textbf{K}\psi)(\textbf{x})=\int_{-\infty}^\infty K(\textbf{x},\textbf{y})\psi(\textbf{y})d\textbf{y}=\zeta^2 \psi(\textbf{x})$.
Based on the properties of the Fourier transform of the convolution function in the integral transform, we can consider the above integral in $(\textbf{K}\psi)(\textbf{x})$ as a convolution of $\mathcal{K}(\textbf{y})$ and $\psi(\textbf{y})$. Then the eigenvalue spectrum $\zeta^2(\omega)$ in the frequency space is obtained through Fourier transform \cite{Bracewell1989} as
$\hat{\mathcal{K}}(\omega)\hat{\psi}(\omega)=\mathcal{F}\left\{\mathcal{K}(\mathbf{y})\right\}\mathcal{F}\left\{\psi(\mathbf{y})\right\}=\zeta^2\hat{\psi}(\omega).$ From this, we can also obtain the singular value spectrum $\zeta(\omega)=\sqrt{\hat{\mathcal{K}}(\omega)}$ of $p(x_{t+1}|x_t)$ in the frequency space as $\omega\in\mathcal{R}^n$.

\begin{definition}\label{dfn.Gamma_alpha}
(Approximate reversibility): Suppose the singular value spectrum of GIS $p(x_{t+1}|x_t)=\mathcal{N}(Ax_t+a_0,\Sigma)$ is $\zeta(\omega)$, $a_0,\omega,x_t,x_{t+1}\in\mathcal{R}^n$, $A,\Sigma\in\mathcal{R}^{n\times n}$, then the $\alpha$-ordered approximate dynamical reversibility of $p(x_{t+1}|x_t)$ is defined as
\begin{eqnarray}\label{ED}
		\begin{aligned}
		\Gamma_\alpha&\equiv\int_{-\infty}^\infty\zeta^\alpha(\omega)d\omega
   \end{aligned}
\end{eqnarray}
where $\alpha\in(0,2)$.
\end{definition}
According to Definition \ref{dfn.Gamma_alpha}, the approximate dynamical reversibility $\Gamma_\alpha$ of GIS $x_{t+1}=a_0+Ax_t+\eta_t,\eta_t\sim\mathcal{N}(0,\Sigma)$ is given in following theorems.  We can start the deduction process from situations of GIS in $\mathcal{R}^1$.
\begin{lem}
For Gaussian kernels like
\begin{eqnarray}\label{ED}
		\begin{aligned}
			K(\textbf{x},\textbf{y})=\frac{1}{\sqrt{2\pi}\sigma}\exp\left\{-\frac{(\textbf{y}-\textbf{x})^2}{2\sigma^2}\right\}
   \end{aligned}
\end{eqnarray}
Based on the structure of a Gaussian kernel $K(\textbf{x},\textbf{y})=\mathcal{K}(\textbf{x}-\textbf{y})$, we can perform a convolution operation on the integral, which allows us to perform a Fourier transform on the integral as directly
\begin{eqnarray}\label{ED}
		\begin{aligned}
			\hat{\mathcal{K}}(\omega)\hat{\psi}(\omega)=\lambda\hat{\phi}(\omega)
   \end{aligned}
\end{eqnarray}
In frequency space, the eigenvalue spectrum can be directly expressed as
\begin{eqnarray}\label{ED}
		\begin{aligned}
			\lambda(\omega)=\hat{K}(\omega)=\exp\left\{-\frac{\omega^2\sigma^2}{2}\right\}.
   \end{aligned}
\end{eqnarray}
\end{lem}

\begin{proof}
We can perform a Fourier transform on the integral as
\begin{eqnarray}\label{ED}
		\begin{aligned}
			\hat{K}(\omega)\hat{\psi}(\omega)&=\mathcal{F}\left\{(\textbf{K}\psi)(\textbf{x})\right\}=\mathcal{F}\left\{\int_{-\infty}^\infty K(\textbf{x},\textbf{y})\psi(\textbf{y})d\textbf{y}\right\}=\mathcal{F}\left\{\int_{-\infty}^\infty \frac{1}{\sqrt{2\pi}\sigma}\exp\left\{-\frac{(\textbf{x}-\textbf{y})^2}{2\sigma^2}\right\}\psi(\textbf{y})d\textbf{y}\right\}\\
           &=\mathcal{F}\left\{\frac{1}{\sqrt{2\pi}\sigma}\exp\left\{-\frac{\textbf{y}^2}{2\sigma^2}\right\}\right\}\mathcal{F}\left\{\psi(\textbf{y})\right\}=\exp\left\{-\frac{\omega^2\sigma^2}{2}\right\}\hat{\psi}(\omega).
   \end{aligned}
\end{eqnarray} 
This calculation process is actually a Fourier transform of Gaussian functions, and it is almost the same as the process of finding the characteristic function for Gaussian distribution.
\end{proof}
Based on the solution of the eigenvalues of the Gaussian kernel function $K(\textbf{x},\textbf{y})=\mathcal{K}(\textbf{x}-\textbf{y})$ mentioned above, we can obtain the singular value of TPM $p(y|x)$.
\begin{thm}
    \label{dfn.gamma}
    Suppose the transitional probability from $x\in\mathcal{R}$ to $y\in\mathcal{R}$ is
    \begin{eqnarray}\label{ED}
		\begin{aligned}
			p(y|x)=\frac{1}{\sqrt{2\pi}\sigma}\exp\left\{-\frac{(y-ax)^2}{2\sigma^2}\right\}
		\end{aligned}
	\end{eqnarray}
    its singular value spectrum is:
    \begin{eqnarray}\label{ED}
		\begin{aligned}
		\zeta^2(\omega)=\frac{1}{|a|}\exp\left\{-\frac{\omega^2}{2}\frac{2\sigma^2}{a^2}\right\}
   \end{aligned}
\end{eqnarray}
then the $2$-ordered approximate dynamical reversibility of $p(y|x)$ is defined as:
\begin{eqnarray}\label{ED}
		\begin{aligned}
		\Gamma_2&=\int_{-\infty}^\infty\zeta^2(\omega)d\omega=\frac{\sqrt{\pi}}{\sigma}
   \end{aligned}
\end{eqnarray}
\end{thm}
\begin{proof}
According to the calculation process of calculating the singular values for discrete TPM $P$ is equal to calculating the eigenvalues of $PP^T$. In continuous space, $PP^T$ for $p(y|x)$ is equal to the Gaussian kernel
\begin{eqnarray}\label{ED}
		\begin{aligned}
			K(\textbf{x},\textbf{y})&=\int_{-\infty}^\infty p(\textbf{z}|\textbf{x})p(\textbf{z}|\textbf{y})d\textbf{z}=\int_{-\infty}^\infty \frac{1}{2\pi\sigma^2}\exp\left\{-\frac{(\textbf{z}-a\textbf{x})^2+(\textbf{z}-a\textbf{y})^2}{2\sigma^2}\right\}d\textbf{z}\\
   &=\int_{-\infty}^\infty \frac{1}{2\pi\sigma^2}\exp\left\{-\frac{2(\textbf{z}-\frac{a(\textbf{x}+\textbf{y})}{2})^2+\frac{a^2(\textbf{x}-\textbf{y})^2}{2}}{2\sigma^2}\right\}d\textbf{z}\\
   &=\frac{1}{\sqrt{2\pi}\sigma}\exp\left\{-\frac{a^2(\textbf{x}-\textbf{y})^2}{4\sigma^2}\right\}\int_{-\infty}^\infty\frac{1}{\sqrt{2\pi}\sigma}\exp\left\{-\frac{(\textbf{z}-\frac{a(\textbf{x}+\textbf{y})}{2})^2}{\sigma^2}\right\}d\textbf{z}\\
   &=\frac{1}{\sqrt{2\pi}(\sqrt{2}\sigma)}\exp\left\{-\frac{a^2(\textbf{x}-\textbf{y})^2}{4\sigma^2}\right\}\int_{-\infty}^\infty\frac{1}{\sqrt{2\pi}(\sigma/\sqrt{2})}\exp\left\{-\frac{(\textbf{z}-\frac{a(\textbf{x}+\textbf{y})}{2})^2}{2(\sigma/\sqrt{2})^2}\right\}d\textbf{z}\\
   &=\frac{1}{|a|}\frac{1}{\sqrt{2\pi}(\sqrt{2}\sigma/|a|)}\exp\left\{-\frac{(\textbf{x}-\textbf{y})^2}{2(\sqrt{2}\sigma/|a|)^2}\right\}.
   \end{aligned}
\end{eqnarray}

It's also important to be reminded that this is only a special case of linear dynamics. In nonlinear dynamics, we can only get
\begin{eqnarray}\label{ED}
		\begin{aligned}
			K(\textbf{x},\textbf{y})&=\frac{1}{\sqrt{2\pi}(\sqrt{2}\sigma)}\exp\left\{-\frac{(f(\textbf{x})-f(\textbf{y}))^2}{2(\sqrt{2}\sigma)^2}\right\}.
   \end{aligned}
\end{eqnarray}

Suppose $\zeta$ is a singular value of $p(y|x)$, then $\zeta^2$ is the eigenvalue of $K(\textbf{x},\textbf{y})=\mathcal{K}(\textbf{x}-\textbf{y})$ and $\psi(\textbf{x})$ is the eigenfunction corresponding to the eigenvalue $\zeta^2$. By analogy with the theorem of the spectrum of stochastic integral operators in functional analysis, we can derive $\zeta^2$ by
\begin{eqnarray}\label{ED}
		\begin{aligned}
			(\textbf{K}\psi)(x)=\zeta^2 \psi(x),
   \end{aligned}
\end{eqnarray}
from which we can get
\begin{eqnarray}\label{ED}
		\begin{aligned}
		\zeta^2(\omega)=\frac{1}{|a|}\exp\left\{-\frac{\omega^2}{2}\frac{2\sigma^2}{a^2}\right\}.
   \end{aligned}
\end{eqnarray}
We can see that the calculation process is actually a Fourier transform of Gaussian functions. With integral calculation, we get $\Gamma_\alpha$ when $\alpha=2$ as
\begin{eqnarray}\label{ED}
		\begin{aligned}
		\Gamma_2&=\int_{-\infty}^\infty\zeta^2(\omega)d\omega=\int_{-\infty}^\infty\frac{1}{|a|}\exp\left\{-\frac{\omega^2}{2}\frac{2\sigma^2}{a^2}\right\}d\omega\\
  &=\frac{1}{|a|}\frac{|a|}{\sqrt{2}\sigma}\int_{-\infty}^\infty\exp\left\{-\frac{\omega^2}{2}\frac{2\sigma^2}{a^2}\right\}d\left(\frac{\sqrt{2}\sigma}{|a|}\omega\right)\\
  &=\frac{\sqrt{2\pi}}{\sqrt{2}\sigma}=\frac{\sqrt{\pi}}{\sigma}.
   \end{aligned}
\end{eqnarray}
Then we can derive the general value of approximate dynamical reversibility $\Gamma_\alpha$ according to $\Gamma_2$.
\end{proof}
After knowing the quantitative index of approximate reversibility in the case of $\alpha=2$, we can extend it to the general case as $\alpha\in(0,2)$.
\begin{thm}
    \label{dfn.gamma}
    Suppose the transitional probability from $x\in\mathcal{R}$ to $y\in\mathcal{R}$ is
    \begin{eqnarray}\label{ED}
		\begin{aligned}
			p(y|x)=\frac{1}{\sqrt{2\pi}\sigma}\exp\left\{-\frac{(y-ax)^2}{2\sigma^2}\right\}
		\end{aligned}
	\end{eqnarray}
    its $\alpha$-ordered singular value spectrum is:
    \begin{eqnarray}\label{ED}
		\begin{aligned}
		\zeta^\alpha(\omega)=|a|^{-\frac{\alpha}{2}}\exp\left\{-\frac{\alpha}{2}\left(\frac{\omega\sigma}{a}\right)^2\right\}
   \end{aligned}
\end{eqnarray}
then the $\alpha$-ordered approximate dynamical reversibility of $p(y|x)$ is defined as:
\begin{eqnarray}\label{ED}
		\begin{aligned}
		\Gamma_\alpha&=\int_{-\infty}^\infty\zeta^\alpha(\omega)d\omega=\sqrt{\frac{2\pi}{\alpha}}\frac{|a|^{1-\frac{\alpha}{2}}}{\sigma}
   \end{aligned}
\end{eqnarray}
\end{thm}
\begin{proof}
The $\alpha$-ordered approximate dynamical reversibility of $p(y|x)$ is defined as:
\begin{eqnarray}\label{ED}
		\begin{aligned}
		\zeta^\alpha(\omega)=|a|^{-\frac{\alpha}{2}}\exp\left\{-\frac{\alpha}{2}\left(\frac{\omega\sigma}{a}\right)^2\right\}.
   \end{aligned}
\end{eqnarray}
So we can get the index of approximate dynamical reversibility as
\begin{eqnarray}\label{ED}
		\begin{aligned}
		\Gamma_\alpha&=\int_{-\infty}^\infty\zeta^\alpha(\omega)d\omega=\int_{-\infty}^\infty |a|^{-\frac{\alpha}{2}}\exp\left\{-\frac{\omega^2}{2}\frac{\alpha\sigma^2}{a^2}\right\}d\omega\\
  &=\frac{|a|^{1-\frac{\alpha}{2}}}{\sqrt{\alpha}\sigma}\int_{-\infty}^\infty\exp\left\{-\frac{\omega^2}{2}\frac{k\sigma^2}{a^2}\right\}d\left(\frac{\sqrt{\alpha}\sigma}{|a|}\omega\right)\\
  &=\frac{\sqrt{2\pi}}{\sqrt{\alpha}\sigma}|a|^{1-\frac{\alpha}{2}}=\sqrt{\frac{2\pi}{\alpha}}\frac{|a|^{1-\frac{\alpha}{2}}}{\sigma}.
   \end{aligned}
\end{eqnarray}
This is the indicator expression for approximate reversibility in one-dimensional situations.
\end{proof}
Fig.\ref{basicdatatest}b shows the logarithmic relationship between $EI$ and $\Gamma_\alpha$. It can be seen that for different noise covariances $\sigma$, as $\Gamma_\alpha$ increases, $EI$ shows an approximate logarithmic growth relationship with $\Gamma_\alpha$ as $\alpha=1$. The dimensional average reversible information $\ln\Gamma_\alpha$ and dimensional average effective information $\mathcal{J}$ (Fig.\ref{basicdatatest}c), as well as their corresponding determinism (Fig.\ref{basicdatatest}d) and non-degeneracy (Fig.\ref{basicdatatest}e), are linearly related.
\subsection{Approximate dynamical reversibility of GIS in $\mathcal{R}^n$}
After knowing the approximate reversibility index of random variables in one-dimensional space, we can extend it to multi-dimensional variables as $x\in\mathcal{R}^n$ and $y\in\mathcal{R}^n$. First, we can study the case of the full rank matrix $A$ and $\Sigma$.
\begin{thm}
    \label{dfn.gamma}
    Suppose the transitional probability from $x\in\mathcal{R}^n$ to $y\in\mathcal{R}^n$ is
    \begin{eqnarray}\label{ED}
		\begin{aligned}
			p(y|x)=\frac{1}{(2\pi)^\frac{n}{2}\det(\Sigma)^\frac{1}{2}}\exp\left\{-\frac{1}{2}(y-Ax)^T\Sigma^{-1}(y-Ax)\right\}
		\end{aligned}
	\end{eqnarray}
 $x,y\in\mathcal{R}^n, A\in \mathcal{R}^{n\times n}$, its $\alpha$-ordered singular value spectrum is:
    \begin{eqnarray}\label{zetafullarank}
		\begin{aligned}
		\zeta^\alpha(\omega)=|\det(A)|^{-\frac{\alpha}{2}}\exp\left\{-\frac{\alpha}{2}\left(\omega^T(A^{-1})^T\Sigma A^{-1}\omega\right)\right\}
   \end{aligned}
\end{eqnarray}
   $\omega\in\mathcal{R}^n$, the $\alpha$-ordered approximate dynamical reversibility of $p(y|x)$ is defined as:
 \begin{eqnarray}\label{Gammafullrank}
		\begin{aligned}
		\Gamma_\alpha=\left(\frac{2\pi}{\alpha}\right)^\frac{n}{2}{\rm det}(A^T\Sigma^{-1}A)^{\frac{1}{2}-\frac{\alpha}{4}}{\rm det}(\Sigma^{-1})^\frac{\alpha}{4}
   \end{aligned}
\end{eqnarray}
Besides the first constant, this formula contains two terms: ${\rm det}(\cdot)$ of the inverse covariance matrices of the forward and backward dynamics.
\end{thm}
\begin{proof}
Suppose $y=Ax+\eta,\eta\sim\mathcal{N}(0,\Sigma)$ as the transitional probability from $x\in\mathcal{R}^m$ to $y\in\mathcal{R}^n$ is 
\begin{eqnarray}\label{ED}
		\begin{aligned}
			p(y|x)=\frac{1}{(2\pi)^\frac{n}{2}\det(\Sigma)^\frac{1}{2}}\exp\left\{-\frac{1}{2}(y-Ax)^T\Sigma^{-1}(y-Ax)\right\}.
		\end{aligned}
	\end{eqnarray}
We calculate the Gaussian kernel, which corresponds to $P\cdot P^T$ of TPM when the probability space is continuous, as
 \begin{eqnarray}\label{ED}
		\begin{aligned}
			K(\textbf{y},\textbf{x})=\frac{1}{|\det(A)|}\frac{1}{(2\pi)^\frac{n}{2}\det(2(A^{-1})^T\Sigma A^{-1})^\frac{1}{2}}\exp\left\{-\frac{1}{2}(\textbf{y}-\textbf{x})^TA^T\Sigma^{-1}A(\textbf{y}-\textbf{x})\right\}
		\end{aligned}
	\end{eqnarray}
According to $\hat{\mathcal{K}}(\omega)\hat{\psi}(\omega)=\mathcal{F}\left\{\mathcal{K}(\mathbf{y})\right\}\mathcal{F}\left\{\psi(\mathbf{y})\right\}=\zeta^2\hat{\psi}(\omega)$  we can get Eq.(\ref{zetafullarank}). We can see that the calculation process is actually a Fourier transform of Gaussian functions, so both low dimensional and high-dimensional can directly borrow the form of characteristic functions. When $A$ is full rank, the $\alpha$-ordered approximate dynamical reversibility of $p(y|x)$ as
 \begin{eqnarray}\label{Gammaalphaprocess}
		\begin{aligned}
		\Gamma_\alpha&= \int_{\mathcal{X}}\zeta^\alpha(\omega)d\omega= \int_{\mathcal{X}}|\det(A)|^{-\alpha/2} \exp\left\{ -\frac{\alpha}{2} \omega^T A^T\Sigma^{-1}A \omega \right\} d\omega\\
&= |\det(A)|^{-\alpha/2} \int_{\mathcal{X}} \exp\left\{ -\frac12 \omega^T (\alpha A^T\Sigma^{-1}A) \omega \right\}d\omega\\&= |\det (A)|^{-\alpha/2} \cdot \frac{(2\pi)^{n/2}}{\alpha^{n/2} \cdot \dfrac{(\det \Sigma)^{1/2}}{|\det(A)|}} = |\det(A)|^{-\alpha/2} \cdot \frac{(2\pi)^{n/2} \, |\det(A)|}{\alpha^{n/2} \, (\det \Sigma)^{1/2}} \\
&=\left(\frac{2\pi}{\alpha}\right)^\frac{n}{2}\frac{|\det(A)|^{1-\frac{\alpha}{2}}}{\det(\Sigma)^\frac{1}{2}},
   \end{aligned}
\end{eqnarray}
which can be transformed into the form of Eq.(\ref{Gammafullrank}). 
\end{proof}
When we assume $A$ and $\Sigma$ are reversible matrices, the forward dynamics is $x_{t+1}=a_0+Ax_t+\eta_t$, and the corresponding backward dynamics is $x_{t}=A^{-1}(x_{t+1}-a_0-\eta_t)$. The corresponding conditional probability distributions are $p(x_{t+1}|x_t)=\mathcal{N}(a_0+Ax_t,\Sigma)$ and $p(x_{t}|x_{t+1})=\mathcal{N}(A^{-1}(x_{t+1}-a_0),A^{-1}\Sigma (A^{-1})^T)$, it can be seen that $\Sigma^{-1}$ and $A^T\Sigma^{-1}A$ represent the inverse covariance matrices of the forward and backward dynamics, respectively. 

We can also deal with the case of a non-full rank matrix $A$, and this method can better reflect the essence of approximate reversibility $\Gamma_\alpha$.
\begin{thm}
    \label{dfn.gamma}
    Suppose the transitional probability from $x\in\mathcal{R}^m$ to $y\in\mathcal{R}^n$ is
    \begin{eqnarray}\label{ED}
		\begin{aligned}
			p(y|x)=\frac{1}{(2\pi)^\frac{n}{2}\det(\Sigma)^\frac{1}{2}}\exp\left\{-\frac{1}{2}(y-Ax)^T\Sigma^{-1}(y-Ax)\right\}
		\end{aligned}
	\end{eqnarray}
 $x\in\mathcal{R}^m, y\in\mathcal{R}^n, A\in \mathcal{R}^{n\times m},\Sigma\in \mathcal{R}^{n\times n}$, its $\alpha$-ordered singular value spectrum is:
    \begin{eqnarray}\label{ED}
		\begin{aligned}
		\zeta^\alpha(\omega)=\left\{(2\pi)^{\frac{n-m}{2}}{\rm det}(\Sigma)^{\frac{1}{2}}{\rm pdet}(A^T\Sigma^{-\frac{1}{2}} A)^{\frac{1}{2}}\right\}^{-\frac{\alpha}{2}}\exp\left\{-\frac{\alpha}{2}\left(\omega^T(A^T\Sigma^{-1}A)^{\dagger}\omega\right)\right\}
   \end{aligned}
\end{eqnarray}
   $\omega\in\mathcal{R}^m $, the $\alpha$-ordered approximate dynamical reversibility of $p(y|x)$ is defined as:
 \begin{eqnarray}\label{gammagaussian}
		\begin{aligned}
		\Gamma_\alpha&=(2\pi)^{\frac{\alpha(m-n)}{4}}\left(\frac{2\pi}{\alpha}\right)^\frac{m}{2}{\rm pdet}(A^T\Sigma^{-1}A)^{\frac{1}{2}-\frac{\alpha}{4}}{\rm det}(\Sigma^{-1})^\frac{\alpha}{4}.\\
   \end{aligned}
\end{eqnarray}
After that, we can get the indicator for local approximate dynamical reversibility in a general GIS defined as
 \begin{eqnarray}\label{Gamma_gaussian}
		\begin{aligned}
		\Gamma_\alpha=\left(\frac{2\pi}{\alpha}\right)^\frac{n}{2}{\rm pdet}(A^T\Sigma^{-1}A)^{\frac{1}{2}-\frac{\alpha}{4}}{\rm pdet}(\Sigma^{-1})^\frac{\alpha}{4},
   \end{aligned}
\end{eqnarray} 
Besides the first constant, this formula contains two terms: ${\rm pdet}(\cdot)$ of the inverse covariance matrices of the forward and backward dynamics.
\end{thm}
\begin{proof}
Suppose $y=Ax+\eta,\eta\sim\mathcal{N}(0,\Sigma)$ as the transitional probability from $x\in\mathcal{R}^m$ to $y\in\mathcal{R}^n$ is 
\begin{eqnarray}\label{ED}
		\begin{aligned}
			p(y|x)=\frac{1}{(2\pi)^\frac{n}{2}\det(\Sigma)^\frac{1}{2}}\exp\left\{-\frac{1}{2}(y-Ax)^T\Sigma^{-1}(y-Ax)\right\}.
		\end{aligned}
	\end{eqnarray}
When $A$ is not a full rank matrix, we need to calculate the Gaussian kernel which corresponds to $P\cdot P^T$ of TPM when probability space is continuous as
\begin{eqnarray}\label{ED}
		\begin{aligned}
			K(\textbf{x},\textbf{y})&=\int_{-\infty}^\infty p(\textbf{z}|\textbf{x})p(\textbf{z}|\textbf{y})d\textbf{z}=\frac{1}{(2\pi)^\frac{n}{2}\det(2\Sigma)^\frac{1}{2}}\exp\left\{-\frac{1}{4}(\textbf{x}-\textbf{y})^T(A^T\Sigma^{-1}A)(\textbf{x}-\textbf{y})\right\}\\&=\frac{{\rm pdet}(2(A^T\Sigma^{-1}A)^\dagger)^\frac{1}{2}}{(2\pi)^\frac{n}{2}\det(2\Sigma)^\frac{1}{2}{\rm pdet}(2(A^T\Sigma^{-1}A)^\dagger)^\frac{1}{2}}\exp\left\{-\frac{1}{4}(\textbf{x}-\textbf{y})^T(A^T\Sigma^{-1}A)(\textbf{x}-\textbf{y})\right\}\\&=\frac{{\rm pdet}(A^T\Sigma^{-1}A)^{-\frac{1}{2}}}{(2\pi)^\frac{n-m}{2}\det(2\Sigma)^\frac{1}{2}}\frac{1}{(2\pi)^\frac{m}{2}{\rm pdet}(2(A^T\Sigma^{-1}A)^\dagger)^\frac{1}{2}}\exp\left\{-\frac{1}{4}(\textbf{x}-\textbf{y})^T(A^T\Sigma^{-1}A)(\textbf{x}-\textbf{y})\right\},
   \end{aligned}
\end{eqnarray}
By analogy with the theorem of the spectrum of stochastic integral operators in functional analysis, we can derive $\zeta^2$ by
\begin{eqnarray}\label{ED}
		\begin{aligned}
			\textbf{K}\psi(x)=\zeta^2 \psi(x)
   \end{aligned}
\end{eqnarray}
and we can get
\begin{eqnarray}\label{ED}
		\begin{aligned}
		\zeta^\alpha(\omega)=\left\{(2\pi)^{\frac{n-m}{2}}{\rm det}(\Sigma)^{\frac{1}{2}}{\rm pdet}(A^T\Sigma^{-\frac{1}{2}} A)^{\frac{1}{2}}\right\}^{-\frac{\alpha}{2}}\exp\left\{-\frac{\alpha}{2}\left(\omega^T(A^T\Sigma^{-1}A)^{\dagger}\omega\right)\right\}
   \end{aligned}
\end{eqnarray}
the $\alpha$-ordered approximate dynamical reversibility of $p(y|x)$ is derived from the third line of Eq.(\ref{Gammaalphaprocess}), which changes to
\begin{eqnarray}\label{ED}
		\begin{aligned}
		\Gamma_\alpha&=\left(\frac{2\pi}{\alpha}\right)^\frac{n}{2}\left(\frac{{\rm pdet}(A^T\Sigma^{-1}A)^{-\frac{1}{2}}}{(2\pi)^\frac{n-m}{2}\det(\Sigma)^\frac{1}{2}}\right)^\frac{\alpha}{2}{\rm pdet}(A^T\Sigma^{-1}A)\\&=(2\pi)^{\frac{\alpha(m-n)}{4}}\left(\frac{2\pi}{\alpha}\right)^\frac{m}{2}{\rm pdet}(A^T\Sigma^{-1}A)^{\frac{1}{2}-\frac{\alpha}{4}}{\rm det}(\Sigma^{-1})^\frac{\alpha}{4}.
   \end{aligned}
\end{eqnarray}
In general dynamics, $m=n$. Considering the situation in Definition \ref{sigma-1}, we can directly use the form in Eq.(\ref{Gamma_gaussian})
to measure the approximate dynamical reversibility of the dynamics.
\end{proof}
$\Gamma_\alpha$ can be considered a measure of approximate dynamical reversibility because when the uncertainty in the system's forward and backward dynamics (i.e., the pseudo-determinant of the covariance matrix) is large, the pseudo-determinants of $\Sigma^{-1}$ and $A^T\Sigma^{-1}A$ become very small, and the dynamics become highly irreversible. Conversely, if the uncertainty in the system's forward and backward dynamics is very small, or even approaches zero, the pseudo-determinants of $\Sigma^{-1}$ and $A^T\Sigma^{-1}A$ become very large, or even approach infinity. In this case, the system's dynamics are very close to reversible dynamics, and using the de-noised reversible dynamics to approximate the GIS becomes effective.

To avoid the excessive influence of variable dimension on the reversibility of dynamics, we can use the dimensional average reversible information $\hat\gamma_\alpha\equiv\frac{1}{n}\ln\Gamma_\alpha$ to more effectively quantify the approximate dynamical reversibility. After removing the constant term, $\gamma_\alpha=\hat\gamma_\alpha-\frac{1}{2}\ln\left(\frac{2\pi}{\alpha}\right)$ is a variable only determined by the singular values after the SVD of $\Sigma^{-1}$ and $A^T\Sigma^{-1}A$.
\begin{lem}   
\label{dimensionaveragedinf}
(Dimensional average reversible information based on SVD) For GIS like Eq.(\ref{micro-dynamics-appendix}), we only need to obtain the singular value spectra of inverse covariance matrices $A^T\Sigma^{-1}A$ and $\Sigma^{-1}$ to quantify the approximate dynamical reversibility. The corresponding indicator is the dimensional average reversibility information after removing the constant term, defined as
 \begin{eqnarray}\label{dimensionaveraged-appendix}
		\begin{aligned}
		\gamma_\alpha=\frac{1}{n}\left[(\frac{1}{2}-\frac{\alpha}{4})\sum_{i=1}^{r_s}\ln s_i+\frac{\alpha}{4}\sum_{i=1}^{r_\kappa}\ln\kappa_i\right]
   \end{aligned}
\end{eqnarray}
Here $r_s,r_\kappa$ denote the ranks of matrices $A^T\Sigma^{-1}A, \Sigma^{-1}$, respectively. $s_{i}$ and $\kappa_i$ denote the $i$-th singular value of matrices $A^T\Sigma^{-1}A$ and matrix $\Sigma^{-1}$ respectively.
\end{lem}
\begin{proof}
Due to the fact that $\Gamma_ {\alpha}$ is related to the dimensionality of the input variable and is influenced by a constant that is not related to the variable, we need to average the dimensionality of the input variable by taking the logarithm and eliminate the constant term as
 \begin{eqnarray}\label{ED}
		\begin{aligned}
		\gamma_\alpha&\equiv\frac{1}{n}\ln\Gamma_\alpha-\frac{1}{2}\ln\left(\frac{2\pi}{\alpha}\right)\\&
        =\frac{1}{n}\ln{\rm pdet}(A^T\Sigma^{-1}A)^{\frac{1}{2}-\frac{\alpha}{4}}+\frac{1}{n}\ln{\rm pdet}(\Sigma^{-1})^\frac{\alpha}{4}.
        \\&=\frac{1}{n}(\frac{1}{2}-\frac{\alpha}{4})\sum_{i=1}^{r_s}\ln s_{i}+\frac{\alpha}{4n}\sum_{i=1}^{r_\kappa}\ln \kappa_i
   \end{aligned}
\end{eqnarray}
to characterize dimensionally independent approximate dynamic reversibility for a more reasonable comparison of iterative systems of different dimensions. As $r(A^T\Sigma^{-1} A)\equiv r_s$, $r(\Sigma^{-1})\equiv r_\kappa$, $s_{i}(A^T\Sigma^{-1} A)\equiv s_i$ and $s_{i}(\Sigma^{-1})\equiv\kappa_i$, we can get our final results.
\end{proof}
\subsubsection{Determinism and degeneracy based on $\Gamma_\alpha$}
As pointed out by \cite{Zhang2025}, $\Gamma_\alpha$'s reflection on determinism in Eq.(\ref{DegeneracyDeterminismEI1}) and degeneracy in Eq.(\ref{DegeneracyDeterminismEI2}) mainly depends on the transformation of hyperparameter $\alpha$. By adjusting the parameter $\alpha\in(0,2]$, we can make $\Gamma_\alpha$ reflect the \textbf{determinism} or \textbf{degeneracy} of $p(x_{t+1}|x_t)=\mathcal{N}(Ax_t+a_0,\Sigma)$. 

As $\alpha\to 0$,
\begin{eqnarray}\label{ED}
		\begin{aligned}
		\Gamma_{\alpha\to 0}=\left(\frac{2\pi}{\alpha}\right)_{\alpha\to 0}^\frac{n}{2}{\rm pdet}(A^T\Sigma^{-1}A)^{\frac{1}{2}},\\
   \end{aligned}
\end{eqnarray} 
in which ${\rm pdet}(A^T\Sigma^{-1}A)$ is the Moore-Penrose generalized inverse matrix of $A^\dagger\Sigma (A^\dagger)^T$. $A^\dagger\Sigma (A^\dagger)^T$ is the covariance matrix of the inverse dynamics $x_{t}=A^\dagger x_{t+1}-A^\dagger a_0-A^\dagger\eta_t$ of $x_{t+1}=a_0+Ax_t+\eta_t$. In the calculation of $\mathcal{J}$, non-degeneracy $\mathcal{J}_2$ in Eq.(\ref{DegeneracyDeterminismEI2}) describes how exactly we can infer the state $x_{t-1}$ in previous time step from the current state $x_{t}$ which is also the predictability of backward dynamics, so $\Gamma_{\alpha\to 0}$ resembles the non-degeneracy term in the definition of $\mathcal{J}$. When ${\rm pdet}(A^T\Sigma^{-1}A)$ becomes smaller, $p(x_{t+1}|x_t)$ becomes more degeneracy. 

Similarly, as $\alpha\to2$, since
\begin{eqnarray}\label{ED}
		\begin{aligned}
		\Gamma_{\alpha\to2}&=\left(\pi\right)^\frac{n}{2}{\rm det}(\Sigma^{-1})^\frac{1}{2},\\
   \end{aligned}
\end{eqnarray}
and covariance matrix $\Sigma$ of $x_{t+1}=a_0+Ax_t+\eta_t$ directly determine the determinism term $\mathcal{J}_1$ in $\mathcal{J}$ which measures how the current state $x_t$ can deterministically influences the state $x_{t+1}$ in next time step, $\Gamma_2$ is comparable with the determinism term. An increase of ${\rm det}(\Sigma^{-1})$ leads to higher maximum transition probabilities, reflecting stronger determinism in the underlying dynamics. When ${\rm det}(\Sigma^{-1})\to\infty$, $p(x_{t+1}|x_t)$ approaches the Dirac distribution, the determinism will tend towards $\infty$.

In practice, $\alpha=1$ is often chosen to balance $\Gamma_\alpha$'s emphasis on both determinism and degeneracy. For $\alpha<1$, $\Gamma_\alpha$ tends to capture more of the non-degeneracy of $p(x_{t+1}|x_t)$. In contrast, for $\alpha> 1$, $\Gamma_\alpha$ emphasizes the determinism of $p(x_{t+1}|x_t)$. Given the significance of $\alpha=1$, we primarily present results for this case, and we will denote 
\begin{eqnarray}\label{ED}
		\begin{aligned}
		\Gamma\equiv\Gamma_1&=(2\pi)^{\frac{n}{2}}{\rm pdet}(A^T\Sigma^{-1}A)^{\frac{1}{4}}{\rm det}(\Sigma^{-1})^\frac{1}{4}\\
   \end{aligned}
\end{eqnarray}
in the subsequent discussion. 

\subsection{Nearly linear relationship}
To study the relationship between two CE indicators, we first need to understand the relationship between $EI$ and $\Gamma_\alpha$. In continuous systems, approximate equivalence can be analytically proven.

\begin{thm}\label{thm.gammaabeic}
(Correlation between $\Gamma_\alpha$ and $EI$): When the backward dynamics $p(x_t|x_{t+1})=\mathcal{N}(A^\dagger x_{t+1}-A^\dagger a_0,A^\dagger\Sigma(A^\dagger)^T)$ is close to a normalized normal distribution as $A^\dagger\Sigma(A^\dagger)^T\approx I_n$, $\Gamma_\alpha$ and dimension averaged $EI$ are positively correlated as
\begin{eqnarray}\label{gammaabeic}
\begin{aligned}
		\ln\Gamma_\alpha&\simeq n(1-\frac{\alpha}{4})\mathcal{J}+C.
    \end{aligned}
\end{eqnarray}
$\ln\Gamma_\alpha$ is the reversible information and measures the degree of approximate reversibility. $C=\frac{n}{2}\ln\left(\frac{2\pi}{\alpha}\right) -n(1-\frac{\alpha}{4})\ln(\frac{L}{\sqrt{2\pi e}})$ is a constant term independent of $A$ and $\Sigma$. When the backward dynamics $p(x_t|x_{t+1})$ is a normalized normal distribution, i.e. $A\in\mathcal{R}^{m\times n}$ is reversible and $A^\dagger\Sigma(A^\dagger)^T=A^{-1}\Sigma(A^{-1})^T=I_n$, the equal sign holds.
\end{thm}
\begin{proof}
If we only study a single random mapping $p(y|x)=\mathcal{N}(Ax,\Sigma)$, then $A$ may not be a full-rank square matrix. However, coarsening strategies are generally applied to Gaussian iterative systems $x_{t+1}=Ax_t+\eta_t$, where $A$ must be a square matrix to ensure the stability of the system's dimensions. So we need to set $m=n$ and $A,\Sigma\in\mathcal{R}^{n\times n}$. From Eq.(\ref{ASAEI}) and (\ref{gammagaussian}), we can obtain two approximate correlations as
\begin{eqnarray}\label{aproxEIandGamma}
		\begin{aligned}
		\ln\Gamma_\alpha&
        \simeq\frac{n}{2}\ln\left(\frac{2\pi}{\alpha}\right)+(1-\frac{\alpha}{2})\frac{1}{2}\ln{\rm pdet}(A^TA)+\frac{1}{2}\ln\det(\Sigma)^{-1}=C_1 + \textbf{a}\textbf{d},\\
        EI&\simeq n\ln\left(\frac{L}{\sqrt{2\pi e}}\right)+\frac{1}{2}\ln{\rm pdet}(A^TA)+\frac{1}{2}\ln\det(\Sigma)^{-1}=C_2 + \textbf{b}\textbf{d},
   \end{aligned}
\end{eqnarray}
where $\textbf{a}=\left((1-\frac{\alpha}{2}),1\right)$, $\textbf{b}=\left(1,1\right)$ and $\textbf{d}=\left(\frac{1}{2}\ln{\rm pdet}(A^TA),\frac{1}{2}\ln\det(\Sigma)^{-1}\right)^T$. 
The equal sign holds when $A$ is a full rank matrix as $r(A)=n$ or $\Sigma$ is an identity matrix as $\Sigma=I_n$. At the same time, we can use the generalized inverse matrix to obtain an approximation $\textbf{d}\simeq \textbf{b}^\dagger(EI-C_2)$ as $\textbf{b}^\dagger=(1/2,1/2)^T$. When ${\rm pdet}(A^TA)=\det(\Sigma)^{-1}$, the equal sign holds; the closer ${\rm pdet}(A^TA)$ and $\det(\Sigma)^{-1}$ is, the higher the degree of approximation. So we conclude
\begin{eqnarray}\label{gammaabeic}
\begin{aligned}
		\ln\Gamma_\alpha&\simeq C_1+\textbf{a}\textbf{b}^\dagger(EI-C_2)=(1-\frac{\alpha}{4})EI+C.
    \end{aligned}
\end{eqnarray}
$C_1, C_2$ and $C$ are all constants. Fig.\ref{basicdatatest}c presents the linear relationship between two indicators as $\alpha=1$.
\end{proof}

\section{Causal Emergence (CE)}

CE refers to a special property of Markov dynamics: if the system's dynamics can yield a macroscopic dynamics with larger EI under a predefined coarse-graining strategy, then the system exhibits CE. However, this definition relies on the selection of coarse-graining strategies. Hoel et al. proposed the principle of maximizing EI \cite{Hoel2013, Hoel2017, Yang2024, Liu2024} to eliminate the arbitrariness of coarse-graining. However, this approach requires solving an often intractable optimization problem to obtain the optimal coarse-graining strategy and derive CE. To overcome this difficulty, Zhang et al. \cite{Zhang2025} introduced a new measure of CE based on approximate dynamical reversibility with SVD. Independent of coarse-graining, CE can be quantified directly from the singular value spectrum of the system’s TPM. In the following section, we will extend this framework to GIS.

\subsection{Causal emergence based on maximizing effective information}
We have extended the CE framework to Gaussian iterative systems under Gaussian noise called GIS in the previous article \cite{Liu2024}, and many definitions and conclusions were drawn from this. A GIS \cite {Dunsmuir1976} is a sequence of random variables $x_t$ at different times $t$ which usually represent the dynamic evolution of a certain quantity, such as stock prices, temperature changes, traffic flows, etc. 

Based on micro-states like Eq.(\ref{micro-dynamics-appendix}), we define the macro-state $y_t=\phi(x_t)=Wx_t$, where $W\in\mathcal{R}^{k\times n}$ is the parameter of coarse-graining strategy $\phi(x_t)$. The derived macro-state dynamic is,
\begin{equation}
\label{macro-dynamics}
y_{t+1}=a_{M,0}+A_My_t+\eta_{M,t},\eta_{M,t}\sim\mathcal{N}(0,\Sigma_M),
\end{equation}
where $a_{M,0}=Wa_0$, $A_M=WAW^\dagger$, $\Sigma_M=W\Sigma W^T$ and $k<n$. In this article, $(\cdot)^\dagger$ is the Moore-Penrose generalized inverse matrix \cite{Barata2012,Horn2012} of $\cdot$.

Obtaining $\mathcal{J}$, we still have a free parameter $L$, which is artificially assumed and greatly influences the results of $\mathcal{J}$. This hyperparameter can be subtracted by calculating dimensional averaged CE as
\begin{eqnarray}\label{dimensionally_averaged_CE}
	\Delta\mathcal{J} = \mathcal{J}_M-\mathcal{J},
\end{eqnarray}
where $\mathcal{J}_M\equiv \mathcal{J}(A_M,\Sigma_M)$ is the effective information for the macro-dynamics and $\Delta\mathcal{J}$ is the degree of CE. 

After introducing CE $\Delta\mathcal{J}$ for GIS based on EI, it is obvious that one of its shortcomings is the high dependence on coarse-graining $\phi(x)=Wx$, $W\in\mathcal{R}^{k\times n}$, as 
\begin{eqnarray}\label{dimensionally_averaged}
\mathcal{J}_M\equiv \mathcal{J}(A_M,\Sigma_M)=\mathcal{J}(WAW^\dagger,W\Sigma W^T).
\end{eqnarray}
Quantifying CE by $\Delta\mathcal{J}$ requires pre-setting coarse-graining strategies $\phi$ and optimizing its parameters $W\in\mathcal{R}^{k\times n}$, $k$ is the pre-defined macro dimension as $k<n$. Then, our problem can be formulated as
\begin{eqnarray}
\begin{aligned}
\label{eq:maximize_ei_general} 
\max\quad&\Delta\mathcal{J}(W)\\
s.t\quad&W\in\mathcal{W}\subseteq\mathcal{R}^{k\times n}.
\end{aligned}
\end{eqnarray}

\subsubsection{Maximizing effective information by optimizing coarse-graining}

When optimizing $W$, we must specify a subset $\mathcal{W}$ of $\mathcal{R}^{k\times n}$ in advance as a limited range to avoid $\mathcal{J}_M$ being divergent in Eq.(\ref{eq:maximize_ei_general}), which is very subjective and complicated. At the end of \cite{Liu2024}, a coarse-grained strategy that ensures stable information entropy is presented as $W^\dagger=W^T$. In this case, we can optimize $W$ to obtain optimal macro-states and theoretical maximum $\Delta\mathcal{J}^{*}$ in \cite{Liu2024} as
\begin{eqnarray}\label{DeltaJstar}
\Delta\mathcal{J}^{*}=\frac{1}{2k}\sum_{i=1}^{k}\ln s_i-\frac{1}{2n}\sum_{i=1}^{n}\ln s_i,
\end{eqnarray}
where $s_i,i=1,\cdots,n$, are the singular values of $A^T\Sigma^{-1}A$. The pre-setting and optimization of $W$ increases computational complexity and reduces accuracy. Even after the optimization of $W$, it is difficult to obtain the optimal $\Delta\mathcal{J}^{*}$ in Eq.(\ref{DeltaJstar}) numerically based on data. We need to find a way to explore the potential of CE in GIS directly. 

It is worth noting that Erik Hoel's effective information method suffers from a flaw in its calculation of non-degeneracy. This definition of degeneracy essentially involves iterating the Markov dynamics to obtain $x_{t+1}$ when $x_t$ is known, and then calculating the determinism of the state $\tilde{x}_t$ at the previous moment, rather than directly inferring the previous moment $x_t$ from the knowledge of the next moment $x_{t+1}$. This lack of effective simplification makes it difficult to detect the problem when calculating the discrete TPM matrix EI \cite{Hoel2013, Hoel2016, Hoel2017}. However, this flaw is particularly evident in Eq.(\ref{DegeneracyDeterminismEI2}), which first multiplies the covariance matrix of the forward dynamics and then divides it by the covariance matrix of the inverse dynamics. Consequently, when the determinism and non-degeneracy are added together, the covariance matrix of the forward dynamics is eliminated. When $W$ is a unitary matrix, the SVD decomposition matrix $A^T\Sigma^{-1}A$ is actually the inverse covariance matrix of the inverse dynamics. The coarse-graining strategy obtained by maximizing EI ignores the covariance of the forward dynamics in the original dynamics, resulting in a suboptimal coarse-graining of the forward dynamics in the new macro dynamics. Therefore, blindly maximizing effective information as in \cite{Chvykov2020,Liu2024} is not always reasonable.\\

\textbf{Proof of Equation (\ref{dimensionally_averaged_CE})}
\begin{proof}
	For the micro dynamical systems like Equation (\ref{micro-dynamics-appendix}) and macro dynamical systems like Equation (\ref{macro-dynamics}) after coarse-graining, we can calculate the micro effective information $\mathcal{J}_m$ and macro effective information $\mathcal{J}_M$ of two kinds of dynamical systems separately. In this way, causal emergence
	\begin{eqnarray}\label{dCE}
		\begin{aligned}
			\Delta\mathcal{J}&=\mathcal{J}_M-\mathcal{J}_m=\frac{EI_M}{k}-\frac{EI_m}{n}\\
			&=\frac{1}{k}\ln\displaystyle\frac{|\det(A_M)|L^k}{\displaystyle \det(\Sigma_M)^\frac{1}{2}(2\pi e)^\frac{k}{2}}-\frac{1}{n}\ln\displaystyle\frac{|\det(A)|L^n}{\displaystyle \det(\Sigma)^\frac{1}{2}(2\pi e)^\frac{n}{2}}\\
			&=\frac{1}{k}\mathop{\left(\ln\displaystyle\frac{1}{\displaystyle \det(\Sigma_M)^\frac{1}{2}(2\pi e)^\frac{k}{2}}+\ln\displaystyle\left(|\det(A_M)|L^k\right)\right)}_{Determinism Information \quad Degeneracy Information }\\&-\frac{1}{n}\left(\ln\displaystyle\frac{1}{\displaystyle \det(\Sigma)^\frac{1}{2}(2\pi e)^\frac{n}{2}}+\ln\displaystyle\left(|\det(A)|L^n\right)\right)\\
			&=\mathop{\ln\frac{|\det(WAW^\dagger)|^\frac{1}{k}}{|\det(A)|^\frac{1}{n}}}_{Degeneracy Emergence}+\mathop{\ln\frac{|\det(\Sigma)|^\frac{1}{2n}}{|\det(W\Sigma W^{T})|^\frac{1}{2k}}}_{Determinism Emergence}.
		\end{aligned}
	\end{eqnarray}
\end{proof}

\textbf{Proof of Equation (\ref{DeltaJstar})}
\begin{proof}
	According to Jensen's inequality, we can obtain inequalities
	\begin{eqnarray}\label{CEG1}
		\begin{aligned}
			|\det(W\Sigma W^\dagger)|^{-\frac{1}{2}}=\left[\sum \left(\gamma_m \prod\kappa_i\right)\right]^{-\frac{1}{2}}\leq\left[\sum \gamma_m\left( \prod\kappa_i\right)^{-\frac{1}{2}}\right]=|\det(W\Sigma^{-\frac{1}{2}} W^\dagger)|,
		\end{aligned}
	\end{eqnarray}
    if and only if the eigenvalues are equal, the equal sign holds. By utilizing this inequality and combining it with the relevant properties of linear algebra, we can conclude that
	\begin{eqnarray}\label{CEG}
	\begin{aligned}
		\Delta\mathcal{J}&=\ln\frac{|\det(WAW^\dagger)|^\frac{1}{k}}{|\det(W\Sigma W^\dagger)|^\frac{1}{2k}}-\ln\frac{|\det(A)|^\frac{1}{n}}{|\det(\Sigma)|^\frac{1}{2n}}\\
		&= \frac{1}{k}\ln(|\det(W\Sigma W^\dagger)|^{-\frac{1}{2}}|\det(WAW^\dagger)|)-\frac{1}{n}\ln|\Sigma^{-\frac{1}{2}}A|\\
		&\leq \frac{1}{k}\ln(|\det(W\Sigma^{-\frac{1}{2}} W^\dagger)||\det(WAW^\dagger)|)-\frac{1}{n}\ln|\Sigma^{-\frac{1}{2}}A|\\
		&= \frac{1}{k}\ln|\det(W\Sigma^{-\frac{1}{2}} W^\dagger WAW^\dagger)|-\frac{1}{n}\ln|\Sigma^{-\frac{1}{2}}A|\\
		_{(\Sigma=VKV^T, A=V\Lambda V^T)}&\leq \frac{1}{k}\ln|\det(W\Sigma^{-\frac{1}{2}} AW^\dagger)|-\frac{1}{n}\ln|\Sigma^{-\frac{1}{2}}A|\\
		_{(W=(\tilde{W}_k,O_{k\times{(n-k)}})V^{T})}&\leq\frac{1}{k}\sum_{i=1}^{k}\ln\displaystyle|\tilde{d}_i|-\frac{1}{n}\sum_{i=1}^{n}\ln\displaystyle|\tilde{d}_i|\\
        &=\frac{1}{2k}\sum_{i=1}^{k}\ln s_i-\frac{1}{2n}\sum_{i=1}^{n}\ln s_i,
	\end{aligned}
\end{eqnarray}
When $A$ and $\Sigma$ share the same $n$ eigenvectors, $|\tilde{d}_1|\geq\dots\geq|\tilde{d}_n|$are $n$ eigenvalues of $A\Sigma^{-1/2}$. That is to say, $\Sigma=VKV^T$ and $A=V\Lambda V^T$, where $V$ is the matrix formed by the shared eigenvectors, and $K$ and $\Lambda$ are the eigenvalues for $A$ and $\Sigma$, respectively. $\kappa=(\kappa_1,\dots,\kappa_n)$ and $\lambda=(\lambda_1,\dots,\lambda_n)$. The eigenvalues of $\Sigma^{-1/2}A$ are equal to the product of the eigenvalues of $A$ and $\Sigma^{-1/2}$. $W=(\tilde{W}_k,O_{k\times{(n-k)}})V^{T}$ represents that both $A$ and $\Sigma^{-1/2}$ retain the eigenvalues with the highest norm, the second and third unequal signs can also be taken as equal signs. $s_i,i=1,\cdots,n$, are the singular values of $A^T\Sigma^{-1}A$, we can directly get $\Delta\mathcal{J}^{*}$ in Eq.(\ref{DeltaJstar}).
\end{proof}

\begin{figure}[htbp]
	\centering
\begin{minipage}[c]{0.4\textwidth}
		\centering
		\includegraphics[width=1\textwidth]{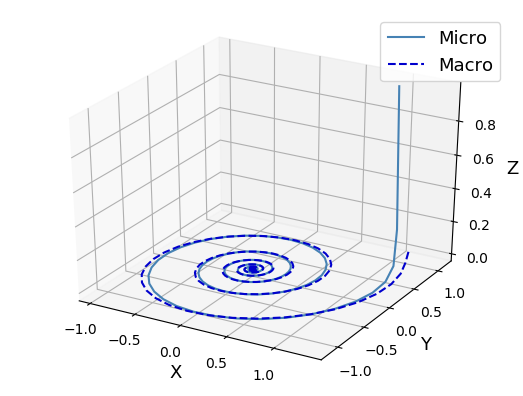}
		\centerline{(a)}
	\end{minipage}
        \begin{minipage}[c]{0.4\textwidth}
		\centering
		\includegraphics[width=1\textwidth]{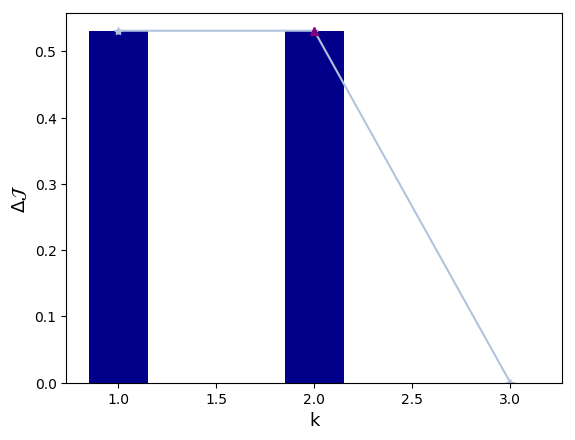}
		\centerline{(b)}
	\end{minipage}	
	\caption{Experimental results of spiral rotating. Assuming the direction vector of the rotation axis is $u_0=(0,0.1,1)^T$, after normalization as $u=u_0/||u_0||$, the rotation axis $u$ can be obtained. At the same time, we specify the rotation angle $\theta=\pi/16$ and $\eta=0$ as the dimension averaged Shannon entropy of both $x_t$ and $y_t$ remains the same. Through $\theta$ and $u$, we can get the rotation matrix $R$. By changing $\Psi$, we will obtain causal emergence in different forms. 
    in Fig.\ref{fig:rotation}c, when $\Psi={\rm diag}(0.99,0.97,0.2)$ and $x_0=(1,1,1)^T$, the path of $x_t$ will first shrink to a plane perpendicular to $u$. In \ref{fig:rotation}d, we project $x_t$ onto the plane as the macro-state, where $k=2$, $\Delta\mathcal{J}=0.5295$. The reason why we do not take $k=1$ here is that the more dimensionality we reduce, the greater the information loss. Therefore, in the absence of further improvement in causal emergence, we do not need to continue the dimensionality reduction of the system. (a) The spiral curve when $\Psi={\rm diag}(0.99,0.97,0.2)$ and $x_0=(1,1,1)^T$, the trajectory of $x_t$ is compressed to a plane perpendicular to $u$ at the initial stage. (b) We can project $x_t$ on the plane as a macro-state, where $k=2$, $\Delta\mathcal{J}=0.5295$.}
 	\label{fig:rotation}
\end{figure}
\subsubsection{Selection of $k$}
In general, we choose the dimension macro $k$ by comparing the maximum $\Delta\mathcal{J}^{*}$ after optimization under different dimensions $k$. Then, our problem can be formulated as
\begin{eqnarray}
\begin{aligned}
\label{eq:maximize_ei} 
\max\quad&\Delta\mathcal{J}(W)\\
s.t\quad&W\in\mathcal{W}\subseteq\mathcal{R}^{k\times n}\\ &k=1,\cdots,n
\end{aligned}
\end{eqnarray}
It can be seen that if $k$ is not specified in advance, in the process of maximizing causal emergence, $k$ will also be optimized along with $W$ to obtain the result. So when we optimize $\Delta\mathcal{J}(W)$ under different $k$, we can get a sequence $\{\Delta\mathcal{J}_i^*\}$ as $i=1,\cdots,n$. Equally, our problem can also be presented as solving
\begin{eqnarray}
k=\arg\max_{i}\Delta\mathcal{J}_i^*.
\end{eqnarray}
to select the final macro dimension $k$. If there are two dimensions with equal $\Delta\mathcal{J}^*$, to minimize information loss, we choose the slightly higher dimension.  For example, in the rotation model in Fig.\ref{fig:rotation}, we choose $k=2$ as the macro dimension.

\subsection{Causal emergence based on approximate dynamical reversibility with SVD}

We have now demonstrated that the approximate dynamical reversibility of GIS can be quantified directly based on SVD of $\Sigma^{-1}$ and $A^T\Sigma^{-1}A$, and has an approximately linear relationship with maximized EI. We also know that traditional EI-based CE relies on coarse-graining strategies and macro dynamics, but after maximizing EI when $W$ is a unitary matrix, CE can also be obtained by eliminating the singular values of $A^T\Sigma^{-1}A$. So, can we directly define a metric for CE based on the singular value spectrum, bypassing the search for the optimal coarse-graining solution $W$? This is the basic idea of the quantitative definition of SVD-based CE.

\subsubsection{SVD-based CE for TPM}\label{ReversibilityTPM}
According to Zhang's previous work \cite{Zhang2025}, since $\Gamma_{\alpha}$ is size-dependent, we need to normalize them by dividing the size of $P$ as $\gamma_{\alpha}=\Gamma_{\alpha}/N$, to characterize the size-independent approximate dynamical reversibility. There is an integer $i\in[1,N)$ such that $\zeta_i>\epsilon$, then there is \textbf{vague causal emergence} with the level of vagueness $\epsilon$ occurred in the system as
    \begin{equation}
        \label{eq:degree_vague_emergence}
        \Delta\Gamma_{\alpha}(\epsilon) = \frac{\sum_{i=1}^{r_{\epsilon}}\zeta_i^{\alpha}}{r_{\epsilon}}-\frac{\sum_{i=1}^{N}\zeta_i^{\alpha}}{N},
    \end{equation}
where $r_{\epsilon}=\max\{i|\zeta_i>\epsilon\}$. $\epsilon$ can be selected according to the relatively clear cut-offs in the spectrum of singular values. When $\epsilon=0$, if $r\equiv rank(P)<N$ then \textbf{clear causal emergence} occurs, and the degree is $\Delta\Gamma_{\alpha}(0)$. This definition is independent of any coarse-graining strategy and reflects the inherent property of the Markov chain.

Finally, a concise coarse-graining method based on the SVD of $P$ to obtain a macro-level reduced TPM can be obtained by projecting the row vectors $P_i,\forall i\in [1,N]$ in $P$ onto the sub-spaces spanned by the eigenvectors of $P\cdot P^T$ such that the major information of $P$ is conserved, as well as $\Gamma_{\alpha}$ is kept unchanged.


\subsubsection{SVD-based CE for GIS}
From the definition of dimensional average reversible information $\hat\gamma_\alpha$, we can obtain that its final value depends on two spectra, one is the spectrum $s_1\geq\cdots\geq s_n$ of $A^T\Sigma^{-1}A$, and the other is the spectrum $\kappa_1\geq\cdots\geq \kappa_n$ of $\Sigma^{-1}$. Therefore, following the quantitative equation of CE (Eq.(\ref{DeltaJstar})), we can define CE based on approximate dynamical reversibility from the distribution of these two spectra.

One of the main contributions of this article is the new quantification of CE for GIS based on approximate dynamical reversibility under SVD following the definition of TPM in \cite{Zhang2025}, which depends only on the system's parameters and does not require the optimization of coarse-graining strategies. Because $\Gamma_{\alpha}$ is calculated based on generalized determinants ${\rm pdet}(A^T\Sigma^{-1}A)$, removing zero singular values will not change the first term of $\gamma_ {\alpha}$. 

Because $\Gamma_ {\alpha}$ is calculated based on generalized determinants ${\rm pdet}(A^T\Sigma^{-1}A)$, removing zero singular values will not change $\gamma_ {\alpha}$.
If $r\equiv\min\{r_s,r_\kappa\}<n$, we can directly replace $n$ with $r$ to obtain the macro-state dimension averaged reversibility information as
 \begin{eqnarray}\label{ED}
		\begin{aligned}
		\gamma_{\alpha}(0)&=\frac{1}{r}(\frac{1}{2}-\frac{\alpha}{4})\sum_{i=1}^{r_s}\ln s_i+\frac{\alpha}{4r}\sum_{i=1}^{r_\kappa}\ln\kappa_i.
   \end{aligned}
\end{eqnarray}

\begin{definition}
    \label{dfn:clear_emergence}
    For GIS like Eq.(\ref{micro-dynamics-appendix}), if $r\equiv\min\{r_s,r_\kappa\}<n$ then \textbf{clear causal emergence} occurs in this system. The degree of causal emergence is 
    \begin{eqnarray}
\label{eq:degree_clear_emergence}
\begin{aligned}
\Delta\Gamma_\alpha(0)&=\gamma_\alpha(0)-\gamma_\alpha\\
&=\frac{1}{r}(\frac{1}{2}-\frac{\alpha}{4})\sum_{i=1}^{r_s}\ln s_i+\frac{\alpha}{4r}\sum_{i=1}^{r_\kappa}\ln\kappa_i\\&-\frac{1}{n}(\frac{1}{2}-\frac{\alpha}{4})\sum_{i=1}^{r_s}\ln s_i-\frac{\alpha}{4n}\sum_{i=1}^{r_\kappa}\ln\kappa_i.  
\end{aligned}
        \end{eqnarray}
\end{definition}
Fig.\ref{basicdatatest}f and g shows the distribution of $\Delta\mathcal{J}_R^{*}$ and $\Delta\Gamma_\alpha(0)$ with $a_2$ and  $\sigma_2$ in two-dimensional iterative systems as
 \begin{eqnarray}\label{ED}
		\begin{aligned}
		A=\begin{pmatrix}
		    a_1&0\\
                0&a_1
		\end{pmatrix},
        \Sigma=\begin{pmatrix}
		    \sigma^2_1&0\\
                0&\sigma^2_2
		\end{pmatrix}.
   \end{aligned}
\end{eqnarray}
It can be seen that when $a_2$ is large and $\sigma_2$ is small, there will be larger values for both, and their distributions are very close.

However, clear CE is difficult to directly detect in reality because it is unlikely for the parameter matrix derived from the data to have strictly zero singular values. So it is difficult to find a strict non-full rank matrix, even if the singular value of $A^T\Sigma^{-1}A$ is quite small and cannot be directly treated as a strict non-full rank matrix. On the other hand, $\Sigma^{-1}$ may have very small singular values. So more often than not, we need to establish a lower bound for singular values as $\epsilon$, where only singular values $s_i,\kappa_i\geq\epsilon$ are considered effective. Due to the information loss caused by deleted singular values $s_i,\kappa_i\leq\epsilon$, $\epsilon$ also determines the upper limit of the system's loss of information. The number of effective singular values is called the effective rank $r_\epsilon(A^T\Sigma^{-1}A)$ and $r_\epsilon(\Sigma^{-1})$. Based on global effective rank of the whole system as $r_\epsilon=\min\{r_\epsilon(A^T\Sigma^{-1}A),r_\epsilon(\Sigma^{-1})\}$, we can define vague macro reversibility as
 \begin{eqnarray}\label{ED}
		\begin{aligned}
		\gamma_{\alpha}(\epsilon)&=\frac{1}{r_\epsilon}(\frac{1}{2}-\frac{\alpha}{4})\sum_{i=1}^{r_\epsilon}\ln s_i+\frac{\alpha}{4r_\epsilon}\sum_{i=1}^{r_\epsilon}\ln\kappa_i
   \end{aligned}
\end{eqnarray}

\begin{definition}
    \label{dfn:vague_emergence}
    For GIS like Eq.(\ref{micro-dynamics-appendix}), suppose the singular values of $A^T\Sigma^{-1}A$ are $s_1\geq s_2\geq\cdots\geq s_{r_s}\geq 0$. For a given real value $\epsilon\in [0,s_1]$, if there is an integer $i\in[1,r_s]$ such that $s_i\equiv s_i(A^T\Sigma^{-1}A)>\epsilon$, then there is \textbf{vague causal emergence} with the level of vagueness $\epsilon$ occurred in the system. We define
    \begin{eqnarray}\label{ED}
		\begin{aligned}
r_\epsilon\equiv\min\{r_\epsilon(A^T\Sigma^{-1} A),r_\epsilon(\Sigma^{-1})\}
   \end{aligned}
\end{eqnarray}
    and the degree of causal emergence is:
    \begin{eqnarray}
\label{eq:degree_vague_emergence}
        \begin{aligned}
\Delta\Gamma_\alpha(\epsilon)&=\gamma_\alpha(\epsilon)-\gamma_\alpha\\
&=\frac{1}{r_\epsilon}(\frac{1}{2}-\frac{\alpha}{4})\sum_{i=1}^{r_\epsilon}\ln s_i+\frac{\alpha}{4r_\epsilon}\sum_{i=1}^{r_\epsilon}\ln\kappa_i\\&-\frac{1}{n}(\frac{1}{2}-\frac{\alpha}{4})\sum_{i=1}^{r_s}\ln s_i-\frac{\alpha}{4n}\sum_{i=1}^{r_\kappa}\ln\kappa_i
        \end{aligned}
        \end{eqnarray}
    where $r_{\epsilon}=\max\{i|s_{i}>\epsilon\}$.
\end{definition}
Fig.\ref{basicdatatest}j and i shows the distribution of $\Delta\mathcal{J}^{*}$ and $\Delta\Gamma_\alpha(\epsilon)$ with $a_2$ and  $\sigma_2$ in two-dimensional iterative systems and their distributions are very close.

\subsection{Equivalence between $\Delta\Gamma_\alpha$ and $\Delta\mathcal{J}^{*}$}\label{sec:connections}
The third important contribution of this article is the rigorous proof of the positive correlation between $\Delta\Gamma_\alpha$ and $\Delta\mathcal{J}^{*}$. 

\begin{thm}\label{thm.dgammaanddJ}
(Correlation between $\Delta\Gamma_\alpha(\epsilon)$ and $\Delta\mathcal{J}^{*}$): When the backward dynamics $p(x_t|x_{t+1})=\mathcal{N}(A^\dagger x_{t+1}-A^\dagger a_0,A^\dagger\Sigma(A^\dagger)^T)$ is close to a normalized normal distribution as $A^\dagger\Sigma(A^\dagger)^T\approx I_n$, $\Delta\Gamma_\alpha(\epsilon)$ and $\Delta\mathcal{J}^{*}$ are approximately linear and positively correlated as
\begin{eqnarray}
\label{eq:dgammaanddJ-appendix}
        \begin{aligned}
\Delta\Gamma_\alpha(\epsilon)\simeq(1-\frac{\alpha}{4})\Delta\mathcal{J}^*.
        \end{aligned}
        \end{eqnarray}
The equal sign holds when the inverse dynamics $p(x_t|x_{t+1})=\mathcal{N}(A^\dagger x_{t+1}-A^\dagger a_0,A^\dagger\Sigma(A^\dagger)^T)$ is a normalized normal distribution as $A^\dagger\Sigma(A^\dagger)^T=A^{-1}\Sigma(A^{-1})^T=I_n$. 
\end{thm}
\begin{proof}
We can compare $\Delta\mathcal{J}$ and $\Delta\Gamma_\alpha(\epsilon)$. When $W^{\dagger}=W^T$, $r=r_s$ is the rank of $A^T\Sigma^{-1}A$, the CE of
linear Gaussian iterative systems satisfy
   \begin{eqnarray}\label{JandJ^star}
   \begin{aligned}
\Delta\mathcal{J}&=\frac{1}{2k}\ln\det(A_M^T\Sigma^{-1}_MA_M)-\frac{1}{2n}\ln{\rm pdet}(A^T\Sigma^{-1}A)\\
&\leq\frac{1}{2k}\sum_{i=1}^{k}\ln\displaystyle s_{i}-\frac{1}{2n}\sum_{i=1}^{r}\ln\displaystyle s_{i}\equiv\Delta\mathcal{J}^{*}.
   \end{aligned}
\end{eqnarray} 
Similar to Eq.(\ref{aproxEIandGamma}), when $A$ approximates full rank or $\Sigma\approx\sigma^2I_n$, $\sum_{i=1}^{r}\ln s_{i}\simeq\sum_{i=1}^{r}\ln a_i+\sum_{i=1}^{n}\ln\kappa_i$. If the coupling degree between $A$ and $\Sigma^{-1}$ spaces is high as $s_i=a_i\kappa_i, i=1,\dots,n$, the macro-state will also satisfy the same approximation properties. In this way, when $k=r_\epsilon$, $m=n$ and $r=r_s$, we can obtain two approximate forms of CE as
\begin{eqnarray}\label{JandJ^star}
   \begin{aligned}
\Delta\mathcal{J}^{*}&\simeq\frac{1}{2k}\sum_{i=1}^{k}\ln\displaystyle a_{i}-\frac{1}{2n}\sum_{i=1}^{r}\ln\displaystyle a_{i}+\frac{1}{2k}\sum_{i=1}^{k}\ln\displaystyle \kappa_{i}-\frac{1}{2n}\sum_{i=1}^{n}\ln\displaystyle \kappa_{i},\\
\Delta\Gamma_\alpha(\epsilon)&\simeq\frac{1}{2r_\epsilon}(1-\frac{\alpha}{2})\sum_{i=1}^{k}\ln a_i-\frac{1}{2n}(1-\frac{\alpha}{2})\sum_{i=1}^{r}\ln a_i+\frac{1}{2r_\epsilon}\sum_{i=1}^{k}\ln\kappa_i-\frac{1}{2n}\sum_{i=1}^{r}\ln\kappa_i,
   \end{aligned}
\end{eqnarray} 
where $a_1\geq a_2\geq\dots\geq a_r $ are singular values of $A^TA$. The equal sign holds when $r(A)=n$ and $s_i=a_i\kappa_i$ or $\Sigma=\sigma^2I_n$. 

Based on the monotonicity of elementary functions and matrix theory like Eq.(\ref{gammaabeic}), we can conclude that $\Delta\Gamma_\alpha(\epsilon)$ and $\Delta\mathcal{J}^{*}$ are approximately linear and positively correlated as Theorem \ref{thm.dgammaanddJ} shows. The equal sign is established when (1) $\Sigma = \sigma^2I_n$ and (2) ${\rm pdet}(A^TA)=\det(\Sigma)^{-1}$. Otherwise, 3 conditions must be met together: (1) $r(A)=n$, (2) $s_i=a_i\kappa_i, i=1,\dots,n$, (3) ${\rm pdet}(A^TA)=\det(\Sigma)^{-1}$. If the model approaches either of the above two sets of conditions, the approximation is effective. Fig.\ref{basicdatatest}f,g and j presents the linear relationship between $\Delta\mathcal{J}^{*}$ and $\Delta\Gamma_\alpha(0)$. Fig.\ref{basicdatatest}h,i and l presents the linear relationship between $\Delta\mathcal{J}^{*}$ and $\Delta\Gamma_\alpha(\epsilon)$ as $\alpha=1$. We can see that the scatter points are very close to the slope and theoretical value of the curve.
\end{proof}

From a logical perspective, our definition of CE is actually an approximate necessary condition for maximizing EI-based CE as defined by Hoel et al. \cite{Hoel2013,Hoel2017}. In other words, for any Markov dynamics, if we want to find a coarse-graining strategy that maximizes EI, we can actually perform SVD on the dynamics first. For example, the discussion at the end of \cite{Liu2024} on the case where $W$ is a unitary matrix actually transforms the coarse-graining strategy into a screening of singular values of the matrix $A^T\Sigma^{-1}A$, which is part of the calculation of SVD-based CE, namely the SVD of the inverse covariance matrix $A^T\Sigma^{-1}A=USU^T$ of backward dynamics. If the singular values $S={\rm diag} (s_1, \cdots, s_n)$, and the orthogonal singular vector matrix $U=(u_1, \cdots, u_n)$, in order to preserve the maximum $k$ singular values, we only need to preserve the dimensions corresponding to these singular vectors $U_1=(u_1, \cdots, u_k)$ and discard $U_2=(u_{k+1}, \cdots, u_n)$. Due to the orthogonality of $U$, we can directly obtain $U_1^TA^T\Sigma^{-1}AU_1 =U_1^TUSU^TU_1={\rm diag}(s_1,\cdots,s_k)$. It can be seen that we only need to multiply the coarse-graining parameter matrix $W=U_1^T$ by the original variables to make the inverse covariance matrix $WA^T\Sigma^{-1}AW^T$ of the macroscopic backward dynamics become a diagonal matrix that only retains $s_1,\cdots,s_k$. The singular vectors $u_1,\cdots,u_k$ corresponding to those larger singular values actually correspond to the optimal coarse-graining strategy we are looking for. 

Although there is a high likelihood that the two frameworks overlap, they are not entirely equivalent because the coarse-graining strategy obtained by maximizing EI ignores the covariance of forward dynamics. So, we propose an algorithm that constructs a coarse-graining strategy based on SVD, considering both the inverse covariance matrices $\Sigma^{-1}$ and $A^T\Sigma^{-1}A$ of the forward and backward dynamics, and attempts to preserve the maximum singular values of the two matrices as much as possible. In this way, it can ensure that the coarse-graining strategy matrix $W$ is unitary, and only the singular values of the inverse covariance matrix will be screened without changing the singular values, so that the probability space retained by the macro dynamics and the corresponding micro dynamics has the same probability distribution, and it can also ensure that both forward and backward dynamics are taken into account in the calculation of the coarse-graining strategy at the same time without only focusing on one side.

\begin{figure}[htbp]
    \centering
\includegraphics[width=1\textwidth,trim=10 20 50 0, clip]
{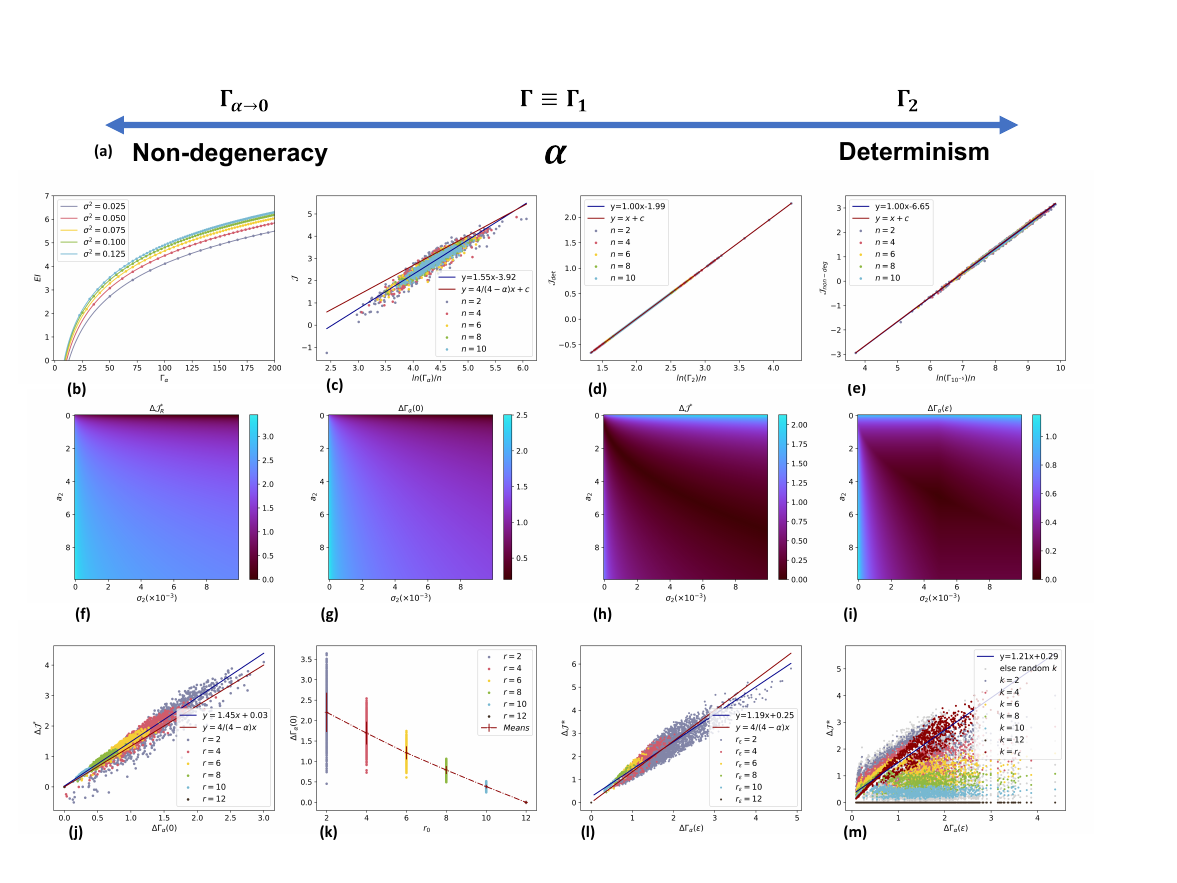}
    \caption{(a) $\alpha=1$ is often chosen to balance $\Gamma_\alpha$'s emphasis on both determinism and degeneracy. For $\alpha<1$, $\Gamma_\alpha$ tends to capture more non-degeneracy of $p(x_{t+1}|x_t)$. For $\alpha> 1$, $\Gamma_\alpha$ emphasizes the determinism. (b) The logarithmic relationship between $EI$ and $\Gamma_\alpha$. It can be seen that for different noise covariances $\sigma$, as $\Gamma_\alpha$ increases, $EI$ shows an approximate logarithmic growth relationship with $\Gamma_\alpha$ as $\alpha=1$. (c) The linear relationship between two indicators $\mathcal{J}$ and $\Gamma_\alpha$ as $\alpha=1$. (d) the linear relationships between determinism to $\Gamma_\alpha$ as $\alpha\to 2$. (e) the linear relationships between non-degeneracy to $\Gamma_\alpha$ as $\alpha\to 0$. The dimensional average reversible information $\ln\Gamma_\alpha$ and dimensional average effective information $\mathcal{J}$, as well as their corresponding determinism and non-degeneracy, are all linearly related. (f), (g), (h) and (i) shows the distribution of $\Delta\mathcal{J}_R^{*}$,$\Delta\Gamma_\alpha(0)$,$\Delta\mathcal{J}^{*}$ and $\Delta\Gamma_\alpha(\epsilon)$ with $a_2$ and  $\sigma_2$ in two-dimensional iterative systems. (j) The linear relationship between $\Delta\mathcal{J}_R^{*}$ and $\Delta\Gamma_\alpha(0)$ as $\alpha=1$. (k) $\Delta\Gamma_\alpha(0)$ under different fixed $r$. (l) The linear relationship between $\Delta\mathcal{J}^{*}$ and $\Delta\Gamma_\alpha(\epsilon)$ as $\alpha=1$. (l) The relationship between $\Delta\mathcal{J}^{*}$ and $\Delta\Gamma_\alpha(\epsilon)$ as $\alpha=1$ under different $k$ and $r_\epsilon$.}
\label{basicdatatest}
\end{figure}

\section{Coarse-graining strategy based on SVD}\label{Coarse-graining}
Firstly, we need to review the coarse-grained strategy based on SVD for discrete TPM, which is mainly based on the singular vectors corresponding to the largest few singular values.

\subsection{Coarse-graining strategy based on SVD of
TPM}
\label{sec:coarse-graining}
Although clear or vague CE phenomena can be defined and quantified without coarse-graining, a simpler coarse-grained description for the original system is needed to compare with the results derived from EI. Therefore, we also provide a concise coarse-graining method based on the singular value decomposition of $P$ to obtain a macro-level reduced TPM. The basic idea is to project the row vectors $P_i,\forall i\in [1,N]$ in $P$ onto the sub-spaces spanned by the eigenvectors of $P\cdot P^T$ such that the major information of $P$ is conserved, as well as $\Gamma_{\alpha}$ is kept unchanged. 

Concretely, the coarse-graining method contains five steps: 

1) We first make SVD decomposition for $P$ (suppose $P$ is irreducible and recurrent such that stationary distribution exist):
\begin{equation}
    \label{eq:svd_decomp}
    P = U\cdot\Sigma\cdot V^T, 
\end{equation}
where, $U$ and $V$ are two orthogonal and normalized matrices with dimension $N\times N$, and $\Sigma=\rm{diag}(\sigma_1,\sigma_2,\cdot\cdot\cdot,\sigma_N)$ is a diagonal matrix which contains all the ordered singular values. 

2) Selecting a threshold $\epsilon$ and the corresponding $r_{\epsilon}$ according to the singular value spectrum;

3) Reducing the dimensionality of row vectors $P_i$ in $P$ from $N$ to $r$ by calculating the following Equation:
\begin{equation}
    \label{eq:P_V_projection}
    \Tilde{P}\equiv P\cdot V_{N\times r},
\end{equation}
where $V_{N\times r}=(V_1^T,V_2^T,\cdot\cdot\cdot,V_r^T)$. 

4) Clustering all row vectors in $\Tilde{P}$ into $r$ groups by K-means algorithm to obtain a projection matrix $\Phi$, which is defined as:
\begin{equation}
    \label{eq:projection_matrix}
    \Phi_{ij}=\begin{cases} 1 &\mbox{if $\Tilde{P}_i$ is in the $r$th group}\\ 
0 & \mbox{otherwise},
\end{cases} 
\end{equation}
for $\forall i,j\in[1,N]$.

5) Obtain the reduced TPM according to $P$ and $\Phi$. 

To illustrate how we can obtain the reduced TPM, we will first define a matrix called stationary flow matrix as follows:
\begin{equation}
    \label{eq:stationary_flow}
    F_{ij}\equiv \mu_i\cdot P_{ij}, \forall i,j\in[1,N],
\end{equation}
where $\mu$ is the stationary distribution of $P$ which satisfies $P\cdot\mu=\mu$. 

Secondly, we will derive the reduced flow matrix according to $\Phi$ and $F$:
\begin{equation}
    \label{eq:F_definition}
    F'=\Phi^T\cdot F\cdot\Phi,
\end{equation}
where, $F'$ is the reduced stationary flow matrix. Finally, the reduced TPM can be derived directly by the following formula:
\begin{equation}
    \label{eq:coarse-graining}
    P'_i=F'_i/\sum_{j=1}^N(F'_i)_j, \forall i\in[1,N].
\end{equation}
Finally, $P'$ is the coarse-grained TPM.

\subsection{Coarse-graining strategy based on SVD of GIS}
After learning that $\Delta\Gamma_\alpha$ and $\Delta\mathcal{J}^{*}$ are approximately positively correlated and only determined by the difference in singular values of $A^T\Sigma^{-1}A$ and $\Sigma^{-1}$, we need to find the optimal coarse-graining strategy $\phi(x_t)=Wx_t$ defined by $W\in\mathcal{R}^{r_\epsilon\times n}$, based on SVD, which is determined by the singular vectors of $A^T\Sigma^{-1}A$ and $\Sigma^{-1}$. Since both $A^T\Sigma^{-1}A$ and $\Sigma^{-1}$ are symmetric matrices, the SVD of them are $A^T\Sigma^{-1}A=USU^T
,\Sigma^{-1}=VKV^T$, 
$S={\rm diag}(s_1,\cdots,s_n)$ and $K={\rm diag}(\kappa_i,\cdots,\kappa_n)$, $s_1\geq\cdots\geq s_n$, $\kappa_1\geq\cdots\geq \kappa_n$, are singular value matrices of $A^T\Sigma^{-1}A$ and $\Sigma^{-1}$. Corresponding singular vector matrices $U=(u_1,\cdots,u_n)$ and $V=(v_1,\cdots,v_n)$. Due to the filtering of singular values when calculating $\Delta\Gamma_\alpha(\epsilon)$, we need to redefine the singular vectors corresponding to the retained and discarded singular values during coarse-graining.

The optimal coarse-graining strategy $\phi(x_t)=Wx_t$ requires considering both $A^T\Sigma^{-1}A$ and $\Sigma^{-1}$ to balance forward and backward dynamics, so it is necessary to construct $W$ based on the vectors of $U$ and $V$ together. Since $U\not\equiv V$, we generally cannot guarantee that $s_i,\cdots,s_{r_\epsilon}$ and $\kappa_i,\cdots,\kappa_{r_\epsilon}$ are retained together. Therefore, when constructing $W$, we need a classified discussion of the span spaces of $U$ and $V$. Here we present an algorithm to achieve this.

\subsubsection{Algorithm}
This section presents the algorithm for constructing a coarse-graining strategy based on SVD.\\

\textbf{Input:} Matrices $A$, $\Sigma$

\textbf{Step 1:} 
Perform SVD on $A^T \Sigma^{-1} A$ and $\Sigma^{-1}$:
\begin{align*}
    & U S U^T= A^T \Sigma^{-1} A ,\\
    & VKV^T = \Sigma^{-1},
\end{align*}

\textbf{Step 2:}
Combine $S={\rm diag}(s_1,\cdots,s_n)$ and $K={\rm diag}(\kappa_i,\cdots,\kappa_n)$, then sort $s_1\geq\cdots\geq s_n$ and $\kappa_1\geq\cdots\geq \kappa_n$ in descending order to get $\tilde S={\rm diag}(\tilde{s}_1,\cdots,\tilde{s}_{2n})$ as $\tilde{s}_1\geq\cdots\geq\tilde{s}_{2n}$.

\textbf{Step 3:}
Truncate $\tilde S$ with a threshold $\epsilon$, and keep corresponding singular vectors and values:
\begin{eqnarray}
    \tilde S_1 = \mathrm{Truncate}(\tilde S, \epsilon),
\end{eqnarray}
Retaining $\tilde S_1=(\tilde{s}_1,\cdots,\tilde{s}_{r_{\epsilon1}})$, $\tilde{s}_1>\cdots>\tilde{s}_{r_{\epsilon1}}>\epsilon$, and keeps the corresponding singular vectors from $U=(u_1,\cdots,u_n)$ and $V=(v_1,\cdots,v_n)$, which are combined as $\tilde{U}_1=(\tilde{u}_1,\cdots,\tilde{u}_{r_{\epsilon1}})\in\mathcal{R}^{{n}\times r_{\epsilon1}}$. $r_{\epsilon1}<2n$ is the effective rank we get from SVD for the first time.

\textbf{Step 4:}
Apply SVD on the matrix $\tilde U_1\tilde S_1\in\mathcal{R}^{{n}\times r_{\epsilon1}}$:
\begin{eqnarray}
    \hat U \hat S \hat V^T= \tilde U_1\tilde S_1,
\end{eqnarray}
$\hat{S}\in\mathcal{R}^{n\times r_{\epsilon1}}$ contains singular values of matrix $\tilde{U}_{1}\tilde{S}$, $\hat{U}\in\mathcal{R}^{n\times n}$ are the corresponding singular vectors. 

\textbf{Step 5:}
Truncate $\hat U$ and $\hat S$ using the same threshold $\epsilon$, and keep the retained singular vectors as $\hat U_1\in\mathcal{R}^{n\times r_{\epsilon2}}$. Use these to obtain the final quantization strategy as
\begin{eqnarray}
    W = \hat U_1^T,
\end{eqnarray}
where $r_{\epsilon2}<n$ is the effective rank we get from SVD for the second time, which can be regarded as the final dimension of macro states $r_\epsilon$.

\textbf{Output:} Coarse-graining function $\phi(x)=Wx$.
\subsubsection{Explanation}
According to \cite{Zhang2025}, the reversibility of dynamics analyzed with the SVD method relates directly to maximizing EI in macroscopic dynamics. The singular vectors corresponding to the largest singular values of the TPM align with the directions that maximize macroscopic EI. For continuous systems, we decompose the deterministic and non-degenerate terms with SVD separately and then analyze and screen the combined singular value spectra. In this process, coarse-graining based on determinism and that based on non-degeneracy do not necessarily remain independent, so $\tilde U_1$ may fail to stay orthogonal and linearly independent. For example, if $A$ is the identity matrix, the $U$ and $V$ obtained in the first step are identical. This results in different macroscopic variables that actually correspond to the same microscopic variables. To prevent this issue, we perform a second SVD to decouple the different coarse-graining directions, thus forming a unified coarse-graining scheme.

Specifically, when dealing with known models, we can get the global effective rank $r_\epsilon$ in advance, so we can directly specify $r_{\epsilon1}=n$ and $r_{\epsilon2}=r_\epsilon$ to complete our coarse-graining operation. What's more, when $\Sigma\approx\sigma^2 I$ or $A^\dagger\Sigma(A^\dagger)^T\approx(\sigma^2/a^2) I$, as $a$ and $\sigma$ are constants, our parameter is equivalent to $W=U_1^T$ or $W=V_1^T$ as $U_1=(u_1,\cdots,u_{r_\epsilon}), V_1=(v_1,\cdots,v_{r_\epsilon})$.

\section{Experiments}
Here are some supplements to the experiments in the main text.

\begin{figure}[htbp]
    \centering
\includegraphics[width=1\textwidth,trim=0 165 470 0, clip]
{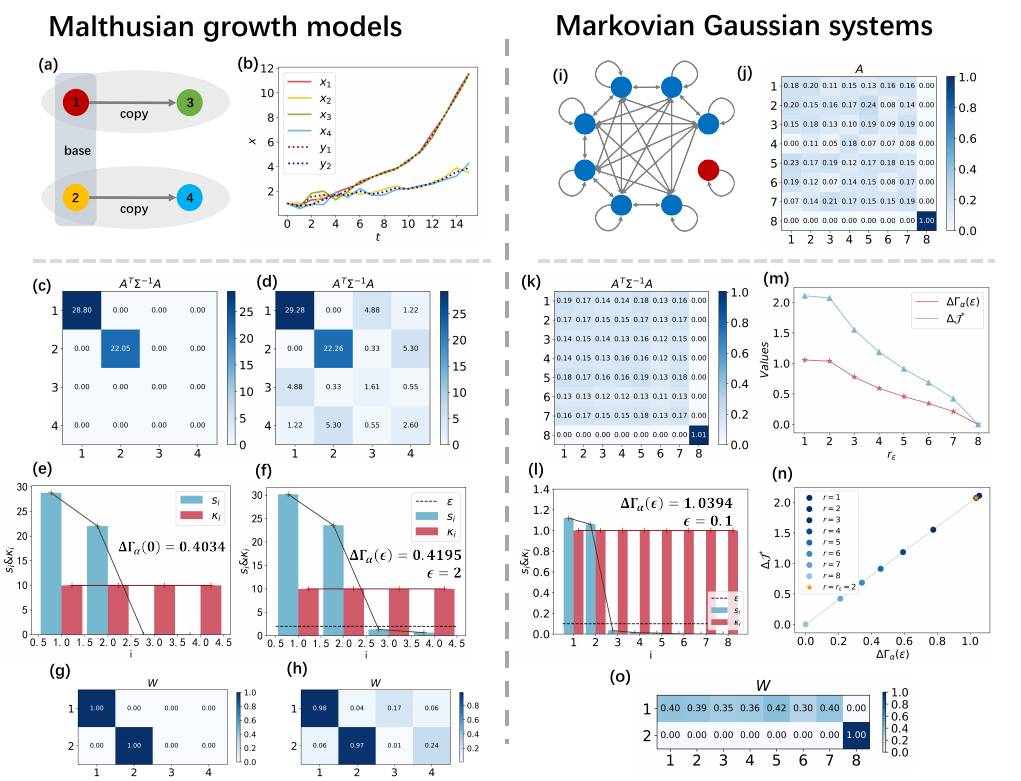}
    \caption{Analytical solutions obtained from numerical simulations of known models. (a) A Malthusian growth mode where $x_3,x_4$ are copies of $x_1,x_2$. (b) A sample of trajectory generated by the Malthusian growth models with growth rates of 0.2 and 0.05. $x_i$, $i=1,2,3,4$, are micro-states, while $y_1$ and $y_2$ are theoretical macro-states obtained by coarse-graining. (c) The backward dynamics covariance matrix $A^T\Sigma^{-1}A$ when $r(A)=2$. (d) The singular value spectrum of $A^T\Sigma^{-1}A$ when $A^T\Sigma^{-1}A$ only has two singular values. Clear CE can be calculated as $\Delta\Gamma_\alpha(0)=0.4034$. (e) Coarse-graining parameter $W$ obtained by truncating $U$ in the presence of Clear CE, the number of columns represents the macroscopic dimension, and the number of rows represents the microscopic dimension. (f) $A^T\Sigma^{-1}A$ after random perturbations to $A$. (g) The singular value spectrum of $A^T\Sigma^{-1}A$ when $A$ is perturbed. The degree of Vague CE is $\Delta\Gamma_\alpha(\epsilon)=0.4195$ when $\epsilon=2$. (h) The coarse-graining parameter $W$ for vague CE case, combined $x_1$ and $x_3$, $x_2$ and $x_4$. (i) A Markov Gaussian system. (j) The connectivity matrix $A$ of the Markov Gaussian system with the first 7 variables are depedent each other and the 8th variable is isolated. (k) The backward dynamics covariance matrix $A^T\Sigma^{-1}A$ of the Markov Gaussian system. (l) The singular value spectrum of $A^T\Sigma^{-1}A$ in the Markov Gaussian system has two larger singular values and the cut off can be taken at $\epsilon=0.1$. Vague CE occurs, and the degree of CE is $\Delta\Gamma_\alpha(\epsilon)=1.0394$. (m) The dependence of $\Delta\mathcal{J}^{*}$ and $\Delta\Gamma_\alpha(\epsilon)$ with the manually setting $r_\epsilon$. (n) The correlation between $\Delta\mathcal{J}^{*}$ and $\Delta\Gamma_\alpha(\epsilon)$ by variying $r_{\epsilon}$. (o) The coarse-graining parameter $W$ obtained by the singular vector with the largest singular values, it merges the first seven
    dimensions of Markov Gaussian system to form a new macro-variable, and keep the last variable as another macro-variable.}
\label{fig:Known_model-appendix}
\end{figure}

\subsection{Details of Malthusian growth models}
In a dynamical system, the degrees of freedom represent the number of independent variables needed to fully describe its state. When strong correlations or constraints exist between different dimensions (e.g., state variables), they no longer evolve independently. This interdependence, which can arise from conservation laws, synchronization, or other mathematical constraints, implies that the value of one variable can be determined by the others. Consequently, the system's effective degrees of freedom are reduced. Instead of exploring the full high-dimensional space, the system's trajectory is confined to a lower-dimensional manifold, simplifying its overall behavior.

The first example contains 4 variables, in which the first two variables $x_1,x_2$ follow the Malthusian growth model \cite{Galor2000} with different growth rates of 0.2 and 0.05. To study CE, we define the other two variables $x_3,x_4$ as the copies of the first two variables shown as Fig.\ref{fig:Known_model-appendix}a, thus, they are redundant dimensions. If $x=(x_1,x_2,x_3,x_4)$, the evolution of $x$ is a GIS $x_{t+1}=a_0+Ax_t+\eta_t, \eta_t\sim\mathcal{N}(0,\sigma^2 I_4)$ as $x_t,x_{t+1}\in\mathcal{R}^{4}$, $\sigma^2=0.1$, $a_0=0$, and
 \begin{eqnarray}
\label{A:growth}
        A = \left(\begin{matrix}
            1.2 & 0 &0 &0 \\
            0 & 1.05 &0 &0 \\
            1.2 & 0 &0 &0 \\
            0 & 1.05 &0 &0  
        \end{matrix}\right).
        \end{eqnarray}
        
A sample of evolutionary trajectories generated by this model is shown in Fig.\ref{fig:Known_model-appendix}b. In this model, $\Sigma=\sigma^2\cdot I_4$ is a full rank matrix, so we only need to study the backward dynamics covariance matrix $A^T\Sigma^{-1}A$ (Fig.\ref{fig:Known_model-appendix}c). This matrix only has two singular values as the singular value spectrum is shown in Fig.\ref{fig:Known_model-appendix}e with $r=2<4$, the horizontal axis represents the sequence number of singular values arranged in descending order as $s_1\geq\dots\geq s_4$, while the vertical axis represents the magnitude of singular values $s_i$, so clear CE can be calculated as $\Delta\Gamma_\alpha(0)=0.4034$. Fig.\ref{fig:Known_model-appendix}g shows the coarse-graining parameter $W\in\mathcal{R}^{2\times 4}$ obtained by truncating the orthogonal matrix $U$ after SVD of $A^T\Sigma^{-1}A$. When clear CE exists, variables $x_1$ and $x_2$ hold all the relevant information contained within $x_3$ and $x_4$. Therefore, the role of $\phi(x)=Wx$ is to remove the last two dimensions and only retain the first two dimensions of $x$ to describe the evolution of the growth model.

To show the concept of vague CE, we can add some perturbations to $A$ and obtain
 \begin{eqnarray}
\label{hatA:growth}
       A = \left(\begin{matrix}
            1.2 & 0 &0 &0 \\
            0 & 1.05 &0.001 &0 \\
            1.22 & 0 &0.4 &0.1 \\
            0 & 1.06 &0.03 &0.5  
        \end{matrix}\right)
        \end{eqnarray}
and $A^T\Sigma^{-1}A$ (Fig.\ref{fig:Known_model-appendix}d) is a full rank matrix as $r=n=4$. Then clear CE $\Delta\Gamma_\alpha(0)=0$. However, by observing the singular value spectrum in Fig.\ref{fig:Known_model-appendix}f, we can see that only two dimensions have significant impacts, so we need to calculate vague CE. By using the threshold selection as $\epsilon=2$, vague CE can be calculated as $\Delta\Gamma_\alpha(\epsilon)=0.4195$. The coarse-graining parameter $W$ is shown in Fig.\ref{fig:Known_model-appendix}h, the columns represent the macroscopic dimensions, and the rows represent the microscopic dimensions. Due to disturbances, $x_3$ and $x_4$ also contain some independent information. $\phi(x)=Wx$, $W=U_1^T$, is to merge $x_1$ and $x_3$, $x_2$ and $x_4$ by weighted summation and the two row-vectors that makeup $W$ are exactly equal to the singular vectors corresponding to the two largest singular values $s_1, s_2$.

This case illustrates that sometimes high dimensions may reduce the efficiency of reversibility of a system due to dimensional redundancy. Retaining only the maximum $r_\epsilon$ singular values of the covariance matrix can improve the efficiency of reversibility and generate CE.
\subsection{Markovian Gaussian systems}
In addition to directly copying, the system also has information redundancy if there are several dimensions with strong correlations. In this situation, CE will also be evident. A classic case here is the Markov Gaussian system \cite{Barrett2011,Seth2011,Oizumi2016,Tajima2017}, which can be seen as a Markov chain in a continuous state space. The example systems are characterized by the dynamics $x_{t+1} =Ax_t+\eta_t$, where $x_t$ contains $n$ variables, $A=(a_{ij})_{n\times n}$ is the connectivity matrix and $\eta_t\sim\mathcal{N}(0,I_n)$. $0\leq a_{ij}\leq 1$ reflects whether there is a connection between $i$ and $j$ and the strength of the connection. In this way, the system's state fluctuates within a stable region, and the Markovian Gaussian system can model EEG data \cite{Barrett2011,Varley2023} to study consciousness-related issues. The relevant data can be transformed into a 0-1 time series and studied using TPM, but the information is lost.

We considered a system with $n=8$ and connectivity as shown in Figure. Connection strengths between the first seven nodes are $1/7$ and the eighth node is a self-ring with a connection strength of 1. To avoid $\gamma_\alpha=0$ caused by singular values all being 1, we can add some perturbations to $A$, such as
 \begin{eqnarray}
\label{hatA:GM}
   A\sim\mathcal{N}\left(\begin{pmatrix}
    1/7&1/7&1/7&1/7&1/7&1/7&1/7&0\\
    1/7&1/7&1/7&1/7&1/7&1/7&1/7&0\\
    1/7&1/7&1/7&1/7&1/7&1/7&1/7&0\\
    1/7&1/7&1/7&1/7&1/7&1/7&1/7&0\\
    1/7&1/7&1/7&1/7&1/7&1/7&1/7&0\\
    1/7&1/7&1/7&1/7&1/7&1/7&1/7&0\\
    1/7&1/7&1/7&1/7&1/7&1/7&1/7&0\\
    0&0&0&0&0&0&0&1
    \end{pmatrix},0.05I_8\right).
        \end{eqnarray}
As Fig.\ref{fig:Known_model-appendix}i showed, the first seven dimensions of the system have the same connection strength and equal correlation. The eighth dimension is independent of the first seven dimensions and without correlation. After randomly generating a matrix $A$, we can calculate $A^T\Sigma^{-1}A$, from Fig.\ref{fig:Known_model-appendix}j and \ref{fig:Known_model-appendix}k we can see that matrix $A^T\Sigma^{-1}A$ can also be approximately divided into two independent parts as $A$. By observing the singular value spectrum in Fig.\ref{fig:Known_model-appendix}l, we can observe that $A^T\Sigma^{-1}A$ has two larger singular values and $r_\epsilon=2$ as $\epsilon=0.0552$. Vague CE can be calculated as $\Delta\Gamma_\alpha(\epsilon)=1.0394$. The coarse-graining parameter $W$ is shown in Fig.\ref{fig:Known_model-appendix}o, $\phi(x)=Wx$ can merge the first seven dimensions of the Markov Gaussian system to obtain a system where the efficiency of reversibility is improved. 

We can also analyze the optimal value of $r_\epsilon$ by balancing the enhancement of reversibility efficiency and information loss. Analyzing the impact of manually setting effective rank $r_\epsilon$ on CE in Fig.\ref{fig:Known_model-appendix}m, we find that from $r_\epsilon=8$ to $r_\epsilon=2$, $\Delta\mathcal{J}^{*}$ and $\Delta\Gamma_\alpha(\epsilon)$ significantly increases as $r_\epsilon$ decreases. But comparing $r_\epsilon=2$ with $r_\epsilon=1$, we found that due to the close similarity in the size of the first two singular values and the weights of the two dimensions, CE when $r_\epsilon=1$ cannot be significantly improved compared to CE when $r_\epsilon=2$. The property of information entropy based on Gaussian distribution, the loss of information entropy in the inverse dynamics of macro-states $y_t=Wx_t$ and micro-states $x_t$ is 
$H(x_t|x_{t+1})-H(y_t|y_{t+1})=\sum_{i=r_\epsilon}^{n}\ln s_i$, if $r_\epsilon<2$, the loss will be significantly increased.   This can explain why $r_\epsilon=2$ is the optimal effective rank, and the optimal macro-state dimension is also 2 dimensions. If we coarsen the system into one dimension, not only will CE not be significantly improved, but the model may also lose too much information. At the same time, in Fig.\ref{fig:Known_model-appendix}n we can also confirm that there is a linear relationship between $\Delta\mathcal{J}^{*}$ and $\Delta\Gamma_\alpha(\epsilon)$ under different $r_\epsilon$. 

\begin{figure}[htbp]
    \centering
\includegraphics[width=1\textwidth,trim=0 180 490 0, clip]
{Figure/Figure5.pdf}
    \caption{An example of applying machine learning to extract dynamics from data and measure CE of the dynamical system. (a) Schematic diagram of the SIR model. (b) Schematic diagram of learning dynamics and calculating $\Gamma$ through NN based on data. (1) and (3) are the variables $x_t$ and predicted $x_{t+\Delta t}$ at adjacent time points. (2) NN model for machine learning, whose structure is in Appendix F.3. (4) Dynamics $f(x_t)$ learned by NN model based on data of SIR. (5) $A_{x_t}=\nabla f(x_t)$ is the Jacobian matrix of the trained model at $x_t$. (6) The output $L_\Sigma$ of the Covariance Learner Network, the covariance matrix $\Sigma_{x_t}$ is derived from which. (7) The numerical solution of $\Gamma$, which is related to $A$ and $\Sigma$. (8) Applying $A^T\Sigma^{-1}A=\left<A^T_{x_t}\Sigma^{-1}_{x_t}A_{x_t}\right>_{x_t\in\mathcal{X}}$ and $\Sigma^{-1}=\left<\Sigma^{-1}_{x_t}\right>_{x_t\in\mathcal{X}}$, SVD-based CE $\Delta\Gamma(\epsilon)$ can be calculated under suitable $\epsilon$, $\mathcal{X}$ is the domain of SIR. (c) Training data which is the same as NIS+ \cite{Yang2024}. The full dataset (entire triangular region) used for training is displayed, along with four example trajectories with the same infection and recovery or death rates. The method of generating training data can be found in Appendix F.3. (d) $A^T\Sigma^{-1}A=\left<A^T_{x_t}\Sigma^{-1}_{x_t}A_{x_t}\right>_{x_t\in\mathcal{X}}$ derived by NN trained SIR model. (e) $\Sigma^{-1}=\left<\Sigma^{-1}_{x_t}\right>_{x_t\in\mathcal{X}}$ derived by NN trained SIR model. (f) Distribution of $r_\epsilon$ obtained by traverse the dynamical space of SIR with different $(S,I)$ and $x_t$ generated from which. (g) Distribution of $\Delta\Gamma_\alpha$ under different $(S,I)$ obtained by traverse the dynamical space. (h) The singular value spectra of $A^T\Sigma^{-1}A$ and $\Sigma^{-1}$ in NN trained SIR model. (i) The frequency of $r_\epsilon$ under different samples $x_t$ in test dataset. (j) Maximum $\Delta\Gamma_\alpha(\epsilon)=0.8685$ when the training period is around 50,000 as $\epsilon=5$. (k) The changing trend of CE under different $\sigma$, the threshold for trend change is around $\sigma=0.01$. (l) The coarse-graining matrix $W$ derived from Section \ref{Coarse-graining} and Appendix E.2 use $A^T\Sigma^{-1}A=\left<A^T_{x_t}\Sigma^{-1}_{x_t}A_{x_t}\right>_{x_t\in\mathcal{X}}$ and $\Sigma^{-1}=\left<\Sigma^{-1}_{x_t}\right>_{x_t\in\mathcal{X}}$. (m) The coarse-graining matrix $W_{NIS+}$ derived from NIS+. For the convenience of observation, we have taken absolute values for each dimension in $W$ and $W_{NIS+}$.}
\label{fig:SIR_appendix}
\end{figure}

\subsection{Details of SIR based on NN}\label{NNSIR}
Most systems in reality are unable to obtain precise dynamic models to calculate analytical solutions for CE, as demonstrated in the previous two examples. However, we can train a neural network to obtain approximate dynamics from observed time series data. Our third case is to show the phenomenon of CE obtained by a well-trained neural network (NN) on the training time series data generated by a Susceptible-Infected-Recovered (SIR) dynamical model \cite{Satsuma2004} as shown in Fig.\ref{fig:SIR_appendix}b with the following dynamics:
\begin{eqnarray}\label{SIRmodel}
\begin{aligned}
\begin{cases}  
\frac{\mathrm{d}S}{\mathrm{d}t}=-\beta SI,  \\
\frac{\mathrm{d}I}{\mathrm{d}t}=\beta SI - \gamma I, \\
\frac{\mathrm{d}R}{\mathrm{d}t}= \gamma I,
\end{cases}
\end{aligned}
\end{eqnarray}
as shown in Fig.\ref{fig:SIR_appendix}a, where $S,I,R\in[0,1]$ represent the rate of susceptible, infected, and recovered individuals in a population, $\beta=1$ and $\gamma=0.5$ are parameters for infection and recovery rates.

To generate the time series data of the micro-state, we adopt the same method as in our previous work of NIS+ \cite{Yang2024}. We generate data by converting $\mathrm{d}S/\mathrm{d}t,\mathrm{d}I/\mathrm{d}t$ into $\Delta S/\Delta t,\Delta I/\Delta t$ as $\Delta t = 0.01$ and $(S_{t+\Delta t},I_{t+\Delta t})\approx(S_{t},I_{t})+(\mathrm{d}S_t/\mathrm{d}t,\mathrm{d}I_t/\mathrm{d}t)\Delta t$. Then we duplicate the macro-state $(S_t,I_t)$ as shown in Fig.\ref{fig:SIR_appendix}c and added Gaussian random noise to form the micro-state $x_t$ as
\begin{eqnarray}
\begin{aligned}
\label{eq:sir noise}
    x_t = (S_t,I_t,S_t,I_t)+\boldsymbol{\xi}_{t},
\end{aligned}
\end{eqnarray}
where $\boldsymbol{\xi}_{t} \sim \scriptsize{N}(0,\Sigma)$ and
\begin{eqnarray}
 \Sigma = \sigma^2\left(\begin{matrix}
            1 &  -\frac{1}{2} &0 &0 \\
            -\frac{1}{2} &1 &0 &0 \\
            0 & 0 &1 & -\frac{1}{2} \\
            0 & 0 & -\frac{1}{2} &1  
        \end{matrix}\right).
\end{eqnarray}
as $\sigma^2=0.04$.
\\

By feeding the micro-state data into a forward neural network (NN) called \textbf{Covariance Learner Network}, we can use this model to approximate the micro-dynamics of the SIR model. The model we trained has the following structure:

$\bullet$ \textbf{Input layer}: The NN has $n=4$ input neurons, corresponding to the size of the input vector $x_t$.

$\bullet$ \textbf{Hidden layers}: The network contains two hidden layers. The first hidden layer has $hidden\_size=64$ neurons, followed by a $Leaky ReLU$ activation function. The second hidden layer also has $hidden\_size=64$ neurons, followed by another $Leaky ReLU$ activation. 

$\bullet$ \textbf{Output layer}: The output layer contains two parts. The first part, $f_{\mu}$, outputs a mean vector $\mu$ of size $n$. The second part, $f_L$, outputs the elements of the lower triangular part of the Cholesky decomposition of the covariance matrix, with size $n\times n$.

Formally, given an input, the network applies a series of transformations:
\begin{eqnarray}
\begin{aligned}
h_1&=LeakyRELU(B_1x_t+b_1),\\
h_2&=LeakyRELU(B_2h_1+b_2),\\
\mu&=B_\mu h_2+b_\mu,\\
L_\Sigma &= B_Lh_2+b_L,
\end{aligned}
\end{eqnarray}
where $B$ and $b$ denote trainable weight matrices and biases. The predicted $f(x_t)$ and final covariance matrix are then computed as:
\begin{eqnarray}
\begin{aligned}
f(x_t)&=\mu,\\
\Sigma_{x_t}&=L_\Sigma L_\Sigma^T.
\end{aligned}
\end{eqnarray}
After that, we use a negative logarithmic likelihood loss function (NLL) as
\begin{eqnarray}
\begin{aligned}
NLL=\frac{1}{2}\left[\left(x_{t+1}-f(x_t)\right)^T\Sigma^{-1}_{x_t}\left(x_{t+1}-f(x_t)\right)+\ln\det(\Sigma_{x_t})+n\ln(2\pi)\right]
\end{aligned}
\end{eqnarray}
to train the parameters of the dynamical function $f(x_t)$ and the covariance matrix $\Sigma_{x_t}$ simultaneously.

Since the quantification of our CE framework does not require a pre-setting coarse-grained strategy and the coarse-grained function can be directly obtained from the SVD of $A^T\Sigma^{-1}A$ after setting $r_\epsilon$, we can directly calculate the degree of CE by $\Delta\Gamma_\alpha$ and ignore the learning processes of encoder, decoder and macro-dynamics in NIS+ as mentioned in \cite{Yang2024}. Due to the nonlinearity of the SIR dynamics, we cannot directly use the CE quantification method of linear GIS. Instead, we approximate NN as a linear mapping at different input $x_t$ by Taylor expansion (see the method in Appendix D.1). We can obtain the CE identification result and the final coarse-graining strategy directly by calculating $A_{x_t}=\nabla f(x_t)$ as the Jacobian matrix of the trained NN model at $x_t$. And $\Sigma_{x_t}$ is the covariance matrix that neural networks can directly output. Due to the differences under different $x_t$, we can randomly generate $x_t$ with a uniform distribution on the domain $\mathcal{X}$ of the SIR dynamics and take the average value as $A^T\Sigma^{-1} A\approx \langle A^T_{x_t}\Sigma^{-1}_{x_t} A_{x_t}\rangle_{x_t\in\mathcal{X}}$ and $\Sigma^{-1}\approx \langle\Sigma^{-1}_{x_t}\rangle_{x_t\in\mathcal{X}}$ to calculate $\Delta\Gamma_\alpha$ of the system. Fig.\ref{fig:SIR_appendix}b visualizes the method of computation for $\Delta\Gamma$ on a multivariate Gaussian model.

Using the data in Fig.\ref{fig:SIR_appendix}c, we can perform SVD on matrices $A^T\Sigma^{-1} A$ and $\Sigma^{-1}$ in Fig.\ref{fig:SIR_appendix}d and Fig.\ref{fig:SIR_appendix}e. In Fig.\ref{fig:SIR_appendix}h, we can see the singular value spectrum of matrices $A^T\Sigma^{-1} A$ and $\Sigma^{-1}$ in which the horizontal axis represents the sequence number of singular values arranged in descending order, while the vertical axis represents the magnitude of singular values $s_i,\kappa_i$. Due to the existence of the copy operation, even the simplest NN can recognize only two dimensions with larger singular values as $r_\epsilon=2$ when $\epsilon=5$. If we directly calculate $r_\epsilon$ by $A^T_{x_t}\Sigma^{-1}_{x_t}A_{x_t}$ as $x_t$ is sampled from the test data set and plot $r_\epsilon$ as a matrix related to $S, I$ in Fig.\ref{fig:SIR_appendix}f, we can see that most positions satisfy $r_\epsilon=2$. By calculating the frequency of $r_\epsilon$ under different samples $A^T_{x_t}\Sigma^{-1}_{x_t}A_{x_t}$, from Fig.\ref{fig:SIR_appendix}i we can find that $r_\epsilon=2$ has the highest frequency.

From Fig.\ref{fig:SIR_appendix}h, $A^T\Sigma^{-1}A$ has two larger singular values and $r_\epsilon=2$ as $\epsilon=5$. In model training, we get the value of Vague CE as $\Delta\Gamma_\alpha(\epsilon)=0.8685$ when the training period is 50,000, which is shown in Fig.\ref{fig:SIR_appendix}j. We can also directly calculate $\Delta\Gamma_\alpha(\epsilon)$ by $A^T_{x_t}\Sigma^{-1}_{x_t}A_{x_t}$ as $x_t$ is sampled from the test data set and plot $\Delta\Gamma_\alpha$ as a matrix related to $S, I$ in Fig.\ref{fig:SIR_appendix}g. Comparing Fig.\ref{fig:SIR_appendix}f and Fig.\ref{fig:SIR_appendix}g, we can see that the most stable region of $r_\epsilon$ and the region with the highest $\Delta\Gamma_\alpha(\epsilon)$ are almost identical as $(S,I)$ roughly located within a circle with a radius of 0.5 and a center of $(0,0)$. 

We can also compare the results of our framework and NIS+ mentioned in \cite{Yang2024}. From Fig.\ref{fig:SIR_appendix}l and Fig.\ref{fig:SIR_appendix}m, we can see that the coarse-graining matrix $W$ obtained through the singular vectors and the Jacobian matrix $W_{NIS+}$ of the NIS+ coarse-graining encoder are similar. For the convenience of observation, we have taken absolute values for each dimension in $W$ and $W_{NIS+}$. Both coarse-graining methods indicate that the first and the third micro-state dimensions mainly influence the first macro-state dimension, while the second and fourth micro-state dimensions mainly influence the second macro-state dimension. Both the values of $W$ and $W_{NIS+}$ are consistent with the copy method (Fig.\ref{fig:Known_model-appendix}a), which is used during data generation. In addition, we can test the changing trend of CE under different $\sigma$. From Fig.\ref{fig:SIR_appendix}k, When $\sigma<0.01$, $\Delta\Gamma_\alpha(\epsilon)$ is positively correlated with $\sigma$, and when $\sigma>0.01$, the two are negatively correlated. We can see that the turning point for CE is around $\sigma=0.01$, which is consistent with the value obtained by NIS+ \cite{Yang2024}. 

However, our framework addresses the decline in training efficiency associated with encoder training in NIS+, leading to a more streamlined and effective approach to model training and CE detection.

\section{Methods}

\subsection{System and observation noises}
The noise $\eta\sim\mathcal{N}(0,\Sigma)$ in our model $y=Ax+\eta$ can be decomposed into system noise $e$ and observation noise $\xi$. For the GIS like $\mathbf{y}=A\mathbf{x}+e,e\sim\mathcal{N}(0,E)$ is pure noise inherent in the system, $\mathbf{y}\in \mathcal{R}^n$ follows a normal distribution about $\mathbf{x}\in \mathcal{R}^n$, so $\mathbf{y}\sim \mathcal{N}(A\mathbf{x},E)$, $A\in\mathcal{R}^{n\times n}$, $E\in\mathcal{R}^{n\times n}$. We can add observation noise as
    \begin{eqnarray}\label{p(xtp1|doxt)}
		x=\mathbf{x}+\xi_\mathbf{x}, \xi_\mathbf{x}\sim\mathcal{N}(0,\Xi_\mathbf{x})\\
        y=\mathbf{y}+\xi_\mathbf{y}, \xi_\mathbf{y}\sim\mathcal{N}(0,\Xi_\mathbf{y})
	\end{eqnarray}
   Since $\mathbf{y}=A\mathbf{x}+e$, $y=A(x-\xi_\mathbf{x})+e+\xi_\mathbf{y}\sim\mathcal{N}(Ax_t,\Xi_\mathbf{y}+A\Xi_\mathbf{x}A^{'}+E)=\mathcal{N}(Ax_t+a_0,\Sigma)$, in which 
   \begin{eqnarray}\label{SigmaXiE}  \Sigma=\Xi_\mathbf{y}+A\Xi_\mathbf{x}A^{'}+E\in\mathcal{R}^{n\times n}
	\end{eqnarray}
     is the combination of covariance matrixes of system and observation noises and
     \begin{eqnarray}\label{epsxie}
		\eta=-A\xi_\mathbf{x}+e+\xi_\mathbf{y}\in\mathcal{R}^n
	\end{eqnarray}
     is the combination of two types of noise. In real data, $\Xi$ is more common than $E$ and it's difficult to distinguish between the two directly. The SIR model in this article only has observation noise $\xi$ and the first two cases only have system noises $e$, all included in $\eta$.
     
\subsection{Multi-steps}
In the main text, we mainly focus on one-step Gaussian iterative systems. And if we are facing multi-step iterations, our framework can also be used for analysis.

\subsubsection{Two steps}
We can start analyzing in two steps first. For the linear stochastic iteration system like $x_{t+1}=A_1x_{t}+A_2x_{t-1}+\varepsilon_t,\varepsilon_t\sim\mathcal{N}(0,\Sigma_n)$, $A_1,A_2,\Sigma_n\in\mathcal{R}^{n\times n}$, $x_t,x_{t+1}\in\mathcal{R}^{n}$, can be transformed into a two-step from as
\begin{eqnarray}
		\begin{pmatrix}
  x_{t+2} \\
  x_{t+1}
\end{pmatrix}=\begin{pmatrix}
  A_1^2+A_2 & A_1A_2 \\
  A_1 & A_2
\end{pmatrix}\begin{pmatrix}
  x_{t} \\
  x_{t-1}
\end{pmatrix}+\begin{pmatrix}
  A_1\varepsilon_{t}+\varepsilon_{t+1}\\
  \varepsilon_{t} 
\end{pmatrix}
	\end{eqnarray}
 in which
\begin{eqnarray}
		E_t=\begin{pmatrix}
  A_1\varepsilon_{t}+\varepsilon_{t+1}\\
  \varepsilon_{t} 
\end{pmatrix}\sim\mathcal{N}\left(O,\begin{pmatrix}
  A_1\Sigma_n A_1^{T}+\Sigma_n & O\\
  O & \Sigma_n
\end{pmatrix}\right)
\end{eqnarray}
so, define
\begin{eqnarray}
		A=\begin{pmatrix}
  A_1^2+A_2 & A_1A_2 \\
  A_1 & A_2
\end{pmatrix},
 \Sigma=\begin{pmatrix}
  A_1\Sigma_n A_1^{T}+\Sigma_n & O\\
  O & \Sigma_n
\end{pmatrix}
\end{eqnarray}
as $A,\Sigma \in \mathcal{R}^{2n}$.

\begin{proof}

For the linear stochastic iteration system like $x_{t+1}=A_1x_{t}+A_2x_{t-1}+\varepsilon_t,\varepsilon_t\sim\mathcal{N}(0,\Sigma_n)$, $A_1,A_2,\Sigma_n\in\mathcal{R}^{n\times n}$, $x_t,x_{t+1}\in\mathcal{R}^{n}$, can be transformed into a two-step from as
\begin{eqnarray}
		\begin{pmatrix}
  x_{t+1} \\
  x_{t}
\end{pmatrix}=\begin{pmatrix}
  A_1 & A_2 \\
  I & O
\end{pmatrix}\begin{pmatrix}
  x_{t} \\
  x_{t-1}
\end{pmatrix}+\begin{pmatrix}
  \varepsilon_{t}\\
  O 
\end{pmatrix}
	\end{eqnarray}
so, define
\begin{eqnarray}
		A=\begin{pmatrix}
  A_1 & A_2 \\
  I & O
\end{pmatrix}^2,
 E_t=\begin{pmatrix}
  A_1 & A_2 \\
  I & O
\end{pmatrix}
\begin{pmatrix}
  \epsilon_t \\
  O 
\end{pmatrix}
+\begin{pmatrix}
  \varepsilon_{t+1} \\
  O 
\end{pmatrix}
\end{eqnarray}
and
\begin{eqnarray}
X_t=\begin{pmatrix}
    x_t\\
    x_{t-1}
\end{pmatrix},
X_{t+\tau}=\begin{pmatrix}
    x_{t+2}\\
    x_{t+1}.
\end{pmatrix}
\end{eqnarray}
This method is equal to direct calculate
\begin{eqnarray}
\begin{aligned}
x_{t+2}&=A_1x_{t+1}+A_2x_{t}+\varepsilon_{t+1}\\&=A_1(A_1x_{t}+A_2x_{t-1}+\varepsilon_t)+A_2x_{t}+\varepsilon_{t+1}\\&=(A_1^2+A_2)x_t+A_1A_2x_{t-1}+A_1\varepsilon_t+\varepsilon_{t+1}.
\end{aligned}
\end{eqnarray}
\end{proof}
\subsubsection{$\tau$ steps}
We can use mathematical induction to extend two-step GIS to multi-step. For the linear stochastic iteration system like $x_{t+1}=A_1x_{t}+A_2x_{t-1}+\cdots+A_\tau x_{t-\tau+1}+\varepsilon_t,\varepsilon_t\sim\mathcal{N}(0,\Sigma_n)$, $A_1,A_2,\cdots,A_\tau,\Sigma_n\in\mathcal{R}^{n\times n}$, $x_t,x_{t+1}\cdots,x_\tau\in\mathcal{R}^{n}$, can be transformed into a $\tau$-step from as
\begin{eqnarray}
		\begin{pmatrix}
  x_{t+1}\\
x_t \\
  \cdots \\
  x_{t-\tau+2}
\end{pmatrix}=\begin{pmatrix}
  A_1 & A_2 & \cdots & A_\tau-1 & A_\tau \\
  I & O & \cdots & O & O\\
  O & I & \cdots & O & O\\
  \cdots & \cdots &\cdots &\cdots & \cdots\\
  O & O & \cdots & I & O
\end{pmatrix}\begin{pmatrix}
  x_{t}\\
  x_{t-1}  \\
  \cdots \\
  x_{t-\tau+1}
\end{pmatrix}+\begin{pmatrix}
  \varepsilon_{t}\\
  O\\
  \cdots\\
  O 
\end{pmatrix}.
	\end{eqnarray}
so, define
\begin{eqnarray}
		A=\begin{pmatrix}
  A_1 & A_2 & \cdots & A_\tau-1 & A_\tau \\
  I & O & \cdots & O & O\\
  O & I & \cdots & O & O\\
  \cdots & \cdots &\cdots &\cdots & \cdots\\
  O & O & \cdots & I & O
\end{pmatrix}^\tau,
 E_t=\sum_{i=1}^{\tau}\begin{pmatrix}
  A_1 & A_2 & \cdots & A_\tau-1 & A_\tau \\
  I & O & \cdots & O & O\\
  O & I & \cdots & O & O\\
  \cdots & \cdots &\cdots &\cdots & \cdots\\
  O & O & \cdots & I & O
\end{pmatrix}^{i-1}
\begin{pmatrix}
  \varepsilon_{t-i+\tau}\\
  O\\
  \cdots\\
  O 
\end{pmatrix}.
\end{eqnarray}
and variables
\begin{eqnarray}
X_t=\begin{pmatrix}
    x_t\\
    x_{t-1}\\
    \cdots\\
    x_{t-\tau+1}
\end{pmatrix},
X_{t+\tau}=\begin{pmatrix}
    x_{t+\tau}\\
    x_{t+\tau-1}\\
    \cdots\\
    x_{t+1}
\end{pmatrix}.
\end{eqnarray}
Our covariance matrix can be calculated as
\begin{eqnarray}
\Sigma=\begin{pmatrix}
    \Sigma_n^{(1)} & O & \cdots & O\\
    O & \Sigma_n^{(2)} & \cdots & O\\
    \cdots& \cdots & \cdots &\cdots\\
    O & O & \cdots & \Sigma_n^{(\tau)}
\end{pmatrix}, \Sigma_n^{(k)}=\sum_{i=0}^{\tau-k}\left[\left(\prod_{j=0}^iA_j\right)\Sigma_n\left(\prod_{j=0}^iA_j\right)^T\right].
\end{eqnarray}
Therefore, for multi-step linear GIS, we can convert it into a standardized form as
\begin{eqnarray}
X_{t+\tau}=AX_t+E_t, E_t\sim\mathcal{N}(0,\Sigma).
\end{eqnarray}

\subsubsection{Nonlinear}
For the stochastic iteration system like 
\begin{eqnarray}
\begin{aligned}
x_{t+1}&=f(x_{t},g(x_{t-1},y_{t-1})+\varepsilon_{2,t-1})+\varepsilon_{1,t},\\
y_{t+1}&=g(f(x_{t-1},y_{t-1})+\varepsilon_{1,t-1},y_t)+\varepsilon_{2,t}\\
x_{t+2}&=f(f(x_{t},g(x_{t-1},y_{t-1})+\varepsilon_{2,t-1})+\varepsilon_{1,t},g(x_{t},y_{t})+\varepsilon_{2,t})+\varepsilon_{1,t+1},\\
y_{t+2}&=g(f(x_{t},y_{t})+\varepsilon_{1,t},g(f(x_{t-1},y_{t-1})+\varepsilon_{1,t-1},y_t)+\varepsilon_{2,t})+\varepsilon_{2,t+1}
\end{aligned}
\end{eqnarray}
$\varepsilon_{1,t},\varepsilon_{2,t}\sim\mathcal{N}(0,\Sigma_n),\Sigma_n\in\mathcal{R}^{n\times n}$, $x_t,x_{t+1}\in\mathcal{R}^{n}$, can be transformed into a two-step from as
\begin{eqnarray}\label{DegeneracyDeterminismEI}
\begin{aligned}
		\nabla{\begin{pmatrix}
  x_{t+2} \\
  y_{t+2}\\
  x_{t+1}\\
  y_{t+1}\\
\end{pmatrix}}&\approx\begin{pmatrix}
  \frac{\partial x_{t+2}}{\partial x_{t}} &  \frac{\partial x_{t+2}}{\partial y_{t}} &\frac{\partial x_{t+2}}{\partial x_{t-1}}&\frac{\partial x_{t+2}}{\partial y_{t-1}}\\
  \frac{\partial y_{t+2}}{\partial x_{t}} &  \frac{\partial y_{t+2}}{\partial y_{t}} &\frac{\partial y_{t+2}}{\partial x_{t-1}}&\frac{\partial y_{t+2}}{\partial y_{t-1}}\\
  \frac{\partial x_{t+1}}{\partial x_{t}} &  \frac{\partial x_{t+1}}{\partial y_{t}} &\frac{\partial x_{t+1}}{\partial x_{t-1}}&\frac{\partial x_{t+1}}{\partial y_{t-1}}\\
  \frac{\partial y_{t+1}}{\partial x_{t}} &  \frac{\partial y_{t+1}}{\partial y_{t}} &\frac{\partial y_{t+1}}{\partial x_{t-1}}&\frac{\partial y_{t+1}}{\partial y_{t-1}}
\end{pmatrix}\\
&=\begin{pmatrix}
  f_1f_1+f_2g_1 & f_2g_2  &f_1f_2g_1&f_1f_2g_2\\
  g_1f_1 &  g_1f_2+g_2g_2 &g_2g_1f_1&g_2g_1f_2\\
  f_1 &  0 & f_2g_1&f_2g_2\\
  0 &  g_2 &g_1f_1&g_1f_2
\end{pmatrix}\\
&=A
\end{aligned}
	\end{eqnarray}
so
\begin{eqnarray}
\begin{aligned}
		\begin{pmatrix}
  x_{t+2} \\
  y_{t+2}\\
  x_{t+1}\\
  y_{t+1}\\
\end{pmatrix}&\approx\begin{pmatrix}
  f_1f_1+f_2g_1 & f_2g_2  &f_1f_2g_1&f_1f_2g_2\\
  g_1f_1 &  g_1f_2+g_2g_2 &g_2g_1f_1&g_2g_1f_2\\
  f_1 &  0 & f_2g_1&f_2g_2\\
  0 &  g_2 &g_1f_1&g_1f_2
\end{pmatrix}\begin{pmatrix}
  x_{t} \\
  y_{t}\\
  x_{t-1}\\
  y_{t-1}\\
\end{pmatrix}+\begin{pmatrix}
  f_1f_2\varepsilon_{2,t-1}+f_1\varepsilon_{1,t}+f_2\varepsilon_{2,t}+\varepsilon_{1,t+1} \\
  g_2g_1\varepsilon_{1,t-1}+g_1\varepsilon_{1,t}+g_2\varepsilon_{2,t}+\varepsilon_{2,t+1}\\
  f_2\varepsilon_{2,t-1}+\varepsilon_{1,t}\\
  g_1\varepsilon_{1,t-1}+\varepsilon_{2,t}\\
\end{pmatrix}\\
&=A\begin{pmatrix}
  x_{t} \\
  y_{t}\\
  x_{t-1}\\
  y_{t-1}\\
\end{pmatrix}+\varepsilon_{t-1}.
\end{aligned}
\end{eqnarray}

\subsection{Stochastic differential equation}
Stochastic differential equations, as the name suggests, are differential equations within the framework of stochastic analysis. Its general form is
\begin{eqnarray}\label{p(xtp1|doxt)}
\begin{aligned}
\frac{dx_t}{dt}=f(x_t)+\Sigma^\frac{1}{2}\varepsilon_t, \varepsilon_t=\frac{dW_t}{dt}\approx\frac{1}{\sqrt{dt}}I,
\end{aligned}
\end{eqnarray}
which is equal to
\begin{eqnarray}\label{p(xtp1|doxt)}
\begin{aligned}
dx_t=f(x_t)dt+\Sigma^\frac{1}{2}dW_t, \sqrt{dt}I\approx dW_t\sim\mathcal{N}(0,dtI),
\end{aligned}
\end{eqnarray}
and
\begin{eqnarray}\label{p(xtp1|doxt)}
\begin{aligned}
p(x_t+dx_t|x_t)=\mathcal{N}(x_t+dtf(x_t),dt\Sigma).
\end{aligned}
\end{eqnarray}

Stochastic differential equations can also be linearized locally as $f(x)=A_0+Ax+o(x)$, then we can get
\begin{eqnarray}\label{p(xtp1|doxt)}
\begin{aligned}
p(x_t+dx_t|x_t)=\mathcal{N}(x_t+dtAx_t,dt\Sigma)
\end{aligned}
\end{eqnarray}
which makes these systems interesting and plays an important role in the classification of fixed points in nonlinear systems.

For nonlinear situations, we can directly calculate their time-dependent reversible information as
\begin{eqnarray}
		\begin{aligned}
		\ln\Gamma_\alpha(\tau)&=\ln\left(\frac{2\pi}{\alpha}\right)^\frac{n}{2}+\ln{\rm det}\left((I+\tau\nabla f(x_t))^T(I+\tau\nabla f(x_t))\right)^{\frac{1}{2}-\frac{\alpha}{4}}+\ln{\rm det}\left(\frac{1}{\tau}\Sigma^{-1}\right)^\frac{1}{2}\\
        &=\ln\left(\frac{2\pi}{\alpha}\right)^\frac{n}{2}+\ln\frac{{\rm det}\left(I+\tau\nabla f(x_t)\right)^{1-\frac{\alpha}{2}}}{\det(\tau\Sigma)^\frac{1}{2}}\\
        &=\ln\left(\frac{2\pi}{\alpha}\right)^\frac{n}{2}+\ln\frac{{\rm det}\left(I+\tau\nabla f(x_t)\right)^{1-\frac{\alpha}{2}}}{\det(\tau\Sigma)^\frac{1}{2}}.
   \end{aligned}
\end{eqnarray} 
This form can be extended to a linear form and time $\tau$ can be extracted through approximation
\begin{eqnarray}
		\begin{aligned}
		\ln\Gamma_\alpha(\tau)&=\ln\left(\frac{2\pi}{\alpha}\right)^\frac{n}{2}+\ln{\rm det}\left((I+A\tau)^T(I+A\tau)\right)^{\frac{1}{2}-\frac{\alpha}{4}}-
        \ln{\rm det}\left(\Sigma\tau\right)^\frac{1}{2}\\
        _{(\tau\to 0)}&\approx\frac{n}{2}\ln\left(\frac{2\pi}{\alpha}\right)+{\left(\frac{1}{2}-\frac{\alpha}{4}\right)}\ln{\rm det}\left((e^{A\tau})^Te^{A\tau}\right)-\frac{1}{2}\ln{\rm det}\left(\Sigma\tau\right)\\
        &=\frac{n}{2}\ln\left(\frac{2\pi}{\alpha}\right)+{\left(\frac{1}{2}-\frac{\alpha}{4}\right)}{\rm tr}(A+A^T)\tau-\frac{1}{2}\ln{\rm det}\left(\Sigma\tau\right)\\
        _{\sqrt[n]{\det(\Sigma)}\to\tau\to 0)}&\approx\frac{n}{2}\ln\left(\frac{2\pi}{\alpha}\right)+{\left(\frac{1}{2}-\frac{\alpha}{4}\right)}{\rm tr}(A+A^T)\tau-\frac{1}{2}\tau\ln{\rm det}\left(\Sigma\right)\\
        &=\frac{n}{2}\ln\left(\frac{2\pi}{\alpha}\right)+\left[{\left(\frac{1}{2}-\frac{\alpha}{4}\right)}{\rm tr}(A+A^T)-\frac{1}{2}\ln{\rm det}\left(\Sigma\right)\right]\tau
   \end{aligned}
\end{eqnarray} 
At this point, only differentiation is needed to obtain an approximate dynamical reversibility index for stochastic differential equations that is independent of time as
\begin{eqnarray}
		\begin{aligned}
		\frac{\partial \ln\Gamma_\alpha}{\partial \tau}&\approx{\left(\frac{1}{2}-\frac{\alpha}{4}\right)}{\rm tr}(A+A^T)-\frac{1}{2}\ln{\rm det}\left(\Sigma\right).
   \end{aligned}
\end{eqnarray} 
This indicator can also be transformed into a form determined by eigenvalues and singular values as
\begin{eqnarray}
		\begin{aligned}
		\frac{\partial \ln\Gamma_\alpha}{\partial \tau}&\approx{\left(\frac{1}{2}-\frac{\alpha}{4}\right)}\sum_{i=1}^n\lambda_i+\frac{1}{2}\sum_{i=1}^n\ln\kappa_i,
   \end{aligned}
\end{eqnarray} 
in which 
$\Lambda={\rm diag}(\lambda_1,\cdots,\lambda_n),K={\rm diag}(\kappa_1,\cdots,\kappa_n)$, $A+A^T=U\Lambda U^T, \Sigma^{-1}=VKV^T$.

\end{document}